\documentclass[a4paper, 11pt]{article}
\usepackage[affil-it]{authblk}
\usepackage{ifpdf}
\usepackage{amsthm,amsfonts,amssymb,amsmath,euscript,comment,mathtools}
\usepackage{mathabx}
\usepackage[utf8]{inputenc}
\usepackage[T1]{fontenc}
\usepackage{graphicx}
\usepackage{color}
\usepackage{textcomp}
\usepackage{bbm}
\usepackage{slashed}
\usepackage[colorlinks=true,linkcolor=black,citecolor=blue]{hyperref}
\usepackage{mathrsfs}
\usepackage{dsfont}
\usepackage{geometry}
\usepackage[all]{xy}
\usepackage[dvipsnames]{pstricks}
\usepackage{epsfig}
\usepackage[nottoc]{tocbibind}
\allowdisplaybreaks[3]
\ifpdf
  \renewcommand{\c}{\cdot}
\fi

\usepackage{bm}
\usepackage{cancel}

\DeclarePairedDelimiter{\abs}{\lvert}{\rvert}

\def\bF{\,^{(\mathbf{F})} \hspace{-2.2pt}\b}
\def\bbF{\,^{(\mathbf{F})} \hspace{-2.2pt}\bb}
\def\rhoF{\,^{(\mathbf{F})} \hspace{-2.2pt}\rho}

\def\rhodF{\,^{(\mathbf{F})} \hspace{-2.2pt}\rhod}

\def\a{{\alpha}}

\def\b{{\beta}}

\def\ga{\gamma}
\def\Ga{\Gamma}
\def\de{\delta}
\def\De{\Delta}
\def\ep{\epsilon}
\def\ka{\kappa}

\def\Si{\Sigma}
\def\om{\omega}

\def\vphi{\varphi}

\def\th{\theta}

\def\ze{\zeta}
\def\ka{\kappa}
\def\nab{\nabla}

\def\Up{\Upsilon}

\def\ombc{\widecheck{\omb}}

\def\rhoc{\widecheck{\rho}}
\def\rhoFc{\widecheck{\rhoF}}

\def\kac{\widecheck \ka}
\def\kabc{\widecheck{\underline{\ka}}}

\def\ombc{\underline{\widecheck{\omega}}}

\def\trchc{\widecheck{\tr\chi}}
\def\trchbc{\widecheck{\tr\chib}}
\def\yc{\widecheck{y}}
\def\zc{\widecheck{z}}

\def\pr{{\partial}}
\def\les{\lesssim}
\def\c{\cdot}

\def\trch{{\mbox tr}\, \chi}
\def\chih{{\widehat \chi}}
\def\chib{{\underline \chi}}
\def\chibh{{\underline{\chih}}}

\def\etab{{\underline \eta}}
\def\omb{{\underline{\om}}}
\def\bb{{\underline{\b}}}
\def\aa{\protect\underline{\a}}
\def\xib{{\underline \xi}}

\def\W{\textbf{W}}

\def\qf{\qk}

\def\F{\mathbf{F}}

\renewcommand{\div}{\sdiv}

\def\lap{\De}

\def\DDd{{\slashed{\mathcal{D}}}}

\newcommand{\nabb}{{\bf \nab} \mkern-13mu /\,}

\def\hot{\widehat{\otimes}}
\def\rhod{\,\dual\hspace{-2pt}\rho}

\def\Mext{\ext}

\def\ext{{^{(ext)}\MM}}

\def\Gag{\Ga_g}
\def\Gab{\Ga_b}

\def\D{\mathbf{D}}

\def\dkb{\slashed{\dk}}

\renewcommand{\c}{\cdot}

\def\hot{\widehat{\otimes}}

\def\dec{\delta}

\DeclareMathOperator{\sdiv}{div}

\def\D{\mathbf{D}}

\def\th{\theta}
\def\qk{\mathfrak{q}}

\def\pf{\mathfrak{p}}

\def\hk{\mathfrak{h}}

\def\dk{\mathfrak{d}}

\def\ds{\displaystyle}

\def\ombc{\widecheck{\omb}}

\def\dee{\de_{extra}}

\def\rhoc{\widecheck{\rho}}

\def\kac{\widecheck \ka}
\def\kabc{\widecheck{\underline{\ka}}}

\def\ombc{\underline{\widecheck{\omega}}}

\def\trchc{\widecheck{\tr\chi}}
\def\trchbc{\widecheck{\tr\chib}}
\def\yc{\widecheck{y}}
\def\zc{\widecheck{z}}

\def\muc{\widecheck{\mu}}
\def\rhod{{^*\rho}}

\newcommand{\ov}{\overline}

\newtheorem{thm}{Theorem}[section]

\newtheorem{theorem}[thm]{Theorem}

\newtheorem{remark}[thm]{Remark}
\newtheorem{lemma}[thm]{Lemma}
\newtheorem*{bootstrap}{Bootstrap Assumptions}
\newtheorem*{dominance}{Dominance Condition}
\newtheorem{proposition}[thm]{Proposition}

\newtheorem{corollary}[thm]{Corollary}
\newtheorem{definition}[thm]{Definition}

\def\om{\omega}

\def\les{\lesssim}
\def\MM{\mathcal{M}}

\DeclareMathOperator{\tr}{tr}

\renewcommand{\th}{\theta}
\renewcommand{\a}{\alpha}
\renewcommand{\b}{\beta}
\newcommand{\g}{{\bf g}}

\renewcommand{\aa}{\underline{\a}}

\def\Jp{J^{(p)}}

\def\Mext{\ext}

\newcommand{\Mint}{\,{}^{(int)}\mathcal{M}}

\def\Sitop{\,^{(top)}\Si}

\numberwithin{equation}{section}

\def\MM{{\mathcal M}}

\def\TT{{\mathcal T}}

\def\dual{{^*}}

\def\trch{\tr\chi}
\def\trchb{\tr\chib}

\def\ka{\kappa}
\def\rhoc{\widecheck{\rho}}
\def\kab{\underline{\ka}}
\def\kac{\widecheck{\ka}}
\def\kabc{\widecheck{\kab}}

\def\Si{\Sigma}

\def\dds{\slashed{d}^*}
\def\ddd{\slashed{d}}

\def\hch{\widehat{\chi}}
\def\hchb{\widehat{\chib}}

\def\trchc{\widecheck{\trch}}
\def\trchbc{\widecheck{\trchb}}

\def\ombc{\widecheck{\omb}}

\def\Kc{\widecheck{K}}
\def\ze{\zeta}
\def\nab{\nabla}

\def\De{\Delta}

\DeclareMathOperator{\curl}{curl}

\def\D{{\bf D}}

\def\de{\delta}

\def\ombc{\widecheck{\omb}}

\def\muc{\widecheck{\mu}}

\def\rhoc{\widecheck{\rho}}

\def\dk{\mathfrak{d}}

\makeatother
\allowdisplaybreaks

\usepackage[sortcites,
backend=biber,
date=year,
style=numeric-comp,
doi=false,
isbn=false,
url=false,
maxcitenames=4,
maxbibnames=4,
eprint=true]{biblatex}

\DeclareNameAlias{author}{family-given}
\AtEveryBibitem{
  \clearfield{month}
  \clearfield{url}
  \clearfield{urlyear}
  \clearfield{urlmonth} 
  \ifentrytype{online}{
      \clearfield{year}}
    {
      \clearfield{eprint}}
  }
  \renewbibmacro*{volume+number+eid}{%
    \printfield{volume}%
    \setunit*{\addnbspace}
    \printfield{number}%
    \setunit{\addcomma\space}%
    \printfield{eid}}
  \DeclareFieldFormat[article]{volume}{\textbf{#1}}
  \DeclareFieldFormat[article]{number}{\mkbibparens{#1}}
  \DeclareFieldFormat*{title}{\textit{#1}}
  \DeclareFieldFormat{journaltitle}{#1\isdot}
  \renewbibmacro{in:}{}
\begin{document} 
  \title{Einstein-Maxwell Equations \\ on Mass-Centered GCM Hypersurfaces}

  \author[1]{Allen Juntao Fang\footnote{allen.juntao.fang@uni-muenster.de}}
  \author[2]{Elena Giorgi\footnote{elena.giorgi@columbia.edu}}
  \author[3]{Jingbo Wan\footnote{jingbo.wan@sorbonne-universite.fr}}

  \affil[1]{\small Mathematics M\"unster, Universit\"at M\"unster}
  \affil[2]{\small Department of Mathematics, Columbia University}
  \affil[3]{\small Laboratoire Jacques-Louis Lions de Sorbonne Universit\'e}

  \maketitle

  \begin{abstract}

    The resolution of the nonlinear stability of black holes as solutions to the Einstein equations relies crucially on imposing the right geometric gauge conditions. In the vacuum case, the use of Generally Covariant Modulated (GCM) spheres and hypersurfaces has been successful in the proof of stability for slowly rotating Kerr spacetime \cite{klainermanKerrStabilitySmall2023}. For the charged setting, our companion paper \cite{fangMassCenteredGCMFramework2025} introduced an alternative mass-centered GCM framework, adapted to the additional difficulties of the Einstein–Maxwell system. 
    
    In this work, we solve the Einstein-Maxwell equations on such a mass-centered spacelike GCM hypersurface, which is equivalent to solving the constraint equations there. We control all geometric quantities of the solution in terms of some seed data, corresponding to the gauge-invariant fields describing coupled gravitational–electromagnetic radiation in perturbations of Reissner–Nordström or Kerr–Newman, first identified in \cite{giorgiElectromagneticgravitationalPerturbationsKerr2022} and expected to be governed by favorable hyperbolic equations. This provides the first step toward controlling gauge-dependent quantities in the nonlinear stability analysis of the Reissner–Nordström and Kerr–Newman families.

  \end{abstract}

  \tableofcontents

  \section{Introduction}

  In this paper, we solve the Einstein–Maxwell equations on a mass-centered General Covariant Modulated (GCM) hypersurface. The mass-centered GCM framework, introduced in our companion work \cite{fangMassCenteredGCMFramework2025}, is a refinement of the original GCM setup, which played a central role in the proof of the nonlinear stability of Kerr spacetimes \cite{klainermanKerrStabilitySmall2023, giorgiWaveEquationsEstimates2024}. Since the hypersurface is spacelike, solving the restriction of the Einstein–Maxwell equations there amounts to solving the associated constraint equations. As seed data for these constraints, we prescribe the gauge-invariant quantities encoding the electromagnetic and gravitational radiation.

  \subsection{The GCM framework in the nonlinear stability of black holes}

The nonlinear stability of black hole solutions to the Einstein equations has been a central topic of investigation in mathematical General Relativity, with significant progress achieved in recent years \cite{hintzGlobalNonLinearStability2018,klainermanGlobalNonlinearStability2020,dafermosNonlinearStabilitySchwarzschild2021,klainermanKerrStabilitySmall2023}. As in most stability problems for hyperbolic PDEs, the goal is to establish global existence from initial data satisfying a smallness condition, typically via a continuity argument on a bootstrap region of finite size. Owing to the covariant nature of the Einstein equations, the analysis of the dynamical evolution of such data is deeply dependent on the choice of gauge. In other words, a suitable gauge fixing is essential in order to implement the continuity argument successfully.

  In the recent proof of the nonlinear stability of slowly rotating Kerr black holes
  \cite{klainermanKerrStabilitySmall2023,giorgiWaveEquationsEstimates2024,shenConstructionGCMHypersurfaces2023},
  the gauge choice underpinning the continuity argument is based on the construction of Generally Covariant Modulated (GCM) spheres \cite{klainermanConstructionGCMSpheres2022,klainermanEffectiveResultsUniformization2022}. These spheres foliate a spacelike GCM hypersurface \cite{shenConstructionGCMHypersurfaces2023}, on which several key geometric quantities are prescribed to vanish. More precisely, the bootstrap region (see Figure \ref{fig:myfigure}) has future boundary given by the union $\mathcal{A} \cup \Sitop \cup \Si_*$, where $\mathcal{A}, \Sitop$ and $\Si_*$ are spacelike hypersurfaces, with $\Si_*$ foliated by GCM spheres. The spacetime $\MM$ also contains a timelike hypersurface $\TT$ which divides $\MM$ into an exterior region called $\Mext$ and an interior region called $\Mint$.
  The hypersurface $\Sigma_*$ is required to satisfy a dominance condition, ensuring it lies “far away” (in a precise quantitative sense) from the perturbed event horizon, as well as transversality conditions (see Section \ref{sec:geometric-setting-Sigma}). Furthermore, parameters such as the final mass, angular momentum, and axis of rotation of the black hole are fixed on the last sphere $S_*$ of $\Si_*$, where the bootstrap region terminates. As a consequence of the continuity argument, the GCM foliation of $\Si_*$ induces a canonical foliation of null infinity, thereby resolving the supertranslation ambiguity \cite{klainermanCanonicalFoliationNull2024}.
  \begin{figure}[ht]
    \centering
    \includegraphics[width=0.65\textwidth]{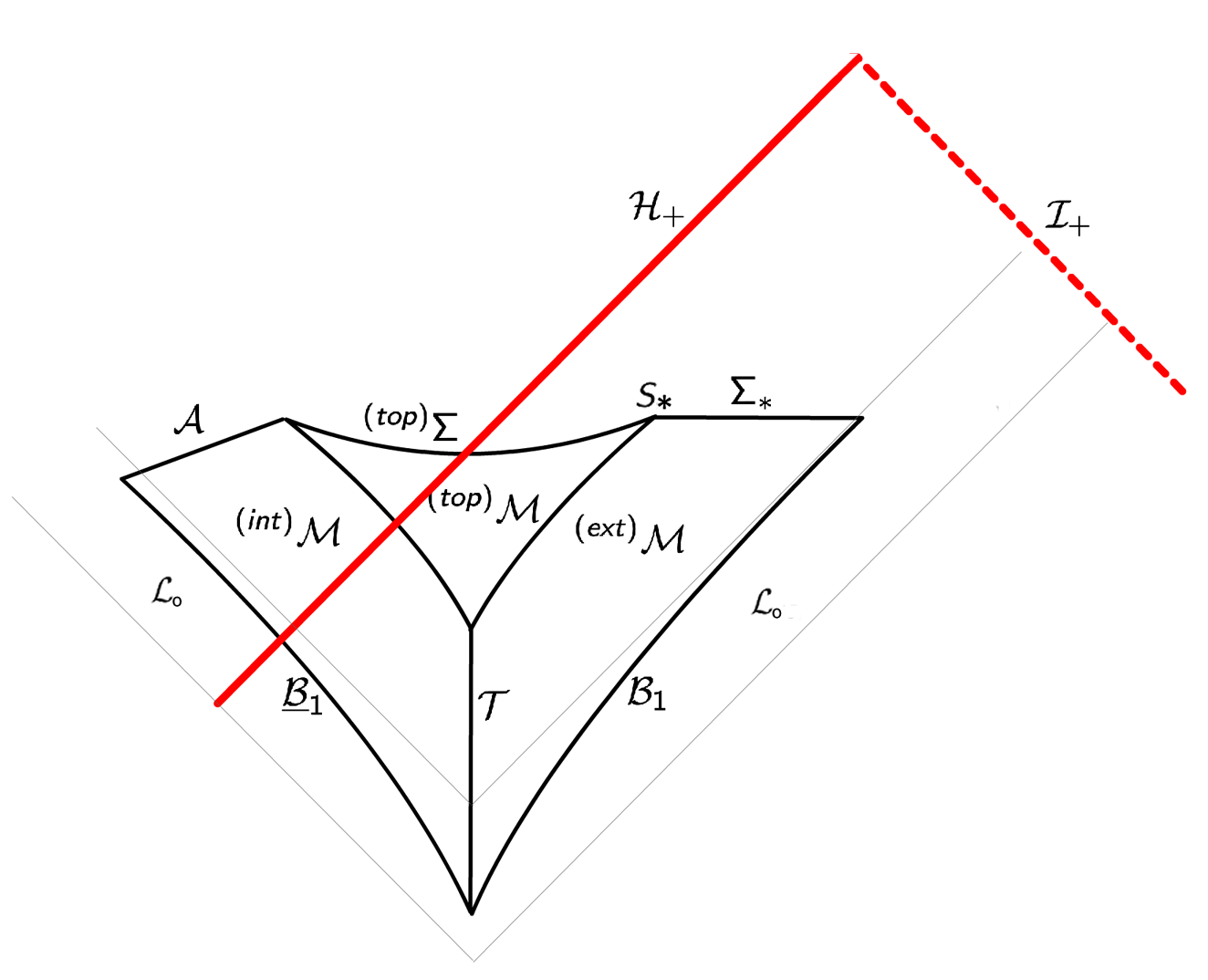} 
    \caption{The GCM admissible spacetime of \cite{klainermanKerrStabilitySmall2023}.}
    \label{fig:myfigure} 
  \end{figure}

  A critical condition imposed on the last sphere $S_*$ is the vanishing of the spacetime’s center of mass, encoded in the $\ell=1$ mode\footnote{Defining $\ell=1$ modes in a canonical way for general spacetimes is a subtle problem, addressed in \cite{klainermanEffectiveResultsUniformization2022}, and used in the construction of GCM spheres and hypersurfaces \cite{shenConstructionGCMHypersurfaces2023}.} of the divergence of the curvature component $\beta$:
  \begin{align*}
    (\div \b)_{\ell=1}=0 \qquad \text{on $S_*$,}
  \end{align*}
  see already Section \ref{sect.LastSlice-Intro} for the definition of the curvature components. Physically, this condition enforces that the last sphere of the bootstrap region shares the same center of mass as the perturbed spacetime, so that its limit as $S_*$ approaches timelike infinity $i^+$ coincides with the final center of mass of the black hole \cite{klainermanCanonicalFoliationNull2024}. 

  In the vacuum case, this choice is possible due to an exceptional structure: the $\ell=1$ mode of $\div \b$ satisfies an outgoing null transport equation involving the spin $+2$ curvature component $\a$, whose $\ell=1$ projection remains regular even under weak decay assumptions \cite{klainermanKerrStabilitySmall2023}. Simultaneously, it satisfies a favorable transport equation along the vector field $\nu$, tangent to $\Si_*$ (see \eqref{def:nu}), where the ingoing derivative dominates. This special structure is crucial to the GCM framework in vacuum, as it permits the controlled propagation of $(\div \b)_{\ell=1}$ both along $\Si_*$ and backwards along outgoing null directions into the exterior region.

  \subsection{Mass-centered GCM hypersurfaces in electrovacuum spacetimes} 

  In the case of charged spacetimes satisfying the Einstein–Maxwell equations (see \eqref{eq:EM}), the exceptional structure present in vacuum breaks down: the $\ell=1$ mode of the perturbations now carries genuine radiation, and the favorable transport properties of $(\div \b)_{\ell=1}$ (or any of its renormalizations) are lost. More precisely, in the charged setting one can construct a renormalization of $(\div \b)_{\ell=1}$ that admits either an acceptable ingoing transport equation or an improved outgoing one, but not both simultaneously. Consequently, the GCM conditions of \cite{klainermanConstructionGCMSpheres2022,klainermanEffectiveResultsUniformization2022,shenConstructionGCMHypersurfaces2023} can no longer be applied directly in the continuity argument for charged black holes, and a modification of the gauge conditions becomes necessary.

  In \cite{fangMassCenteredGCMFramework2025}, we introduced such a modification by imposing that the center of mass — defined as a suitable renormalization of $(\div \b)_{\ell=1}$ with good outgoing transport equation — vanishes on every sphere foliating the GCM hypersurface $\Si$, and not only on its last sphere. This eliminates the need to propagate it along $\Si_*$, while still allowing backward propagation along outgoing null directions in the subsequent step of the proof. Concretely, the \emph{center of mass function} is defined as
  \begin{align*}
    \bm{C}\vcentcolon=\div\b -\frac{Q}{r^2}\div \bF+\frac{2Q}{r^3}\rhoFc,
  \end{align*}
  where $\beta$ and $\bF$ are the null components of the Weyl and electromagnetic fields, and $\rhoFc$ denotes the linearization of an electromagnetic component with respect to its Reissner–Nordström value. We then impose
  \begin{align}\label{eq:condition-C-intro}
    \bm{C}_{\ell=1}=0 \qquad \text{on $\Si_*$.}
  \end{align}
  Here the electric charge $Q$, like the other parameters, is defined on the last sphere $S_*$ by
  \begin{equation*}
    Q \vcentcolon= \frac{1}{4\pi}\int_{S_*} \rhoF= \frac{1}{4\pi}\int_{S_*}\dual\F,
  \end{equation*}
  with the magnetic charge defined analogously (see \eqref{eq:def-e}). A corresponding definition of the angular momentum function $\bm{J}$ is given in \eqref{eq:def-J}, and all parameter functions are recalled in Section \ref{sec:parameterdefinitionsonsstar}.

  This motivates the definition of \emph{mass-centered} GCM spheres and hypersurfaces. The construction of such hypersurfaces is described in detail in our companion paper \cite{fangMassCenteredGCMFramework2025}.

  The choice of renormalization in the definition of $\bm{C}$ is guided by its improved transport along the outgoing null direction, namely (see \eqref{eq:nab-4-C})
  \begin{align*}
    \nab_4(r^5 \bm{C} ) &=
                          -r^5\div\div\a +\text{nonlinear terms}.
  \end{align*}
  Because of the weak decay of curvature components, however,
  $r^5\bm{C}$ fails to decay at infinity, which introduces significant
  subtleties when attempting to integrate this equation. In contrast
  to the vacuum case, where $(\div\b)_{\ell=1}$ enjoys the exceptional
  structure discussed above, the transport of $\bm{C}_{\ell=1}$ along
  $\nu$ (the vector tangent to $\Si_*$, see Proposition
  \ref{lemma:improved-transport-along-Sigma-star}) produces terms that
  decay too weakly to allow integration along the hypersurface. For
  this reason, rather than propagating $\bm{C}_{\ell=1}$ along
  $\Si_*$, we impose the condition \eqref{eq:condition-C-intro}
  directly on every sphere of $\Si_*$. This turns the $\nu$-transport
  equation into an algebraic relation among the perturbation
  quantities, which is then incorporated into the system of equations
  governing the hypersurface.  This reformulation is precisely what
  leads to the notion of \emph{mass-centered} GCM hypersurfaces, and
  it is on such hypersurfaces that we shall now solve the
  Einstein–Maxwell constraint equations.

  \subsection{Einstein-Maxwell equations on hypersurfaces satisfying bootstrap assumptions and the dominance condition}

  In this paper we solve the Einstein–Maxwell equations
  \begin{align}\label{eq:EM}
    \operatorname{Ric}[\g]=2 \F \c \F - \frac 1 2 \g |\F|^2, \qquad  \div \F=0, \qquad d\F=0,
  \end{align}
  where $\operatorname{Ric}[\g]$ is the Ricci curvature of the Lorentzian metric $\g$ and $\F$ is the electromagnetic $2$-form, along a mass-centered GCM hypersurface $\Si_*$ satisfying bootstrap assumptions and the dominance condition, as in the nonlinear stability proof of Kerr by Klainerman–Szeftel \cite{klainermanKerrStabilitySmall2023}. Since $\Si_*$ is spacelike, restricting the Einstein–Maxwell equations to $\Si_*$ is equivalent to solving the associated \emph{constraint equations} on the hypersurface.

  The hypersurface $\Si_*$ is defined by the relation $u+r=c_*$, where $r$ and $u$ play the role of radial and retarded time functions, and $c_*$ is a constant. Both the bootstrap assumptions and the dominance condition are formulated in terms of $r$ and $u$: 
  \begin{enumerate}
  \item[(a)] The \emph{bootstrap assumptions}, compatible with the setup of the nonlinear stability of Kerr–Newman spacetimes, are decay conditions on the perturbation quantities, divided into two classes $\Gag$ and $\Gab$, corresponding to “good’’ and “bad’’ $r$-decay:
    \begin{align}\label{eq:bootstrap-intro}
      \|\Gab\|_{L^\infty(S)} \les \frac{\ep}{r u^{1+\dec}}, \qquad \|\Gag\|_{L^\infty(S)} \les \frac{\ep}{r^2 u^{\frac 1 2 +\dec}},
    \end{align}
    for some $\dec>0$ and a small constant $\ep>0$. Here $\| \cdot \|_{L^\infty(S)}$ denote the $L^\infty$ norm on each sphere $S$ of $\Si_*$.

  \item[(b)] The \emph{dominance condition} is a quantitative relation between $r$ and $u$ on $S_*$, which implies that along $\Si_*$ the $r$-decay dominates as follows:
    \begin{align}\label{eq:dominance-condition}
      \frac{1}{r}\lesssim \frac{\ep_0}{u^{1+\dec}},
    \end{align}
    with $\ep_0$ another smallness parameter satisfying $\ep_0 \ll \ep$ and $\ep^2 \ll \ep_0$. This condition prescribes the position of $S_*$ in the asymptotically flat region of the spacetime.
  \end{enumerate}

  It is well known that the constraint equations form an underdetermined system, and the main difficulty lies in identifying suitable seed data parametrizing their solutions. Classical strategies include the conformal method \cite{lichnerowiczLintegrationEquationsGravitation1944,carlottoGeneralRelativisticConstraint2021} and gluing constructions \cite{corvinoScalarCurvatureDeformation2000,chruscielMappingPropertiesGeneral2003,corvinoAsymptoticsVacuumEinstein2006}.

  In our setting, we show that all geometric quantities on $\Si_*$ describing perturbations of Reissner–Nordström or Kerr–Newman can be controlled in terms of seed data given by the \emph{gauge-invariant radiation fields} for coupled electromagnetic–gravitational perturbations. These quantities,
  \begin{align*}
    \a, \qquad \aa, \qquad \mathfrak{f}, \qquad \underline{\mathfrak{f}}, \qquad \mathfrak{b}, \qquad \underline{\mathfrak{b}},
  \end{align*}
  were identified in \cite{giorgiElectromagneticgravitationalPerturbationsKerr2022} and shown to satisfy a coupled system of Teukolsky equations. The first-order derived quantities of $\mathfrak{f}, \underline{\mathfrak{f}}, \mathfrak{b}, \underline{\mathfrak{b}}$, denoted by
  \begin{align*}
    \mathfrak{q}^{\F}, \qquad \underline{\mathfrak{q}}^{\F}, \qquad \mathfrak{p}, \qquad \underline{\mathfrak{p}},
  \end{align*}
  satisfy generalized Regge–Wheeler equations. These equations are hyperbolic in nature, and thus the control of the radiation fields is expected to be obtainable independently of the rest of the system by energy and decay estimates.
  Results in this direction are already known: boundedness and decay estimates for these quantities have been established in the linearized setting for Reissner–Nordström in the full subextremal range \cite{giorgiLinearStabilityReissner2020a}, for slowly rotating and weakly charged Kerr–Newman \cite{giorgiBoundednessDecayTeukolsky2023}, and for slowly rotating, strongly charged Kerr–Newman spacetimes in axial symmetry \cite{giorgiBoundednessDecayTeukolsky2024}. More recently, decay estimates have also been obtained in the fully nonlinear setting for the Regge–Wheeler system in Reissner–Nordström \cite{wanjingboNonlinearStabilitySubextremal2025}.

  In this work, we go further by showing that, on a mass-centered GCM hypersurface $\Si_*$ satisfying the bootstrap assumptions and dominance condition, \emph{all geometric perturbation quantities} (i.e. the full sets $\Gag$ and $\Gab$) can be expressed and controlled entirely in terms of the free radiation data above. This reduction of the Einstein–Maxwell constraint problem to a parametrization by gauge-invariant radiation fields is the content of our main theorem, stated in the next subsection.

  \subsection{The main theorem and sketch of the proof}
  The application we have in mind is the resolution of the constraint equations on $\Si_*$ as part of the nonlinear stability problem for charged black holes. The gauge-invariant quantities introduced above will serve as the seed data for this underdetermined system. We now state a simplified version of our main result (see Theorem \ref{prop:decayonSigamstarofallquantities} for the precise formulation).

  \begin{theorem}\label{theorem:intro} 

    Let $\Sigma_* \subset \MM$ be a mass-centered GCM hypersurface satisfying the dominance condition on a spacetime $\MM$ solving the Einstein–Maxwell equations, and assume the bootstrap assumptions \eqref{eq:bootstrap-intro} hold for $\Gag$ and $\Gab$. Then all geometric perturbation quantities along $\Si_*$ (i.e. the sets $\Gag$ and $\Gab$) can be expressed and controlled in terms of the gauge-invariant seed data
    $$\text{seed}\vcentcolon=\{\alpha, \aa,  \mathfrak{f}, \mathfrak{b}, \underline{\mathfrak{b}}, \mathfrak{p}, \mathfrak{q}^\F\}. $$

    In particular, 
    \begin{align}\label{eq:theorem-intro}
      \|\Ga_b \|_{L^\infty(S)} \les \frac{\mathscr{F}[\text{seed}]}{r u^{1+\dec}}, \qquad 
      \|\Gag \|_{L^\infty(S)} \les \mathscr{G}[\text{seed}]+\frac{\mathscr{F}[\text{seed}]}{r^2u^{1 +\dec}},
    \end{align}
    where $\mathscr{F}[\text{seed}]$ and $\mathscr{G}[\text{seed}]$ denote suitable flux energy or pointwise norms of the seed data.
  \end{theorem}

  In particular, if one can independently show that
  \begin{align*}
    \mathscr{F}[\text{seed}] \les \ep_0 , \qquad \mathscr{G}[\text{seed}] \les \frac{\ep_0}{r^2 u^{\frac 1 2 +\dec}},
  \end{align*}
  then \eqref{eq:theorem-intro} improves the bootstrap assumptions \eqref{eq:bootstrap-intro} on $\Si_*$.

  \medskip

  The argument follows the proof of Theorem M3 in \cite{klainermanKerrStabilitySmall2023} (see Remark \ref{rem:applications-bhs}) and proceeds in three stages.

  \noindent {\bf Step 1 (unconditional control of $\ell=0$ and most $\ell=1$ modes).}
  Starting from the bootstrap assumptions \eqref{eq:bootstrap-intro} and the dominance condition \eqref{eq:dominance-condition}, we first derive improved decay for the $\ell=0$ and most of the $\ell=1$ spherical modes of the perturbation, independently of the seed data. The $\ell=0$ estimates are obtained from the transport equation for the Hawking mass (see Section \ref{sec:l=0}), while the $\ell=1$ estimates rely on the mass-centered condition imposed on all spheres of $\Si_*$ (see Section \ref{section:estimates-rll=1-Si_*}). These arguments already improve the bootstrap bounds by $\ep_0$, and also derive stronger decay in $r$ and $u$.

  \noindent {\bf Step 2 (control of $\Gab$ in terms of the seed data).}
  We next show that the $\Gab$ quantities can be controlled directly from the prescribed seed data. This is achieved by establishing flux estimates along $\Si_*$ (see Section \ref{sec:controlofthefluxofsomequantitiesonSigmastar}), which yield $L^2$ bounds for $\Gab$ in terms of the seed data. The flux bounds also imply pointwise decay and, crucially, improved estimates for nonlinear interactions of the form $\Gab\c\Gag$ (see Corollary \ref{cor:improved-bounds-nonlinear-=terms}).

  \noindent {\bf Step 3 (recovery of the remaining $\ell=1$ modes and $\Gag$).}
  Finally, with $\Gab$ under control, we recover the remaining $\ell=1$ modes and all $\Gag$ quantities. The nonlinear terms estimated in Step 2 appear on the right-hand side of the transport equation for the angular momentum function $\bm{J}$, which must be integrated along $\Si_*$ (see Section \ref{sec:l=1-2}). This provides control of $\bm{J}$ and hence of the missing $\ell=1$ components in terms of the seed data. Elliptic estimates on the spheres foliating $\Si_*$ are then used to propagate this control to all of $\Gag$, completing the argument (see Section \ref{sec::improvementofdecaybootassonSigmastar}).

  \medskip
  In summary, the constraint equations are resolved by a combination of transport, flux, and elliptic estimates, reducing the system entirely to the specification of the gauge-invariant seed data.

  \begin{remark}\label{rem:applications-bhs}

    Theorem \ref{theorem:intro} is formulated so as to be directly applicable to the setup of the nonlinear stability problem for Kerr–Newman. In the vacuum case of Kerr, the analogous result is Theorem M3 in \cite{klainermanKerrStabilitySmall2023}, which builds on the prior control of gauge-invariant quantities established in Theorem M1 and M2 \cite{giorgiWaveEquationsEstimates2024}. In contrast, in the present setting such control is not assumed: the gauge-invariant fields are treated as seed data for the  induced Einstein-Maxwell equations. This separation ensures that the statement of Theorem~\ref{theorem:intro} remains modular, so that it can later be combined with the independent hyperbolic estimates for the gauge-invariant quantities.
  \end{remark}

  \subsection{Organization of the Paper}

  The paper is organized as follows. 
  \begin{itemize}
  \item Section \ref{sect.LastSlice-Intro} reviews the necessary preliminaries, including the mass-centered GCM hypersurfaces, the Einstein-Maxwell equations, the bootstrap assumptions and dominance condition. It also derives the equations for the renormalized quantities that will be important later.
  \item Section \ref{sec:consequences-BA} collects the first consequences of the bootstrap assumptions and the dominance condition.
  \item  Section \ref{secM3} contains the proof of our main theorem. 
  \end{itemize}

  \paragraph{Acknowledgments} A.J.F. acknowledges  support from the Deutsche
  Forschungsgemeinschaft (DFG, German Research Foundation) through
  Germany’s Excellence Strategy EXC 2044 390685587, Mathematics
  M\"{u}nster: Dynamics–Geometry–Structure, from the Alexander von
  Humboldt Foundation in the framework of the Alexander von Humboldt
  Professorship endowed by the Federal Ministry of Education and
  Research, and from NSF award DMS-2303241. E.G. acknowledges the support of NSF Grants DMS-2306143, DMS-2336118 and of a grant of the Sloan Foundation. J.W. is supported by ERC-2023 AdG 101141855 BlaHSt.

  \section{Preliminaries}\label{sect.LastSlice-Intro}

  We consider a 4-dimensional Lorentzian manifold $(\MM, \g)$ with a 2-form $\F$, satisfying the Einstein-Maxwell equations \eqref{eq:EM}.

  We recall the following standard definition of the null connection
  coefficients and null Weyl curvature components relative to a null
  frame $(e_3, e_4, e_1, e_2)$, where $e_3, e_4$ are respectively
  ingoing and outgoing null directions and $e_1, e_2$ are orthonormal
  frames orthogonal to them:
  \begin{align*}
    \chib_{ab}&=\g(\D_{e_a}e_3, e_b), \qquad \chi_{ab}=\g(\D_{e_a}e_4, e_b), \\
    \xib_a&=\frac 1 2 \g(\D_3 e_3, e_a), \qquad \xi_a=\frac 1 2 \g(\D_4 e_4, e_a), \\
    \omb&=\frac 1 4 \g(\D_3 e_3, e_4), \qquad \om=\frac 1 4 \g(\D_4 e_4, e_3),\\
    \etab_a&=\frac 1 2 \g( \D_4 e_3, e_a), \qquad \eta_a=\frac 1 2 \g( \D_3 e_4, e_a),\\
    \ze_a&=\frac 1 2\g(\D_{e_a}e_4, e_3),
  \end{align*}
  and
  \begin{align*}
    \alpha_{ab}&=\W(e_a, e_4, e_b, e_4), \qquad \aa_{ab}=\W(e_a, e_3, e_b, e_3),\\
    \b_a&=\frac 1 2 \W(e_a, e_4, e_3, e_4), \qquad \bb_a= \frac 1 2\W(e_a, e_3, e_3, e_4), \\
    \rho&=\frac 1 4 \W(e_3, e_4, e_3, e_4), \qquad \dual\rho=\frac 1 4 \dual\W (e_3, e_4, e_3, e_4),
  \end{align*}
  where $\D$ is the spacetime covariant derivative, $\W$ is the Weyl
  curvature of the spacetime metric $\g$ and $\dual \W$ is its Hodge
  dual.  As usual, we denote
  \begin{align*}
    \trch\vcentcolon=\de^{ab}\chi_{ab}, \qquad \trchb\vcentcolon=\de^{ab}\chib_{ab}.
  \end{align*}
  We also define the following null electromagnetic components:
  \begin{align*}
    \bF_a&= \F(e_a, e_4), \qquad \bbF_a= \F(e_a, e_3), \\
    \rhoF&=\frac 1 2 \F(e_3, e_4), \qquad \dual\rhoF=\frac12\dual\F (e_3, e_4).
  \end{align*}

  \subsection{Geometric setting of \texorpdfstring{$\Si_*$}{}}\label{sec:geometric-setting-Sigma}

  Here we introduce the main defining properties of the hypersurface $\Si_*$ in $\MM$. Such properties are the ones assumed to hold for the boundary of a GCM admissible spacetime in the context of nonlinear stability of Kerr or Kerr-Newman, as in \cite{klainermanKerrStabilitySmall2023} (see also \cite{fangMassCenteredGCMFramework2025}). 

  Let $r$ and $u$ be two functions on $\MM$, such that $r$ restricted
  to $\Sigma_*$ induces the area radius function and moreover, there
  exists a constant $c_*$ such that $u= c_*- r$ on $\Sigma_*$.  We
  also impose the transversality conditions on $\Si_*$:
  \begin{align*}
    \xi=0, \qquad \om=0, \qquad \etab=-\ze, \qquad  e_4(r)=1, \qquad e_4(u)=0,
  \end{align*}
  which allows us to make sense of all the Ricci coefficients of $\Si_*$,
  as well as of all first-order derivatives of $r$ and $u$ on $\Si_*$. We define
  \begin{align*}
    y\vcentcolon= e_3(r), \qquad z\vcentcolon= e_3(u).
  \end{align*}
  Using the transversality condition for $e_4(r)$ and $e_4(u)$, and
  imposing that $(e_1, e_2)$ is adapted to the $r$-foliation on
  $\Si_*$, we have
  \begin{align*}
    \nab(r)= \nab(u)=0, \quad e_4(u) =0, \quad e_4(r) =1, \quad e_3(r)=y,\quad e_3(u)=z.
  \end{align*}
  The following relations hold true:
  \begin{align}\label{byz}
    \nab   y =-\xib+\big(\ze-\eta) y, \qquad  \nab z = (\ze-\eta ) z, \qquad b_*=-y-z.
  \end{align}
  We denote
  \begin{align}\label{def:nu}
    \nu= e_3+b_* e_4
  \end{align}
  the vectorfield tangent to $\Si_*$, orthogonal to the foliation and
  normalized by the condition $\g(\nu, e_4)=-2$.

  The hypersurface $\Si_*$ terminates in a boundary sphere $S_*$ on
  which the given function $r$ is constant, i.e. $S_*$ is a leaf of
  the $r$-foliation of $\Si_*$. Let $r_*$, $ u_* $ denote the values
  of $r$ and $u$ on $S_*$.  On $S_*$ there exist coordinates
  $(\th, \vphi)$ such that the induced metric $g$ on $S_*$ takes the
  form
  \begin{align}\label{eq:metric-on-S*}
    g= r^2e^{2\phi}\Big( (d\th)^2+ \sin^2 \th (d\vphi)^2\Big),
  \end{align}
  and the  functions 
  \begin{align*}
    J^{(0)} \vcentcolon=\cos\th, \qquad J^{(-)} \vcentcolon=\sin\th\sin\vphi, \qquad  J^{(+)} \vcentcolon=\sin\th\cos\vphi,
  \end{align*}
  verify  the balanced  conditions 
  \begin{align*}
    \int_{S_*}  J^{(p)} =0, \qquad p=0,+,-.
  \end{align*}

  We rely on a special
  orthonormal basis $(e_1, e_2)$ of the tangent space of $S_*$ given by
  \begin{align*}\label{eq:specialorthonormalbasisofSstar}
    e_1=\frac{1}{re^\phi}\pr_\th, \qquad e_2=\frac{1}{r\sin\th e^\phi}\pr_\vphi, \quad\mbox{on}\quad S_*.
  \end{align*}
  We define
  the 1-forms $f_0$, $f_+$ and $f_-$ defined on $S_*$ by:
  \begin{align*}
    &(f_0)_1 =0,  \qquad \qquad \qquad  (f_0)_2 =\sin\th, \\
    &(f_+)_1 =\cos\th\cos\vphi,  \qquad (f_+)_2 =-\sin\vphi,\\
    & (f_-)_1 =\cos\th\sin\vphi,  \qquad (f_-)_2=\cos\vphi,
  \end{align*}
  and extend them to $\Si_*$ by $\nab_\nu f_{p}=0$ for $p=0, \pm$.

  The coordinates $(\th, \vphi)$ and the functions $\Jp$ on $S_*$ are
  extended to $\Si_*$ by setting $ \nu(\th)=\nu(\vphi)=0$ and
  $\nu(\Jp)=0$, for $p=0,+,-$.  We also impose the transversality
  conditions on $\Si_*$ for $(\th, \vphi)$ and $\Jp$:
  \begin{align*}
    e_4(\th)=0, \qquad e_4(\vphi)=0, \qquad e_4(\Jp)=0, \qquad p=0,+,-.
  \end{align*}
  For a scalar function $\lambda$ on any sphere $S$ of $\Si_*$, we
  define its projection to the $\ell=1$ mode as the triplet functions
  \begin{align*}
    \lambda_{\ell=1}\vcentcolon=\biggl\{ \int_{S}J^{(p)} \lambda, \qquad p\in \{-, 0, +\}\biggl\}.
  \end{align*}
  We also denote by $K$ the Gauss curvature of the spheres $S$ of $\Si_*$.

  \subsection{Definitions of mass, charge and angular momentum at \texorpdfstring{$S_*$}{S*}}\label{sec:parameterdefinitionsonsstar}

  We now define the mass, charge and angular momentum parameters at the sphere $S_*$.

  \begin{definition}\label{def:parameters}
    We define the following constant parameters on $\Sigma_*$:
    \begin{enumerate}
    \item The \emph{auxiliary mass $m$} is defined to be the Hawking mass of $S_*$, i.e.,
      \begin{equation}\label{eq:definition-auxiliary-m}
        \frac{2m}{r_*}\vcentcolon=1+\frac{1}{16\pi} \int_{S_*} \trch\trchb. 
      \end{equation}
    \item  The \emph{electric charge $Q$} is defined as the average of $\rhoF$ at $S_*$, i.e.
      \begin{equation}
        \label{eq:M3:Q:def}
        Q \vcentcolon= \frac{1}{4\pi}\int_{S_*} \rhoF= \frac{1}{4\pi}\int_{S_*}\dual\F.
      \end{equation}
    \item  The \emph{magnetic charge $e$} is defined as the average of $\dual\rhoF$ at $S_*$, i.e.
      \begin{eqnarray}\label{eq:def-e}
        e \vcentcolon= \frac{1}{4\pi}\int_{S_*} \dual \rhoF= \frac{1}{4\pi}\int_{S_*}\F.
      \end{eqnarray}

    \item The \emph{mass parameter $M$} is defined as 
      \begin{align}
        M = m + \frac{Q^2}{2r_*}.
      \end{align}
    \item The \emph{angular  momentum $a$} is defined as 
      \begin{align}\label{eq:def-angular-momentum-S*}
        a\vcentcolon=\frac{r_*^3}{8\pi M}\int_{S_*} J^{(0)}\bm{J},
      \end{align}
      where $\bm{J}$ is defined as 
      \begin{align}\label{eq:def-J}
        \bm{J}\vcentcolon=\curl \b -\frac{Q}{r^2}\curl\bF +\frac{2Q}{r^3}\dual\rhoF.
      \end{align}
    \end{enumerate}
  \end{definition}

  \begin{remark}\label{rem:charge-constant}
    Observe that by Stokes Theorem and Maxwell's equation $d\F=0$,
    $d\dual \F=0$, the charges $Q$ and $e$ are independent of the
    choice of the sphere, i.e.
    \begin{align*}
      Q=\frac{1}{4\pi}\int_{S} \rhoF, \qquad e=\frac{1}{4\pi}\int_{S} \dual \rhoF,
    \end{align*} 
    for any sphere $S$.
    In particular, they are the same in the whole of $\Sigma_*$.
  \end{remark}

  \begin{remark}
    Any rotation of the Maxwell field
    $\F_{\lambda}=\cos\lambda \F + \sin\lambda \dual \F$ satisfies the
    Maxwell equations, therefore we can choose a parameter $\lambda$
    such that $e=0$ on $S_*$, and therefore by Remark
    \ref{rem:charge-constant} $e=0$ on $\Sigma_*$.
  \end{remark}

  \begin{remark}\label{ref:rem-J}
    The definition of $\bm{J}$ is motivated by its favorable behavior
    under transport: it satisfies a suitable transport equation along
    the tangent vector $\nu$ to $\Sigma_*$, \emph{as well as} an
    improved transport equation along the outgoing null direction
    $e_4$ (see already \eqref{eq:nab-4-J}). In contrast with the Kerr
    case, where $\bm{J}=\curl\beta$, the presence of electromagnetic
    contributions in the equations for $\beta$ necessitates additional
    lower-order terms. These terms play a crucial role in canceling
    problematic contributions in the transport equation. This is in
    parallel with the definition of the center of mass $\bm{C}$
    discussed below (see Remark~\ref{rem:bm-C}).
  \end{remark}

  In addition to the constant parameters above, we also make use of
  the following definition of function masses on $\Sigma_*$.
  \begin{definition} We define the following functions on $\Sigma_*$.
    \begin{enumerate}
    \item The \emph{Hawking mass} of a sphere $S$ in $\Si_*$ is given by 
      \begin{align}\label{dfHawking-mass}
        \frac{2m_H}{r} &= 1+\frac{1}{16\pi}\int_S\trch\trchb.
      \end{align}
    \item The \emph{mass aspect function} $\mu$ is defined as
      \begin{align}\label{def-mu}
        \mu\vcentcolon=-\div\ze-\rho +\frac{1}{2}\hch\c\hchb.
      \end{align}
    \item The \emph{center of mass function} is defined as
      \begin{align}\label{eq:definition-bm-C}
        \bm{C}\vcentcolon=\div\b -\frac{Q}{r^2}\div \bF+\frac{2Q}{r^3}\rhoFc 
      \end{align}
      where $\rhoFc \vcentcolon= \displaystyle\rhoF -\frac{Q}{r^2}$.
    \end{enumerate}
  \end{definition}

  \begin{remark}\label{rem:bm-C}
    The definition of the quantity $\bm{C}$ is motivated by the fact
    that it satisfies an improved transport estimate along the
    outgoing null direction $e_4$ (see already \eqref{eq:nab-4-C}),
    given by
    \begin{align*}
      \nab_4 \bm{C} &= -\frac{5}{r}\bm{C}    
                      -\div\div\a +\text{nonlinear terms}.
    \end{align*} 
    Contrary to the case of $\bm{J}$, as in Remark \ref{ref:rem-J}
    above, it is not possible to find a renormalization of
    $(\div\b)_{\ell=1}$ which satisfies an improvement \emph{both} in
    the transport equations along $\Si_*$ and along $e_4$. For this
    reason, the GCM condition $\bm{C}_{\ell=1}=0$ is imposed on all
    the spheres of the hypersurface as opposed to only at $S_*$ as in
    \cite{klainermanKerrStabilitySmall2023}.
  \end{remark}

  \begin{remark}
    Observe that, using the operator $\ddd_1=(\div, \curl)$, the
    quantities $\bm{C}$ and $\bm{J}$ can be written concisely as
    \begin{align}
      (\bm{C}, \bm{J})
      \vcentcolon={}& \ddd_1 \b-\frac{Q}{r^2}\ddd_1 \bF +\frac{2Q}{r^3}(\rhoFc, \dual\rhoF).\label{eqn:C-J-renormalization}
    \end{align}
  \end{remark}

  The mass and charge parameters given in Definition
  \ref{def:parameters} are used to define the linearized quantities
  with respect to their expected value in Reissner-Nordstr\"om.

  \begin{definition}\label{ref:definition-linearized}
    We define the following linearized quantities:
    \begin{align*}
      \trchc &\vcentcolon=  \ds\trch-\frac{2}{r},\qquad
               \trchbc \vcentcolon= \ds\trchb+\frac{2\Up}{r}, \qquad \ombc \vcentcolon= \ds\omb-\frac{M}{r^2}+\frac{Q^2}{r^3}, \\
      \yc &\vcentcolon= y +\Up,\qquad
            \zc \vcentcolon= z -2, \qquad \widecheck{b_*}: = \ds b_* + 2 - \Upsilon, \qquad \widecheck{K}\vcentcolon= K - \frac{1}{r^2}, \\
      \rhoc &\vcentcolon= \ds \rho +\frac{2M}{r^3}  - \frac{2Q^2}{r^4}, \qquad \widecheck{\mu} \vcentcolon= \ds \mu -\frac{2M}{r^3} + \frac{2Q^2}{r^4}, \qquad \rhoFc \vcentcolon= \displaystyle\rhoF -\frac{Q}{r^2},
    \end{align*}
    where $M$ and $Q$ are the mass and electric charge defined in
    Definition \ref{def:parameters} and
    $\Up = 1-\frac{2M}{r} + \frac{Q^2}{r^2}$.
  \end{definition}

  In addition to the above, we also introduce the following
  renormalized quantities:
  \begin{align*}
    \widecheck{\nab J^{(0)}}&\vcentcolon=\nab J^{(0)}+\frac{1}{r}\dual f_0, \qquad \widecheck{\nab J^{(+)}}\vcentcolon=\nab J^{(+)}-\frac{1}{r}f_+, \qquad \widecheck{\nab J^{(-)}}\vcentcolon=\nab J^{(-)}-\frac{1}{r}f_-, \\
    \widecheck{\curl(f_0)}&\vcentcolon=\curl(f_0)-\frac{2}{r}\cos\th, \qquad \widecheck{\div(f_\pm)}\vcentcolon=\div(f_\pm)+\frac{2}{r}J^{(\pm)}.
  \end{align*}

  \subsection{The electrovacuum mass-centered GCM conditions on \texorpdfstring{$\Si_*$}{}}\label{sec:GCMconditionsonSigmastar}

  Here we recall the mass-centered GCM conditions imposed on
  $\Sigma_*$ in \cite{fangMassCenteredGCMFramework2025}.  The
  conditions partly coincide with the ones in Kerr in vacuum, with the
  most important difference being that each sphere along $\Sigma_*$
  has vanishing center of mass.

  \begin{definition}\label{def-mass-cent-GCM-h} We say that $\Si_*$ is an \emph{electrovacuum mass-centered GCM hypersurface} if the following conditions are satisfied. 
    \begin{enumerate}
    \item On the sphere $S_*$ we have: 
      \begin{align}\label{eq:GCM-cond1}
        \trchc=0, \qquad \trchbc =0,\qquad \bm{C}_{\ell=1}=0 , 
      \end{align} 
      and 
      \begin{align}\label{eq:S_*-GCM}
        \int_{S_*}J^{(+)}\bm{J}=0, \qquad \int_{S_*}J^{(-)}\bm{J}=0.
      \end{align}
    \item  On any sphere of the $r$-foliation of $\Si_*$, we have 
      \begin{align}
        \label{eq:Si_*-GCM1}
        \begin{split}
          \trchc &=0,\\
          \trchbc &=\underline{C}_0+\sum_{p=0,+,-}\underline{C}_pJ^{(p)},\\ 
          \widecheck{\mu}  &=M_0+\sum_{p=0, +,-}M_p J^{(p)},
        \end{split}
      \end{align}
      where $\underline{C}_0$, $\underline{C}_p$, $M_0$, $M_p$ are
      scalar functions on $\Si_*$ and are constant on the leaves of the
      foliation. Also,
      \begin{align}
        \label{eq:Si_*-GCM2}
        \bm{C}_{\ell=1}=0 , \qquad\ov{b_*}=-1-\frac{2M}{r}+\frac{Q^2}{r^2}, 
      \end{align}
      where $\ov{b_{*}}$ denotes the average of $b_*$ on the spheres foliating $\Si_*$, and 
      \begin{align} \label{eq:Si_*-GCM2-xib}
        (\div\xib)_{\ell=1}=0.
      \end{align}
    \end{enumerate}
  \end{definition}

  \subsection{Bootstrap assumptions and dominance condition}

  We collect here the two additional assumptions on the mass-centered GCM hypersurface $\Sigma_*$: the bootstrap assumptions and the dominance condition.

  Recall the linearized quantities of Definition \ref{ref:definition-linearized}. We divide all the linearized geometrical quantities (connection coefficients, curvature and electromagnetic components) into two sets, denoted $\Gag, \Gab$, depending on their expected decay:
  \begin{align}
    \Ga_{g}:&=\Big\{ \trchc, \;  \hch,  \;   \ze,  \;  \trchbc, \;  r\a,\;  r\b, \; \bF, \; r \rhoc, \;  r \rhod, \; \dual\rhoF, \; \rhoFc, \; r \muc,  \;  r\widecheck{K} \Big\},\label{eq:defintionofGagforChapter5glsdfiuhgs}\\
    \Gab&=\Big\{ \eta,\;  \hchb, \; \widecheck{\omb}, \; \xib,\; r\bb,\; \aa,  \; \bbF, \; r^{-1}\yc , \; r^{-1}\zc, \; r^{-1} \widecheck{b_*}\Big\}.  \label{eq:Definition-Ga_b}
  \end{align}

  The bootstrap assumptions for the linearized quantities are given in terms of the $L^\infty$ norms on $S(u) \subset\Si_*$, i.e. $$\|f\|_{\infty,k}(u)= \sum_{i=0}^k \|\dk_*^i f\|_{\infty }(u),$$ where $\dk_*$ denotes tangential derivatives to $\Si_*$.

  \begin{bootstrap}
    We assume that there exist positive small constants $\ep>0$, $\dec>0$ and a large enough $N>0$ such that for any element of the set $\Ga_b$ and $\Ga_g$ the following estimates hold true on $\Si_*$ for $0\leq k\le  N$:
    \begin{align}\label{decayGagGabM3}
      \begin{split}
        \|\Gab\|_{\infty, k}&\les\frac{\ep}{ru^{1+\dec}}, \\
        \|\Gag\|_{\infty, k}&\les\frac{\ep}{r^2u^{\frac 1 2 +\dec}},\\
        \|\nab_\nu\Gag\|_{\infty,k-1}&\les\frac{\ep}{r^2u^{1+\dec}}.
      \end{split}
    \end{align}
  \end{bootstrap}
  Observe that, according to the above assumptions, for $r, u \gg 1$, we have
  \begin{align}\label{eq:rule-Gag-Gab}
    r^{-1}\|\Gab \|_{\infty, k}\lesssim \|\Gag\|_{\infty, k}.
  \end{align}

  The \emph{dominance condition} below quantifies the largeness of $r$ and $u$ along $\Si_*$.

  \begin{dominance} We assume that on $S_*$, the values $r_*$ and $u_*$ are related by the following relation
    \begin{align*}
      r_*= \de_*\ep_0^{-1}u_*^{1+\dec},
    \end{align*} 
    where $\de_*>0$, $\ep_0>0$ are positive small constants satisfying $\ep=\ep_0^{\frac 2 3}$. 
    This implies in particular on $\Si_*$
    \begin{align*}
      r\geq r_* \geq\de_*\ep_0^{-1}u^{1+\dec}, 
    \end{align*}
    which can be written as
    \begin{align}\label{dominanceM3}
      \frac 1 r \lesssim \frac{\ep_0}{u^{1+\dec}}.
    \end{align}
  \end{dominance}

  \subsection{The Einstein-Maxwell equations}

  We collect here the equations satisfied by the linearized quantities on $\Si_*$, using the bootstrap assumptions to simplify nonlinear terms\footnote{We keep in the equations only the nonlinear terms with the worst decay according to the rule \eqref{eq:rule-Gag-Gab}.}.

  \begin{proposition}\label{Prop.NullStr+Bianchi-lastslice}The following  equations hold true on $\Si_*$:
    \begin{enumerate}
    \item The linearized Maxwell equations are given by:
      \begin{align*}
        \nab_3 \bF
        -\nab\rhoFc+\dual \nab \dual\rhoF
        - \frac{\Upsilon}{r}\bF      
        ={}& \left(\frac{2M}{r^2}-\frac{2Q^2}{r^3}\right)\bF 
             + \frac{2Q}{r^2}\eta                      
             + \Gamma_b\cdot\Gamma_g,\\           
      \nab_4 \rhoFc-\div \bF
      ={}& - \frac{2}{r}\rhoFc,
      \\
      \nab_3 \rhoFc+\div \bbF
      ={}& \frac{2\Upsilon}{r}\rhoFc
           - \frac{Q}{r^2}\kabc           
           + \frac{2Q}{r^3}\yc
           + \Gamma_b\cdot\Gamma_b,
      \\
      \nab_4 \dual \rhoF-\curl \bF
      ={}&   -\frac{2}{r}\dual \rhoF,
      \\
      \nab_3 \dual \rhoF-\curl \bbF
      ={}&\frac{2\Upsilon}{r}\dual \rhoF
           + \Gamma_b\cdot\Gamma_b.          
    \end{align*}
\item  The linearized null structure equations are given by
 \begin{equation*}
      \begin{split}
        \nab_4\trchc
        &= \Ga_g\c\Ga_g,\\
        \nab_4\chih
        + \frac{2}{r} \chih
        &=-\a,\\
        \nab_4 \ze
        +\frac{2}{r} \ze
        &= -\beta
          - \frac{Q}{r^2}\bF
          +\Ga_g\c \Ga_g,\\
        \nab_4\trchbc
        +\frac{1}{r}\trchbc
        &= - 2 \div \ze + 2 \rhoc + \Ga_b\c \Ga_g,\\
        \nab_4 \chibh + \frac{1}{r} \chibh
        &= \frac{\Up}{r}\chih - \nab\hot \ze   +\Ga_b\c \Ga_g,
      \end{split}
    \end{equation*}
    and 
    \begin{equation*}
      \begin{split}
        \nab_3\trchc
        &=  2   \div \eta + 2\rhoc -\frac{1}{r}\trchbc + \frac{4}{r} \ombc +\frac{2}{r^2}\widecheck{y} +\Ga_b\c\Ga_b,\\
        \nab_3 \trchbc -\frac{2\Up}{r}  \trchbc
        &= 2\div \xib +\frac{4\Up}{r} \ombc -\left(\frac{2M}{r^2}+\frac{2Q^2}{r^3} \right) \trchbc  -\left(\frac{2}{r^2} -\frac{8M}{r^3} +\frac{6Q^2}{r^4} \right)\yc  +\Ga_b\c \Ga_b, \\
        \nab_3 \chibh -\frac{2\Up}{r}\chibh
        &= -\aa -\left(\frac{2M}{r^2} -\frac{2Q^2}{r^3}\right)\chibh +\nab\hot \xib +\Ga_b\c \Ga_b,\\
        \nab_3 \ze 
        &= - \bb- \frac{Q}{r^2}\bbF
          - 2 \nab \ombc
          +\frac{\Up}{r} (\eta+\ze)
          +\frac{1}{r}\xib
          + \left(\frac{2M}{r^2}-\frac{2Q^2}{r^3}\right)(\ze-\eta)
          +\Ga_b\c \Ga_b,\\
        \nab_3 \chih -\frac{\Up}{r} \chih
        &=\nab\hot \eta -\frac{1}{r} \chibh+\left(\frac{2M}{r^2}-\frac{2Q^2}{r^3} \right)\chih+\Ga_b\c \Ga_b.
      \end{split}
    \end{equation*}
    Also,
    \begin{align*}
      \div\chih
      &= \frac{1}{r}\ze
        - \b
        + \frac{Q}{r^2}\bF
        +\Ga_g\c \Ga_g,\\
      \div\chibh
      &= \frac{1}{2}\nab\kabc +
        \frac{\Up}{r}\ze
        +\bb
       - \frac{Q}{r^2}\bbF
        +\Ga_b\c \Ga_g,\\
      \curl\ze&= \rhod-\frac{1}{2}\chih\wedge\chibh,\\
      \curl\eta &= \rhod+\Ga_b\c \Ga_g,\\
      \curl \xib &=   \Ga_b\c \Ga_b,
    \end{align*}
    and
    \begin{align}
      \widecheck{K}&=-\frac{1}{2r} \trchbc  -\rhoc +\frac{2Q}{r^2}\rhoFc+ \frac 12  \chih \c \chibh, \label{eq:Gauss}\\
      \widecheck{\mu} &= -\div\ze -\rhoc + \frac 12  \chih \c \chibh.
    \end{align}
\item The linearized Bianchi identities are given by
 \begin{equation*}
   \begin{split}
     \nab_3\a -\frac{\Up}{r} \a    ={}& \nab\hot \b +\left(\frac{4M}{r^2} -\frac{4Q^2}{r^3}\right)\a +\left(\frac{6M}{r^3} -\frac{6Q^2}{r^4}\right)\chih + \frac{2Q}{r^2}\left( \nabla\hot\bF  - \frac{Q}{r^2}\chih  \right)\\
  &    +\Ga_b\c(\a,\b, r^{-1}\bF)+r^{-1}\dk^{\leq 1}(\Ga_g\c\Ga_g),\\
     \nab_4\beta +\frac{4}{r}\beta
     ={}&-\div\a
          + \frac{Q}{r^2}\nabla_4\bF+\Gag\c (\a,\b, r^{-1}\bF),\\
     \nab_3\b -\frac{2\Up}{r}\b
     ={}& (\nab\rho+\dual\nab\rhod)
          +\left(\frac{2M}{r^2} -\frac{2Q^2}{r^3}\right)\b
          -\left(\frac{6M}{r^3} -\frac{6Q^2}{r^4}\right)\eta
         + \frac{Q}{r^2}\left(
          \nabla\rhoF
          - \dual\nabla\dual\rhoF
          \right)\\
        & + \frac{\Upsilon Q}{r^3}\bF
          - \frac{2Q}{r^3}\bbF
          + r^{-1}\dk^{\leq1}(\Ga_b\c\Ga_g),\\
     \nab_4 \rhoc
     +\frac{3}{r}\rhoc
        &=\div \b
         - \frac{4Q}{r^3}\rhoFc
          + \frac{Q}{r^2}\div \bF
          +r^{-1}\dk^{\leq1}(\Ga_b\c\Ga_g),\\
     \nab_3 \rhoc
     -\frac{3\Up}{r} \rhoc
        &=  -\div\bb
          +\left(\frac{3M}{r^3} -\frac{4Q^2}{r^4}\right)\trchbc
          -\left(\frac{6M}{r^4}-\frac{8Q^2}{r^5}\right) \yc   -
          \frac{1}{2}\chih\c\aa
          +\frac{4Q\Up}{r^3}\rhoFc\\
        & - \frac{Q}{r^2}\div\bbF
          +r^{-1}\dk^{\leq1}(\Ga_b\c\Ga_b),\\
     \nab_4 \rhod+\frac{3}{r} \rhod
        &=-\curl\b
          - \frac{Q}{r^2}\curl\bF
          +r^{-1}\dk^{\leq1}(\Ga_b\c\Ga_g), \\
     \nab_3 \rhod
     -\frac{3\Up}{r} \rhod
        &=-\curl\bb -
          \frac{1}{2} \chih\c\dual \aa
           - \frac{Q}{r^2}\curl\bbF
          + r^{-1}\dk^{\leq1}(\Ga_b\c\Ga_b).
   \end{split}
    \end{equation*}
\end{enumerate}
\end{proposition} 
\begin{proof}
 The proof follows immediately from the Einstein-Maxwell equations (see for example \cite{giorgiElectromagneticgravitationalPerturbationsKerr2022}, the definition of the
  linearized quantities, the definition of $\Ga_g$ and $\Ga_b$, the
  fact that $y=e_3(r)$, and the GCM condition $\trchc=0$ on $\Si_*$.
\end{proof}

\subsubsection{Additional elliptic relations on \texorpdfstring{$\Si_*$}{}}

Let $\xi$ be an arbitrary one-form and $\th$ an arbitrary symmetric traceless $2$-tensor on $S$. We recall the following differential operators on the spheres:
\begin{itemize}
\item $\nabb$ denotes the covariant derivative associated to the metric $g$ on $S$.
\item $\ddd_1$ takes $\xi$ into the pair of functions $(\div \xi, \curl \xi)$, where $$\div \xi=g^{AB} \nabb_A \xi_B, \qquad \curl\xi=\ep^{AB}\nabb_A \xi_B.$$
\item $\dds_1$ is the formal $L^2$-adjoint of $\ddd_1$, and takes any pair of functions $(\rho, \sigma)$ into the one-form $-\nabb_A \rho+\ep_{AB} \nabb^B \sigma$.
\item  $\ddd_2$ takes $\th$ into the one-form $\ddd_2\th=(\div \th)_C=g^{AB}\nabb_A \th_{BC}$.
\item $\dds_2$ is the formal $L^2$-adjoint of $\DDd_2$, and takes $\xi$ into the symmetric traceless two tensor $$(\dds_2\xi)_{AB}=-\frac 12 \left( \nabb_B\xi_A+\nabb_A\xi_B-(\div \xi)g_{AB}\right).$$
\end{itemize}

We collect here some important elliptic relations between the quantities.

\begin{proposition}\label{Prop:nu*ofGCM:0}
  The following identities hold true on $\Si_*$:
  \begin{align*}
    2\dds_2\dds_1 \ddd_1\ddd_2\dds_2\eta&=-\frac{4}{r}\dds_2\dds_1\div\bb+r^{-5} \dkb^{\le 5} \Gag+ r^{-4} \dkb^{\leq 4}(\Gab\c\Gab),\\
    2\dds_2\dds_1 \ddd_1\ddd_2\dds_2\xib &=-\frac{4}{r}\dds_2\dds_1 \div\bb + r^{-5} \dkb^{\le 5} \Gag+ r^{-4} \dkb^{\leq 6}(\Gab\c\Gab).
  \end{align*}  
  Here $r^{-5} \dkb^{\le 5} \Gag$ denote a non specified sum of terms which multiplied by $r^5$ belong to the set $\dkb^{\le 5} \Gag$ in \eqref{eq:defintionofGagforChapter5glsdfiuhgs}.
\end{proposition}
\begin{proof}
The proof follows the same steps as in Proposition 5.23 and Corollary 5.24 in \cite{klainermanKerrStabilitySmall2023}. 
Observe that, even in the presence of the electromagnetic components, using the bootstrap assumption $\bbF \in \Ga_b$, following the same steps as in Proposition 5.22 in \cite{klainermanKerrStabilitySmall2023}, we can still write 
  \begin{align*}
    2 \nab \ombc -\frac{1}{r}\xib &= - \nab_3 \ze   -\bb +\frac{1}{r}\eta +r^{-1}\Ga_g+\Ga_b\c \Ga_b,\\
    2\ddd_2\dds_2\eta &=  -\nab_3\nab\kac    -\frac{2}{r}\nab_3 \ze  -\frac{2}{r}\bb  + r^{-2} \dkb^{\le 1} \Ga_g+ r^{-1} \dkb^{\le 1 } (\Ga_b\c \Ga_b),\\
    2\ddd_2\dds_2\xib  &= -\nab_3\nab\kabc   -\frac{2}{r}\nab_3 \ze    -\frac{2}{r}\bb     +r^{-2}\dkb^{\leq 1}\Ga_g+r^{-1}\dkb^{\leq 1}(\Ga_b\c \Ga_b).
  \end{align*}
  and also
\begin{align*}
    \nab_3\chibh&=-\aa+r^{-1}\dkb^{\leq1}\Gab+\Gab\c\Gab,\\
    \nab_3\rhoc&=-\div\bb-\frac{1}{2}\chibh\c \aa + r^{-2}\Gag + r^{-1}\dkb^{\leq1}(\Gab\c\Gab),\\
    \nab_3\rhod&=-\curl\bb-\frac{1}{2}\chibh\c \dual\aa + r^{-2}\Gag + r^{-1}\dkb^{\leq1}(\Gab\c\Gab).
\end{align*}
In view of the above, the proof follows the same steps as in \cite{klainermanKerrStabilitySmall2023}.
\end{proof}

We collect here a general lemma on commutations. 

\begin{lemma}\label{Lemma:Commutation-Si_*}
  The following commutation formulas hold true  for any  tensor $f$ on $S\subset\Si_*$:
  \begin{align*}
    [ \nab_3, \nab] f &= \frac{\Up}{r}\nab f +\Gab \c\nab_3 f + r^{-1} \Gab \c \dk^{\leq 1} f,\\
    [ \nab_4, \nab] f &=-\frac{1}{r}  \nab f +r^{-1}\Gag \c \dk^{\leq 1} f,\\
    [\nab_3,\De] f&= \frac{2\Up}{r} \lap  f+ r^{-1}\dkb^{\le 1}\Big( \Gab \c\nab_3 f + r^{-1} \Gab \c \dk f \Big),\\
    [\nab_4,\De] f&= -\frac{2}{r}\lap   f+ r^{-1} \dkb^{\le 1}\Big( \Gag \c  \dk f\Big),\\
    [\nab_\nu,\nab]f &=  \frac{2}{r}\nab f+\Gab\c\nab_\nu f+ r^{-1} \Gab \c \dk^{\leq 1} f,\\
    [\nab_\nu,\De]f&= \frac{4}{r} \De f+r^{-1}\dkb^{\le 1}\Big( \Gab \c\nab_\nu f + r^{-1} \Gab \c \dk f \Big).
  \end{align*}
\end{lemma} 
\begin{proof}
  See Lemma 5.1.20 in \cite{klainermanKerrStabilitySmall2023}.
\end{proof}

\subsection{Elliptic estimates on \texorpdfstring{$\Si_*$}{}}\label{sec:preliminaryestimatesonSistar:chap5}

We collect here basic Hodge elliptic estimates.    
For a tensor $f$ on $S$, we define the following standard weighted Sobolev norms for any integer $k\geq 0$
\begin{align*}
\|f\|_{\hk_k(S)}\vcentcolon=\sum_{j=0}^k\|\dkb^jf\|_{L^2(S)},
\end{align*}
where $\dkb$ denotes derivatives tangential to the sphere $S$.
\begin{lemma}\label{prop:2D-Hodge1}
For any sphere  $S=S(u)\subset \Si_*$ we have:
\begin{enumerate}
\item  If   $f$ is a 1-form
\begin{align*}
 \|  f\|_{\hk_{k+1}(S)}&\les r   \|\ddd_1   f  \|_{\hk_k(S)},\\
\|f\|_{\hk_{k+1}(S)}&\les r\|\dds_2 f\|_{\hk_k(S)}+r^2 \big|(\ddd_1f)_{\ell=1}\big|,\\
|(\ddd_1f)_{\ell=1}| &\les r^{-1} \|\ddd_1f\|_{L^2(S)}.
\end{align*}
Moreover, 
\begin{align}\label{eq:extra-estimate-Hodge}
     ||\dk_*^{k}f||_{\hk_3(S)}
 &\les r^4 ||\dds_2\dds_1\ddd_1\dk_*^{k}f||_{L^\infty(S)}+r^2\left|\left(\ddd_1\nab_\nu^{\leq k}f\right)_{\ell=1}\right|.
\end{align}
\item If $f$ is a symmetric traceless 2-tensor
\begin{align*}
\|v\|_{\hk_{k+1}(S)}\les r\|\ddd_2 v \|_{\hk_k(S)}.
\end{align*}
\item  If $(h, \dual h)$ is a pair of scalars 
\begin{align*}
\| (h-\ov{h}, \dual h-\ov{\dual h})\|_{\hk_{k+1} (S)} &\les r \|\dds_1(h,\dual h)\|_{\hk_{k}(S)},\\
\|(h-\ov{h},\dual h-\ov{\dual h})\|_{\hk_{k+2}  (S)} &\les  r^2\|\dds_2\,\dds_1(h, \dual h)\|_{\hk_k(S)}+r^3\big| (\Delta h)_{\ell=1}\big|+r^3\big| (\Delta\dual h)_{\ell=1}\big|.
\end{align*}
\end{enumerate}
\end{lemma}
\begin{proof}
See Lemmas 5.27 and 5.28 and Corollary 5.29 in \cite{klainermanKerrStabilitySmall2023}. To prove \eqref{eq:extra-estimate-Hodge}, we make use of the Hodge estimate to write
    \begin{align*}
      ||\dk_*^{k}f||_{\hk_3(S)}&\les r ||\ddd_1\dk_*^{k}f||_{\hk_2(S)} \les r^2 ||\dds_1\ddd_1\dk_*^{k}f||_{\hk_1(S)}\\
&\les r^3 ||\dds_2\dds_1\ddd_1\dk_*^{k}f||_{L^2(S)}+r^2\left|\left(\ddd_1\dk_*^{\leq k}f\right)_{\ell=1}\right|\\
  &\les r^4 ||\dds_2\dds_1\ddd_1\dk_*^{k}f||_{L^\infty(S)}+r^2\left|\left(\ddd_1\dk_*^{\leq k}f\right)_{\ell=1}\right|.
    \end{align*}
Finally, the $\dk_*^{\leq k}$ in the second term on right hand side can be reduced to only contains $\nab_\nu^{\leq k}$ since $\dkb^{\leq k}  \Jp$ is a basis of $\ell=1$ modes.
\end{proof}

\begin{lemma}\label{JpGag}
  \label{lemma:controloftheconformalfactorphi}
 We have the following estimate for $\phi$ on $\Si_*$:
  \begin{align*}
    \|\dkb^{\leq N}\phi\|_{L^\infty(S_*)} &\les \frac{\ep}{ru^{\frac 1 2 +\dec}}.
  \end{align*}
  The functions $\Jp$ verify the following properties on $\Si_*$:
  \begin{enumerate}
  \item We have
    \begin{align*}
      \begin{split}
        \int_{S}J^{(p)}=O\left(\frac{\ep r}{u^{\frac 1 2 +\dec}}\right), \qquad \int_{S}J^{(p)}J^{(q)}=\frac{4\pi}{3}r^2\de_{pq}+O\left(\frac{\ep r}{u^{\frac 1 2 +\dec}}\right).
      \end{split}
    \end{align*}
  \item We have
    \begin{align*}
\left(\triangle+\frac{2}{r^2}\right)\Jp\in r^{-1}\dkb^{\leq 1}\Gag,\qquad
\dds_2\dds_1\Jp\in r^{-1}\dkb^{\leq 3}\Gag.
    \end{align*}
  \end{enumerate}
\end{lemma}
\begin{proof} By the Gauss equation $\Kc\in r^{-1}\Gag$ and so by the bootstrap assumptions 
\begin{equation*}
    \left|K-\frac{1}{r^2}\right|_{L^\infty(S)}\les \frac{\ep}{r^3 u^{\frac 1 2 +\dec}}.
\end{equation*}
As a consequence, by Klainerman-Szeftel's effective uniformization theorem in \cite{klainermanEffectiveResultsUniformization2022}, the induced metric $g$ on any sphere $S$ is conformal to $\ga_{\mathbb S^2}$ as in \eqref{eq:metric-on-S*}, i.e.
\begin{equation}
    g=r^2 e^{2\phi}\ga_{\mathbb S^2},
\end{equation}
and moreover $\phi\in r \Gag$. The identities above then follow from the relations between $\Jp$ and $\phi$, for example
\begin{align*}
    \left(\triangle+\frac{2}{r^2}\right)\Jp=\frac{2}{r^2}(1-e^{-2\phi}) \Jp.
\end{align*}
  See also Lemmas 5.35, 5.36 and 5.37 in \cite{klainermanKerrStabilitySmall2023}.
\end{proof}

\begin{corollary}\label{prop:2D-Hodge4}
  On a fixed sphere  $S=S(u)\subset\Si_*$, we have for any pair of scalars  $(h, \dual h)$:
  \begin{align*}
    \|(h, \dual h)\|_{\hk_{k+2}  (S)}\les  r^2\|\dds_2\,\dds_1(h, \dual h)\|_{\hk_k(S)}+r\big| (h)_{\ell=1}\big|+r\big| (\dual h)_{\ell=1}\big|+r|\ov{h}|+r|\ov{\dual h}|.
  \end{align*}
\end{corollary}
\begin{proof}
  See Corollary 5.38 in \cite{klainermanKerrStabilitySmall2023}.
\end{proof}
\begin{corollary}\label{cor:Cb0CbpM0MareGagandrm1Gag}
  The scalar functions $\underline{C}_0$, $\underline{C}_p$, $M_0$, $M_p$ in \eqref{eq:Si_*-GCM1}  verify on $\Si_*$
  \begin{align}
    \underline{C}_0\in \Gag, \qquad \underline{C}_p\in \Gag, \qquad M_0\in r^{-1}\Gag, \qquad M_p\in r^{-1}\Gag.
  \end{align}
\end{corollary}
\begin{proof}
  See Corollary 5.39 in \cite{klainermanKerrStabilitySmall2023}.
\end{proof}

\subsection{Renormalized equations on \texorpdfstring{$\Si_*$}{}}

Here we define and derive a set of renormalized quantities, in addition to $(\bm{C}, \bm{J})$ defined in \eqref{eqn:C-J-renormalization},  which exhibit an improvement in the transport equation along the tangent vector to $\Si_*$, $\nu$, at the level of their
    $\ell=1$ mode.

  \begin{definition}
    \label{def:renormalized-quantities}
    We define the following quantities:
    \begin{align}
      \underline{\bm{C}}
      \vcentcolon={}& \div \underline{\beta} - \frac{Q}{r^2}\div \bbF - \frac{2Q}{r^3}\rhoFc, \label{eqn:CBar-renormalization}\\
      (\bm{Z}, \bm{Y})
      \vcentcolon={}& \ddd_1\zeta + \frac{2Q}{r^2}(\rhoFc, \dual\rhoF),\label{eq:zeta-renormalization}\\
      \widehat{\rho}
      \vcentcolon={}& \rhoc-\frac{2Q}{r^2}\rhoFc-\frac{1}{2}\hch\c\hchb,     \label{eq:rho-renormalization}\\
      \dual\widehat{\rho}
      \vcentcolon={}& \dual \rho + \frac{2Q}{r^2}\dual\rhoF -\frac{1}{2}\hch\wedge\hchb.   \label{eq:rhod-renormalization}
    \end{align}
  \end{definition}

\begin{remark}
The definition of $\widehat{\rho}$ is motivated by the Gauss equation \eqref{eq:Gauss}.
In fact, using that
\begin{align*}
     \muc=-\div\zeta-\rhoc+\frac{1}{2}\chih\c\chibh, \qquad \curl\ze=\dual\rho-\frac{1}{2}\chih\wedge\chibh,
\end{align*}
 one can see that the renormalized quantities satisfy
    \begin{align*}
   \widecheck{K}=-\frac{1}{2r} \trchbc  -\widehat{\rho}, \qquad      \bm{Z}=-\widehat{\rho}-\muc,\qquad \bm{Y}=\widehat{\dual\rho}.
    \end{align*}
\end{remark}

We now express the transport equations of each renormalized quantity in terms of the others.
\begin{proposition}
  \label{lemma:improved-transport-along-Sigma-star}
  We have the following transport equations along $\Sigma_{*}$:
  \begin{align*}         
    \nab_\nu\widehat{\rho}
    ={}&\frac{6}{r}\widehat{\rho} -\underline{\bm{C}}
         -(1+O(r^{-1}))\bm{C}  +r^{-3}\dk^{\leq 1}\Gab 
         +r^{-1}\dkb^{\le  1}( \Ga_b \c \Ga_b),\\
    \nabla_{\nu}\bm{C}
  ={}& \big(\frac{8}{r}+O(r^{-2})\big)\bm{C}- \frac{4Q}{r^3}\div \bbF 
       +\bigtriangleup\widehat{\rho}+\frac{2Q}{r^2}\big(\bigtriangleup +\frac{2}{r^2}\big)\rhoFc+O(1+O(r^{-1}))\div\div\a \\
       & -\big(\frac{6M}{r^3} -\frac{4Q^2}{r^4}\big)\div\eta \nonumber+r^{-5}\dk^{\leq 1} \Gag
       +r^{-2}\dkb^{\leq 1}(\Ga_b\c\Ga_g), \\
    \nabla_\nu\bm{J}&= \big(\frac{8}{r}+O(r^{-2})\big)\bm{J}  -\bigtriangleup\dual\widehat{\rho}-\frac{6M}{r^3} \dual\widehat{\rho} +\frac{2Q}{r^2}\big(\bigtriangleup +\frac{2}{r^2}\big)\dual\rhoF +O(1+O(r^{-1}))\curl\div\a \\
    & +r^{-5}\dk^{\le 1}\Gamma_g +r^{-2}\dkb^{\leq 1}(\Ga_b\c\Ga_g).
\end{align*}
\end{proposition}
\begin{proof}
  See Appendix \ref{appendix:lemma:improved-transport-along-Sigma-star}. 
\end{proof}

\begin{remark}
Notice that the transport equation for $\widehat{\rho}$ is not sensitive to the lower order term in $\rhoFc$ (as they behave better than the $r^{-3}\Gab$ terms on the right hand side), while the non linear term $\chih \c \chibh$ is important to cancel a nonfavorable nonlinear term in its transport equation.

While in the equation for $\bm{J}$ the term $\curl\bbF$ gets cancelled, in the equation for $\bm{C}$ the term $\div \bbF$ persists. For this reason we need to impose the vanishing of $\bm{C}$ on $\Si_*$ in order to avoid to transport $\bm{C}$ using its equation. As a consequence, the equation for $\nabla_\nu\bm{C}$ is not used as a transport equation but rather as an algebraic relation, so we only need slower decay and we can simplify it to 
  \begin{align*}
     \nabla_{\nu}\bm{C}
    ={}& \frac{8}{r}\bm{C}+\lap \widehat{\rho}-\frac{6M}{r^3}\div\eta- \frac{4Q}{r^3}\div \bbF         
         + O(1+O(r^{-1}))\div\div\alpha         
    \\
       & + r^{-4}\dk^{\le 1}\Gamma_g
         + r^{-2}\dk^{\le 1}(\Gamma_b\cdot\Gamma_g).
           \end{align*}
\end{remark}

We will now deduce the transport equations for the $\ell=1$ modes of the renormalized quantities. We first recall the following lemma, which appeared as Corollary 5.32 in \cite{klainermanKerrStabilitySmall2023}.

\begin{lemma}\label{Corr:nuSof integrals}
  For any scalar function $h$ on $\Si_*$, we have  
  \begin{align}
    \nu\left(\int_Sh\right) = \int_S\nu (h)-\frac{4}{r}\int_S h+r^3\Gab \c \nu(h)+r^2\Gab \c h
  \end{align}
  and $\nu(r)=\frac{rz}{2}\ov{z^{-1}(\kab + b_*\ka)}= -2 +r\Gab.$
  In particular, we have
  \begin{align*}
    \nu(\ov{h}) = \ov{\nu(h)}+r\Gab\nu(h)+\Gab h, 
  \end{align*}
  where $\ov{h}$ and $\ov{\nu(h)}$ denote respectively the average of
  $h$ and $\nu(h)$ on the spheres of $\Si_*$.
\end{lemma}
\begin{proof}
    See Corollary 5.32 in \cite{klainermanKerrStabilitySmall2023}.
\end{proof}

\begin{corollary}\label{COROFLEMMA:TRANSPORT.ALONGSI_STAR1} 
  We have the following transport equations along $\Sigma_{*}$, for $p=0,+,-$:
  \begin{align*}
    \nu\left(\int_S\widehat{\rho}\Jp\right)
  ={}& O(r^{-1})\int_S\widehat{\rho}\Jp -\int_S\underline{\bm{C}}\Jp 
      +r^{-1}\dk^{\leq1}\Gab +r\dkb^{\le  1}( \Ga_b \c \Ga_b),\\
    \nu\left( \int_S\bm{J}\Jp \right)
    ={}& O(r^{-1})\int_S\bm{J}\Jp
         + O(r^{-2})\int_S\dual\widehat{\rho}\Jp
         + r^{-3}\dk^{\le 1}\Gamma_g
         + \dk^{\le 1}(\Gamma_b\cdot\Gamma_g).
  \end{align*}
  
 We also have the following elliptic identity along $\Sigma_{*}$, for $p=0,+,-$:
  \begin{align} \label{eq:elliptic-identity-diveta}
  \frac{2}{r^2}\int_S \widehat{\rho}\Jp+\frac{6M}{r^3}\int_S\div\eta\Jp+ \frac{4Q}{r^3}\int_S\div \bbF  \Jp ={}& r^{-2}\dk^{\le 1}\Gamma_g  + \dk^{\le 1}(\Gamma_b\cdot\Gamma_g).
\end{align}
\end{corollary}
\begin{proof}
Since $\nu(\Jp)=0$, we have in view of Lemma \ref{Corr:nuSof integrals}
for any scalar function $h$ on $\Si_*$ and any $S\subset\Si_*$
\begin{equation}
  \label{eq:transportequationforell=1modealongSigmastargeneral}
  \nu\left(\int_Sh\Jp\right) = \int_S\nu (h)\Jp  -\frac{4}{r}\int_Sh\Jp +r^3\Ga_b \c \nu(h)+r^2\Ga_b \c h,
\end{equation}
where we also used $\Jp=O(1)$, and where we recall that the notation
$O(r^a)$, for $a\in\mathbb{R}$, denotes an explicit function of $r$
which is bounded by $r^a$ as $r\to+\infty$.

From Proposition \ref{lemma:improved-transport-along-Sigma-star} we deduce, using the GCM condition \eqref{eq:Si_*-GCM2}, that 
\begin{align*}
  \nu\left(\int_S\widehat{\rho}\Jp\right)
  ={}& O(r^{-1})\int_S\widehat{\rho}\Jp -\int_S\underline{\bm{C}}\Jp
      +r^{-1}\dk^{\leq1}\Gab +r\dkb^{\le  1}( \Ga_b \c \Ga_b),
\end{align*}
as stated.

Using Proposition
\ref{lemma:improved-transport-along-Sigma-star}  and
\eqref{eq:transportequationforell=1modealongSigmastargeneral}, and
noticing that the terms $O(r^a)$ only depend on $r$ and are thus
constant on $S$, we infer
\begin{equation*}
  \begin{split}
    \nu\left( \int_S\bm{C}\Jp \right)
    ={}& O(r^{-1})\int_S\bm{C}\Jp
         + \int_S\lap \widehat{\rho}\Jp-\frac{6M}{r^3}\int_S\div\eta\Jp- \frac{4Q}{r^3}\int_S\div \bbF  \Jp\\
         &
         + O(1+O(r^{-1}))\int_S\div\div\alpha\Jp
         + r^{-2}\dk^{\le 2}\Gamma_g
         + \dk^{\le 2}(\Gamma_b\cdot\Gamma_g).
  \end{split}
\end{equation*}
Together with the GCM condition \eqref{eq:Si_*-GCM2}
and integrating by parts to use Lemma \ref{JpGag}, we deduce 
\begin{equation*}
  \begin{split}
    0
    ={}&  -\frac{2}{r^2}\int_S \widehat{\rho}\Jp-\frac{6M}{r^3}\int_S\div\eta\Jp- \frac{4Q}{r^3}\int_S\div \bbF  \Jp\\
    &
         +r\left(\left|\left(\Delta+\frac{2}{r^2}\right)\Jp\right|+\left|\dds_2\dds_1\Jp\right|\right)\Ga_g+ r^{-2}\dk^{\le 2}\Gamma_g
         + \dk^{\le 2}(\Gamma_b\cdot\Gamma_g)\\
    ={}&   -\frac{2}{r^2}\int_S \widehat{\rho}\Jp-\frac{6M}{r^3}\int_S\div\eta\Jp- \frac{4Q}{r^3}\int_S\div \bbF  \Jp
        + r^{-2}\dk^{\le 2}\Gamma_g
         + \dk^{\le 2}(\Gamma_b\cdot\Gamma_g),
  \end{split}
\end{equation*}
where by $\dds_1\Jp$, we mean $\dds_1(\Jp,0)$. Similarly, using Proposition
\ref{lemma:improved-transport-along-Sigma-star} and noticing that $\dual\rhoF$ appears in the combination $\big(\bigtriangleup\dual\rhoF +\frac{2}{r^2}\dual\rhoF\big)$, we deduce the equation for $\bm{J}$ as
\begin{equation*}
  \begin{split}
    \nabla_{\nu}\left( \int_S\bm{J}\Jp \right)
    ={}& O(r^{-1})\int_S\bm{J}\Jp
         - O(r^{-2})\int_S\dual\widehat{\rho}\Jp
        +r\left(\left|\left(\Delta+\frac{2}{r^2}\right)\Jp\right|+\left|\dds_2\dds_1\Jp\right|\right)\Ga_g\\
       &  + O(r^{-3})\dk^{\le 2}\Gamma_g
         + \dk^{\le 2}(\Gamma_b\cdot\Gamma_g)\\
    ={}& O(r^{-1})\int_S\bm{J}\Jp
         - O(r^{-2})\int_S\dual\widehat{\rho}\Jp
         + r^{-3}\dk^{\le 2}\Gamma_g
         + \dk^{\le 2}(\Gamma_b\cdot\Gamma_g),
  \end{split}
\end{equation*}
as stated. 
\end{proof}

The above transport equations in $\nu$ along $\Si_*$ will be used to obtain estimates according to the following lemma.
\begin{lemma}\label{evolutionlemmaSi*}
  Let $f$ and $h$ be two scalar functions on $\Si_*$ satisfying
  \begin{align*}
    \nu(f) =O(r^{-1})f+h.
  \end{align*}
  Then, we have for any integer $m$
  \begin{align}
    r^m\|f\|_{L^\infty(S(u))}\les r_*^m\|f\|_{L^\infty(S_*)}+\int_{u}^{u_*}\|r^mh\|_{L^\infty(S(u'))}du'.
  \end{align}
\end{lemma}
\begin{proof}
  See Lemma 5.33 and Corollary 5.34 in \cite{klainermanKerrStabilitySmall2023}.
\end{proof}

\section{Consequences of the bootstrap assumptions and the dominance condition}\label{sec:consequences-BA}

In this section we obtain control of the $\ell=0$ and $\ell=1$ modes of most scalar quantities and 1-forms on $\Si_*$ solely\footnote{In particular, here we do not use the condition  \eqref{eq:Si_*-GCM2-xib} of $(\div\xib)_{\ell=1}=0$.} as a consequence of the bootstrap assumptions, the GCM conditions \eqref{eq:GCM-cond1},\eqref{eq:S_*-GCM},\eqref{eq:Si_*-GCM2} and the dominance condition \eqref{dominanceM3}.

\subsection{Estimates for \texorpdfstring{$\ell=0$}{l=0} modes on \texorpdfstring{$\Si_*$}{}}\label{sec:l=0}
In this section, we control the average (i.e. the $\ell=0$ mode) of the scalar quantities involved, i.e.
\begin{align*}
    \trchbc, \quad \rhoc, \quad \rhod, \quad  \rhoFc, \quad \rhodF, \quad \muc.
\end{align*} 
As a consequence of definition of charge, the Maxwell equations and Stokes theorem, recall Remark \ref{rem:charge-constant}, we have on $\Sigma_*$ 
\begin{equation}\label{eq:rhofc-ov-0}
    \ov{\rhoFc}=\ov{\rhoF}-\frac{Q}{r^2}=0, \qquad \ov{\dual\rhoF}=0,
\end{equation}
so there is no need to control the $\ell=0$ mode of these two quantities any further.

In what follows, we will make use of the control of the Hawking mass function defined in \eqref{dfHawking-mass}, in terms of the constant auxiliary mass defined in \eqref{eq:definition-auxiliary-m}. Notice that, by definition,
\begin{equation*}
  \frac{2m_H}{r} = 1+\frac{1}{16\pi}\int_S\trch\trchb
  = 1+\frac{1}{16\pi}\int_S\frac{2}{r}\left(-\frac{2\Up}{r}+\trchbc \right)
  =\frac{2M}{r}-\frac{Q^2}{r^2}+2r\ov{\trchbc}=\frac{2M}{r}-\frac{Q^2}{r^2}+r\Ga_g,
\end{equation*}
and so $M = m_H + \frac{Q}{2r}+r^2\Ga_g$.

\begin{lemma}
    We have on $\Si_*$:
    \begin{equation}
    |m_H-m|+ \abs*{m_H+\frac{Q^2}{2r}-M}\les \frac{\ep_0}{u^{1+2\dec}}.\label{eq:Mcontrol}
\end{equation}
\end{lemma}
\begin{proof}
  From the null
  structure equations, we deduce on $\Si_*$, 

  \begin{align*}
    \nu(\trch\trchb) ={}& e_3(\trch\trchb)+b_* e_4(\trch\trchb)\\
    ={}&  -\trch\trchb(\trchb+b_*\trch) +2(\trchb+b_*\trch)\rho+2\trch\div\xib +2\trchb\div\eta\\
                        & +\trch\left(  2 \xib\c(\eta-3\ze)-|\chibh|^2\right)+\trchb\left(  2|\eta|^2 -\chibh\c\chih\right)                  -2\trch\bbF\cdot\bbF  \\
                        & +b_*\left(    -2\trch\div\ze  +\trch\left( 2|\ze|^2 -\chih\c\chibh\right)-\trchb|\chih|^2   -2\trchb\bF\cdot\bF\right)        ,
  \end{align*}
  which we rewrite
  \begin{align*}
    \nu(\trch\trchb)+\trch\trchb(\trchb+b_*\trch)
    ={}&   2(\trchb+b_*\trch)\rho+\frac{4}{r}\div\xib -\frac{4\Up}{r}\div\eta\\
       & -\frac{4\left(\Up -2\right)}{r}\div\ze+r^{-1}\dkb^{\leq 1}(\Ga_b\c\Ga_b).
  \end{align*}
  Applying Lemma \ref{Corr:nuSof integrals}, we infer
  \begin{align*}
    \nu\left(\int_S\trch\trchb\right)
    ={}&  z\int_S\frac{1}{z}\Big(\nu(\trch\trchb)+(\kab+b_*\ka)\trch\trchb     \Big)\\
    ={}& z\int_S\frac{1}{z}\Bigg(2(\trchb+b_*\trch)\rho+\frac{4}{r}\div\xib -\frac{4\Up}{r}\div\eta -\frac{4\left(\Up-2\right)}{r}\div\ze+r^{-1}\dkb^{\leq 1}(\Ga_b\c\Ga_b)\Bigg).
  \end{align*}
  Integrating by parts the divergences, we deduce
  \begin{equation*}
    \nu\left(\int_S\trch\trchb\right) = 2z\int_S\frac{1}{z}(\trchb+b_*\trch)\rho +r\dkb^{\leq 1}(\Ga_b\c\Ga_b).
  \end{equation*}
  We infer
  \begin{equation*}
    \nu\left(\frac{2m_H}{r}\right) = \frac{z}{8\pi}\int_S\frac{1}{z}(\trchb+b_*\trch)\rho +r\dkb^{\leq 1}(\Ga_b\c\Ga_b).
  \end{equation*}
  On the other hand, using that $\nu(r)=\frac{rz}{2}\ov{z^{-1}(\trchb + b_*\trch)}$, we have 
  \begin{equation*}
    \nu\left(\frac{2m_H}{r}\right) = \frac{2\nu(m_H)}{r}-\frac{2m_H}{r^2}\frac{z}{8\pi r}\int_S\frac{1}{z}(\trchb+b_*\trch),
  \end{equation*}
  which yields
  \begin{equation*}
    \nu(m_H) = \frac{rz}{16\pi}\int_S\frac{1}{z}(\trchb+b_*\trch)\left(\rho+\frac{2m_H}{r^3}\right) +r^2\dkb^{\leq 1}(\Ga_b\c\Ga_b).
  \end{equation*}
  Observe that 
  \begin{align*}
    \rhoc=   \rho+\frac{2M}{r^3}-\frac{2Q^2}{r^4}=\rho+\frac{2m_H}{r^3}-\frac{Q^2}{r^4}+r^{-1}\Ga_g,
  \end{align*}
  and using the dominance condition to bound $\frac{1}{r^4} \les \frac{\ep_0}{r^3u^{1+\dec}} \les r^{-1}\Ga_g$, we deduce $\rho+\frac{2m_H}{r^3} \in r^{-1}\Ga_g$.
  Plugging in the above,
  \begin{equation*}
    \nu(m_H) = -\frac{1}{4\pi}\int_S\left(\rho+\frac{2m_H}{r^3}\right) +r^2\dkb^{\leq 1}(\Ga_b\c\Ga_b).
  \end{equation*}
  In order to deduce the decay for the average of $\rho$, recall the Gauss equation
  \begin{equation*}
    K = -\rho+\rhoF^2+\rhodF^2 -\frac{1}{4}\trch\trchb +  \frac{1}{2}\chih\c\chibh.
  \end{equation*}
  Integrating on $S$, and using the definition of the Hawking mass $m_H$, we obtain
  \begin{equation*}
    \int_SK = -\int_S\rho+\int_S\rhoF^2+\int_S\rhodF^2 -4\pi\left(\frac{2m_H}{r}-1\right) +\frac{1}{2}\int_S\chih\c\chibh.
  \end{equation*}
  Since from Gauss--Bonnet we have $\int_SK = 4\pi$,
  we infer
  \begin{equation*}
    \int_S\rho =  -\frac{8\pi m_H}{r} +\frac{1}{2}\int_S\chih\c\chibh
    +\int_S\rhoF^2+\int_S\rhodF^2,
  \end{equation*}
  and hence, 
  \begin{equation}\label{eq:average-rho}
    \ov{\rho}= -\frac{2m_H}{r^3}+\ov{\rhoF^2}+\Ga_b\c\Ga_g=-\frac{2m_H}{r^3}+\frac{Q^2}{r^4}+\Ga_b\c\Ga_g,
  \end{equation}
  where we used that $\ov{\rhoF^2}=\ov{\rhoF}^2+\Ga_g \c \Ga_g$. This gives
  \begin{equation*}
    \nu(m_H) = -\frac{Q^2}{r^2}+r^2 \Ga_b \c \Ga_g  +r^2\dkb^{\leq 1}(\Ga_b\c\Ga_b).
  \end{equation*}
  Using the bootstrap assumptions \eqref{decayGagGabM3} and the
  dominance condition \eqref{dominanceM3}, we deduce
  \begin{align*}
    |\nu(m_H-m)| \les\frac{\ep_0}{u^{2+2\dec}}+\frac{\ep^2}{ru^{\frac 32 +2\dec}}+ \frac{\ep^2}{u^{2+2\dec}}\les \frac{\ep_0}{u^{2+2\dec}}. 
  \end{align*}
  Also, recall that we have $m=m_H$ on $S_*$. Integrating on $\Si_*$
  from $S_*$, we deduce the stated inequality for $\abs*{m_H-m}$ in
  \eqref{eq:Mcontrol}. Similarly, using that $\nu(r)=-2+r\Ga_b$, we
  have
  \begin{align*}
    \nu\left(m_H+\frac{Q}{2r}-M\right)&=-\frac{Q^2}{r^2}-\frac{Q}{2r^2}\nu(r)+r^2 \Ga_b \c \Ga_g  +r^2\dkb^{\leq 1}(\Ga_b\c\Ga_b)\\
                                      &=r^{-1}\Ga_b+r^2 \Ga_b \c \Ga_g  +r^2\dkb^{\leq 1}(\Ga_b\c\Ga_b).
  \end{align*}
  Using the bootstrap assumptions \eqref{decayGagGabM3} and the dominance condition \eqref{dominanceM3}, we deduce
  \begin{align*}
    \abs*{\nu\left(m_H+\frac{Q}{2r}-M\right)}\les \frac{\ep_0}{u^{2+2\dec}}.
  \end{align*}
  Also, recall that we have $M=m_H+\frac{Q}{2r}$ on $S_*$. Integrating
  on $\Si_*$ from $S_*$, we deduce the stated inequality for
  $\abs*{m_H+\frac{Q^2}{2r}-M}$ in \eqref{eq:Mcontrol}.
\end{proof}

\begin{proposition}\label{prop:controlofell=0modesonSigmastar}
  The following estimates hold true on $\Si_*$:
  \begin{equation}
    r\left|\ov{\rhoc}\right|+r\left| \ov{\rhod}\right|+ r\left| \ov{\muc} \right|+\left|\ov{\trchbc}\right|\les\frac{\ep_0}{r^2u^{1+2\dec}}.
  \end{equation}
\end{proposition}

\begin{proof} 
  From \eqref{eq:average-rho} and \eqref{eq:Mcontrol}, we deduce
  \begin{align*}
    \left|\ov{\rhoc}\right|&=\left|\ov{\rho}+\frac{2M}{r^3}-\frac{2Q^2}{r^4} \right| \leq \left|\ov{\rho}+\frac{2m_H}{r^3}-\frac{Q}{r^4}\right|+\left|2\frac{(M-m_H-\frac{Q}{2r})}{r^3} \right|\\
                           &\les \Ga_b\c\Ga_g+\frac{\ep_0}{r^3u^{1+2\dec}} \les \frac{\ep_0}{r^3u^{1+2\dec}} ,
  \end{align*}
  where we used the bootstrap assumptions \eqref{decayGagGabM3}. By definition \eqref{def-mu} of $\mu$, this also implies the same bound for $\ov{\muc}$. Next, taking the average of $\curl\ze=\rhod -\frac{1}{2}\hch\wedge\hchb$,
  we infer from the bootstrap assumptions \eqref{decayGagGabM3}
  \begin{align}\label{rhodcontrol}
    |\ov{\rhod}|\les |\ov{\hch\wedge\hchb}|\les |\Gab\c\Gag|\les \frac{\ep_0}{r^3u^{\frac 3 2 +\dec}}.
  \end{align}
  Next, from the definition \eqref{dfHawking-mass} and the GCM condition for $\trch$ on $\Si_*$, we have  
  \begin{align*}
    \ov{\trchb}=-\frac{2}{r}\left(1-\frac{2m_H}{r}\right)=-\frac{2\Up}{r}+\frac{4}{r^2}\left(m_H+\frac{Q^2}{2r}-M\right).
  \end{align*}
  Together with \eqref{eq:Mcontrol}, we deduce
  \begin{align*}
    \left|\ov{\trchb}+\frac{2\Up}{r}\right|\les\frac{\ep_0}{r^2u^{1+2\dec}}.
  \end{align*}
  This concludes the proof of Proposition \ref{prop:controlofell=0modesonSigmastar}.
\end{proof}

\subsection{Estimates for the \texorpdfstring{$\ell=1$}{l=1} modes on \texorpdfstring{$\Si_*$}{}}\label{section:estimates-rll=1-Si_*}

Recall that
on $S_*$ we have the following, see \eqref{eq:def-angular-momentum-S*}, \eqref{eq:GCM-cond1}, \eqref{eq:S_*-GCM}:
\begin{align*}
  \trch=\frac 2 r, \qquad \trchb=-\frac{2\Up}{r}, \qquad
  \bm{C}_{\ell=1}=0, \qquad
  \bm{J}_{\ell=1,\pm}=0, \qquad  \bm{J}_{\ell=1,0}=\frac{2aM}{r^5}.
\end{align*} 
and on $\Si_*$ we have 
\begin{align*}
    \trch=\frac{2}{r}, \qquad \bm{C}_{\ell=1}=0.
\end{align*}

\begin{lemma}\label{Lemma: K-ell=1}
The following estimates hold true on $\Si_*$:
  \begin{align*}
    |{\widecheck{K}}_{\ell=1}| + \left|\widehat{\rho}_{\ell=1}+\frac{1}{2r}(\trchbc)_{\ell=1}\right|  &\les\frac{\ep_0}{r^4u^{1+2\dec}}.
  \end{align*}
  In particular, on $S_*$ we have
    \begin{align*}
    \left|\widehat{\rho}_{\ell=1} \right| &\les\frac{\ep_0}{r^4u^{1+2\dec}}.
  \end{align*}
\end{lemma}
\begin{proof}
Recall that by Lemma \ref{JpGag}  $\phi\in r \Gag$. By Klainerman-Szeftel's effective uniformization theorem in \cite{klainermanEffectiveResultsUniformization2022}, the Gauss curvature satisfies 
  \begin{align*}
    K&=-\Delta\phi+\frac{e^{-2\phi}}{r^2}.
  \end{align*}
  Hence, we have
  \begin{align*}
    K-\frac{1}{r^2}=-\left(\Delta+\frac{2}{r^2}\right)\phi+\frac{e^{-2\phi}-1+2\phi}{r^2}=-\left(\Delta+\frac{2}{r^2}\right)\phi+\Gag\c\Gag.
  \end{align*}
  Integrating by parts, we obtain from Lemma \ref{JpGag}
  \begin{align*}
    \int_{S_*}\left(K-\frac{1}{r^2}\right) \Jp
    &= -\int_{S_*}\left(\Delta+\frac{2}{r^2}\right)\phi \Jp+r^2\Gag\c\Gag\\
    &=-\int_{S_*}\phi \left(\Delta+\frac{2}{r^2}\right)\Jp+r^2\Gag\c\Gag\\
    &=r^2\dkb^{\leq 1}(\Gag\c\Gag).
  \end{align*}
 Using the bootstrap assumptions \eqref{decayGagGabM3}, we infer
  \begin{align*}
    r^2(\Kc)_{\ell=1}\les \frac{\ep^2}{r^2u^{1+2\dec}}\les\frac{\ep_0}{r^2u^{1+2\dec}}.
  \end{align*}
  Using the Gauss equation $\widecheck{K}=-\frac{1}{2r} \trchbc  -\widehat{\rho}$ on $\Si_*$ , we deduce
      \begin{equation}\label{eq:uniformization-K}
        \begin{split}
    \left|\widehat{\rho}_{\ell=1}+\frac{1}{2r}\kabc_{\ell=1}\right|    &\les \frac{\ep_0}{r^4u^{1+2\dec}},
        \end{split}
\end{equation}
as stated.
\end{proof}

\begin{proposition}\label{prop:control.ell=1modes-Si}
 In addition to $\bm{C}_{\ell=1}=0$, the following estimates hold true on $\Si_*$:
  \begin{align*}
 r^2\left|\bm{Z}_{\ell=1}\right|+ r\left|\widehat{\rho} _{\ell=1}\right|+r \left|\muc_{\ell=1}\right|+  \left|(\trchbc)_{\ell=1}\right|&\les\frac{\ep_0}{r^2u^{1+2\dec}}.
  \end{align*}
\end{proposition}

\begin{proof}
  We make the following local bootstrap assumption on $\Si_*$:
  \begin{align}
    \label{eq:BAlocal-Si_*}
    \begin{split}
      \left|(\trchbc)_{\ell=1}\right|
      &\leq\frac{\ep}{r^2u^{1+2\dec}}. 
    \end{split}
  \end{align}
  Notice that this assumption is stronger than the one implied by the bootstrap assumptions in \eqref{decayGagGabM3} and will be recovered in the process of the proof (in Step 4).

{\bf{Step 1: Estimate for $\bm{Z}_{\ell=1}$.}}  Differentiating with respect to $\div$ and projecting on the $\ell=1$
  modes the Codazzi equation for $\hch$, we infer 
  \begin{align*}
    (\div \ddd_2\hch)_{\ell=1}=\frac{1}{r}(\div\ze)_{\ell=1}-(\div\beta)_{\ell=1} +\frac{Q}{r^2}(\div\bF)_{\ell=1} +r^{-1}\dkb^{\leq 1}(\Gag\c\Gag).
  \end{align*}
  Recalling the definitions \eqref{eqn:C-J-renormalization} of $\bm{C}$
  and \eqref{eq:zeta-renormalization} of $\bm{Z}$ we deduce
  \begin{align*}
    (\div\ddd_2\hch)_{\ell=1}
    =\frac{1}{r}\bm{Z}_{\ell=1}
   -\bm{C}_{\ell=1}
    +r^{-1}\dkb^{\leq 1}(\Gag\c\Gag).
  \end{align*}
From Lemma \ref{JpGag}
  \begin{align*}
    (\ddd_1\ddd_2\hch)_{\ell=1,p}=\frac{1}{|S|}\int_S\ddd_1\ddd_2\hch\Jp =\frac{1}{|S|}\int_S\hch\c\dds_2\dds_1\Jp=r^{-1}\Gag\c\dkb^{\leq 3}\Gag.
  \end{align*}
  We then deduce
  \begin{align*}    
    \bm{Z}_{\ell=1}
    =r\bm{C}_{\ell=1}
    +\dkb^{\leq 3}(\Gag\c\Gag),
  \end{align*}
  and thus using the GCM condition $\bm{C}_{\ell=1}=0$, 
  \begin{equation}
    \label{eq:control.ell=1modes-Si-Step1:main}
    \begin{split}
      \abs*{\bm{Z}_{\ell=1}}
    \lesssim{}&  \dkb^{\leq 3}(\Gag\c\Gag).
    \end{split}    
  \end{equation}
  Using the bootstrap assumptions \eqref{decayGagGabM3}, we
  deduce
  \begin{align}
    \label{eq:control.ell=1modes-Si-Step1}
    \left|\bm{Z}_{\ell=1}\right|\les\frac{\ep^2}{r^4u^{1+2\dec}}\les \frac{\ep_0}{r^4u^{1+2\dec}}. 
  \end{align}

  {\bf{Step 2: Estimate for $\underline{\bm{C}}_{\ell=1}$.}} 
  Differentiating with respect to $\div$ and projecting on the $\ell=1$ modes the Codazzi equation for $\hchb$,
  we infer
  \begin{align*}
    (\div\ddd_2\hchb)_{\ell=1}
    ={}&\frac{1}{2}(\triangle\trchbc)_{\ell=1}+\frac{\Up}{r}(\div\ze)_{\ell=1}
    +\div\big( \underline{\beta}
         - \frac{Q}{r^2}\bbF \big)_{\ell=1}+r^{-1}\dkb^{\leq 1}(\Gab\c\Gag)\\
    ={}&\frac{1}{2}(\triangle\trchbc)_{\ell=1}+\frac{\Up}{r}\bm{Z}_{\ell=1}
    +\underline{\bm{C}}_{\ell=1}+O(r^{-4})\rhoFc+r^{-1}\dkb^{\leq 1}(\Gab\c\Gag).
  \end{align*}
  As in Step 1, we have $(\div\ddd_2\hchb)_{\ell=1}=r^{-1}\dkb^{\leq 3}(\Gab\c\Gag)$.
  Also, we have from Lemma \ref{JpGag},
  \begin{align*}
    (\triangle\kabc)_{\ell=1,p}=\frac{1}{|S|}\int_S\triangle\trchbc\Jp
    &=-\frac{2}{r^2}\frac{1}{|S|}\int_S\trchbc\Jp+\frac{1}{|S|}\int_S\kabc\left(\triangle+\frac{2}{r^2}\right)\Jp\\
    &=O(r^{-2})(\trchbc)_{\ell=1,p}+r^{-1}\dkb^{\leq 1}(\Gag\c\Gag).
  \end{align*}
  We deduce
  \begin{align*}
    |\underline{\bm{C}}_{\ell=1}|\les r^{-2}|(\trchbc)_{\ell=1}|+r^{-1}|\bm{Z}_{\ell=1}|+r^{-4}\Gag+r^{-1}|\dkb^{\leq 3}(\Gab\c\Gag)|.
  \end{align*}
  Together with \eqref{eq:BAlocal-Si_*},
  \eqref{eq:control.ell=1modes-Si-Step1}, the bootstrap assumptions
  \eqref{decayGagGabM3}, and the dominance condition
  \eqref{dominanceM3}, we obtain
\begin{align}\label{eq:control.ell=1modes-Si-Step2}
    |\underline{\bm{C}}_{\ell=1}|\les \frac{\ep}{r^4u^{1+2\dec}}+\frac{\ep_0}{r^5u^{1+2\dec}}+\frac{\ep}{r^6u^{\frac 1 2 +\dec}}+\frac{\ep^2}{r^4u^{\frac 32 +2\dec}}\les 
    \frac{\ep_0}{r^3u^{2 +3\dec}}.
  \end{align}

    {\bf{Step 3: Estimate for $\left(\widehat{\rho}\right)_{\ell=1}$.}}   Recall from Corollary \ref{COROFLEMMA:TRANSPORT.ALONGSI_STAR1} that
  we have along $\Si_*$, for $p=0,+,-$,
  \begin{align}
    \label{rhocevolution}
    \nu\left(\int_S\widehat{\rho}\Jp\right)
    =  O(r^{-1})\int_S\widehat{\rho}\Jp+h_1,
  \end{align}
  where the scalar function $h_1$ is given by 
  \begin{align*}
    h_1={}&-\int_S\underline{\bm{C}}\Jp
     +r^{-1}\Gab +r\dkb^{\le 1}(\Gab\c\Gab).
  \end{align*}  
  In view of 
  \eqref{eq:control.ell=1modes-Si-Step2},  the bootstrap assumptions \eqref{decayGagGabM3} and the dominance condition \eqref{dominanceM3}, we obtain
  \begin{align*}
    |h_1| &\les \frac{\ep_0}{r u^{2 +3\dec}}
            +\frac{\ep}{r^2u^{1+\dec}}+\frac{\ep^2}{ru^{2+2\dec}}
            \les\frac{\ep_0}{ru^{2+2\dec}}.
  \end{align*}
  By integration in $u$ we deduce
  \begin{align*}
    \int_u^{u_*}r|h_1|\les\frac{\ep_0}{u^{1+2\dec}}.
  \end{align*}
  Applying Lemma \ref{evolutionlemmaSi*} to \eqref{rhocevolution}, we deduce
  \begin{align*}
    r\left|\int_S\widehat{\rho} \Jp\right|
    \les r_*\left|\int_{S_*}\widehat{\rho} \Jp\right|+\frac{\ep_0}{u^{1+2\dec}}.
  \end{align*}
  Together with the control of
  $\widehat{\rho}_{\ell=1}$ on $S_*$ provided by Lemma
  \ref{Lemma: K-ell=1}, we infer
  \begin{align}\label{eq:control.ell=1modes-Si-Step3}
    \left|\widehat{\rho}_{\ell=1}\right| &\les \frac{\ep_0}{r^3u^{1+2\dec}}.
  \end{align}

   {\bf{Step 4: Estimate for $(\trchbc)_{\ell=1}$.}} 
Using Lemma \ref{Lemma: K-ell=1} and \eqref{eq:control.ell=1modes-Si-Step3}, we deduce
  \begin{align*}
 \left|(\trchbc)_{\ell=1}\right|  &\les\frac{\ep_0}{r^3u^{1+2\dec}}+ \frac{\ep_0}{r^2u^{1+2\dec}} \les \frac{\ep_0}{r^2u^{1+2\dec}}.
  \end{align*}
  which improves the local bootstrap assumption \eqref{eq:BAlocal-Si_*} for $\trchbc$.

    {\bf{Step 5: Estimate for $\muc_{\ell=1}$.}} 
Using that $ \bm{Z}=-\widehat{\rho}-\muc$, and together with \eqref{eq:control.ell=1modes-Si-Step1},
  \eqref{eq:control.ell=1modes-Si-Step3}, we
  infer 
  \begin{align}\label{eq:control.ell=1modes-Si-Step8}
  \begin{split}
  |  \muc_{\ell=1}|
    \les{}& |\bm{Z}_{\ell=1}|
         +|\widehat{\rho}_{\ell=1}|\\
          \les& \frac{\ep_0}{r^4u^{1+2\dec}}+ \frac{\ep_0}{r^3u^{1+2\dec}}
    \les\frac{\ep_0}{r^3u^{1+2\dec}}.
    \end{split}
  \end{align}
This completes the proof of Proposition \ref{prop:control.ell=1modes-Si}.
\end{proof}

\section{Solving the Einstein-Maxwell equations on \texorpdfstring{$\Si_*$}{}}\label{secM3}

The aim of this section is to solve the Einstein–Maxwell equations on a mass-centered GCM hypersurface $\Si_*$, as defined in Definition~\ref{def-mass-cent-GCM-h}. Since $\Si_*$ is spacelike, this amounts to solving the associated constraint equations on it, that is, deriving decay estimates for the relevant geometric quantities along $\Si_*$
  in terms of suitable seed data. The introduction of seed data is essential, as the constraint equations form an underdetermined system.

\subsection{The seed data and the statement of the theorem}

We define on $\Si_*$ the following quantities:
      \begin{align}
        \mathfrak{f}&\vcentcolon=- \nab \hot \bF +\rhoF \chih, \label{eq:M3:def-ff}\\
                \underline{\mathfrak{f}}&\vcentcolon=- \nab \hot \bbF -\rhoF \chibh, \label{eq:M3:def-ffb}\\
        \mathfrak{b}&\vcentcolon=2\rhoF  \b-3\rho\bF, \label{eq:def-bb-M3}\\
\underline{\mathfrak{b}}&\vcentcolon=2\rhoF\bb-3\rho\bbF , \label{eq:def-bbb-m3}
\end{align}
and
\begin{align}
\mathfrak{x}&\vcentcolon=\nab_4\bF+3\trch \bF.\label{eq:M3:remarkable-transport:bF}
      \end{align}
We also use the following tensors:
\begin{align}
    \pf
    ={}& -2r^3\dds_1(\rho, \dual\rho)
         +3r^2\dds_1(\rhoF,\dual \rhoF) + \mbox{Acceptable Error}_1, \label{eq:elliptic-identity-pF} \\
  \qf^{\F} ={}&  r^3\dds_2\dds_1(\rhoF, -\rhodF)
  + \mbox{Acceptable Error}_2,\label{eq:elliptic-identity-qfF} 
\end{align}
where we allow errors of the form 
\begin{align*}
     \mbox{Acceptable Error}_i=O(r^{-2})
         + \dk^{\leq i}\Ga_b       
         + r^2\dk^{\le i}(\Gamma_b\cdot\Gamma_g).
\end{align*}
Here, $\qf^\F$ also satisfies\footnote{Observe that \eqref{eq:elliptic-identity-nab3-qfF} can be deduced from \eqref{eq:elliptic-identity-qfF} and the Einstein-Maxwell equations.}
\begin{align}
  \nabla_{3}(r\qf^{\F})
  &= r^4\dds_2\dds_1(-\div\bbF, \curl\bbF) + O(r^{-2})
  + r\dkb^{\le 3}\Gamma_g
  + r^3\dk^{\le 3}(\Gamma_b\cdot\Gamma_g).\label{eq:elliptic-identity-nab3-qfF}
\end{align}
where $O(r^a)$ denotes, for $a \in \mathbb{R}$, a function of
$(r, \cos\th)$ bounded by a multiple of $r^a$ as $r\to +\infty$.

\begin{definition}[Seed data and its norms]\label{PROP:IDENTITIESINQF} 
 We say that the set $\{\alpha, \aa,  \mathfrak{f}, \mathfrak{b}, \underline{\mathfrak{b}}, \mathfrak{p}, \mathfrak{q}^\F\}$ constitute the \emph{seed data} for the constraint equations on $\Si_*$.

 We also define the following norms for the seed data: 
 \begin{align}
   \mathcal{F}^k[\aa, \underline{\mathfrak{b}}, \qf^\F]
   \vcentcolon={}&\int_{\Sigma_*(u, u_*) } u^{2+2\dec}    |\dk^{ k}_*\aa|^2+\int_{\Sigma_*(u, u_*) } r^6u^{2+2\dec}    |\dk^{ k+1}_*\underline{\mathfrak{b}}|^2\nonumber\\
                 &+ \int_{\Si_*}u^{2+2\dec}|\dk_*^{k}\nab_3\qf^\F|^2+\int_{\Si_*}|\dk_*^{k}\qf^\F|^2,\label{eq:def-FF-qf}\\
   \mathcal{G}^k[\a, \mathfrak{b}, \mathfrak{f},\pf, \qf^\F]
   \vcentcolon={}&r\|\a\|_{\infty,k}
                   +r\|\mathfrak{f}\|_{\infty,k}
                   + r^3\|\mathfrak{b}\|_{\infty,k+1}
                   + r^{-1}\|\qk^\F\|_{\infty, k}
                   +r^{-1}\|\pf\|_{\infty, \max(k,3)}. \label{eq;def-ee1}
 \end{align}
\end{definition}
Observe that as a consequence of the bootstrap assumptions and the
dominance condition we have
\begin{align}\label{eq:ba-on-Gk}
 \mathcal{G}^k[\a, \mathfrak{b}, \mathfrak{f},\pf, \qf^\F]\les    \frac{\ep}{r^2u^{\frac 1 2+\dec}}.
\end{align}

Combinations of curvature, electromagnetic tensors and Ricci
coefficients such as $\a, \mathfrak{f}$, $\mathfrak{b}$,
$\underline{\mathfrak{b}}$, $\mathfrak{x}$ are precisely the ones
identified in previous works by the second author
\cite{giorgiElectromagneticgravitationalPerturbationsKerr2022} as the
gauge-invariant quantities representing electromagnetic and
gravitational radiation.  The tensors $\pf$ and $\qf^\F$ are obtained
from $\mathfrak{b}$ and $\mathfrak{f}$ respectively through the
so-called Chandrasekhar transformation, see for example
\cite{giorgiElectromagneticgravitationalPerturbationsKerr2022} for the
definition in perturbations of Kerr-Newman. One can deduce the form
\eqref{eq:elliptic-identity-qfF}-\eqref{eq:elliptic-identity-pF} for
the real parts of the tensors as a consequence of the Einstein-Maxwell
equations. The curvature components $\alpha$ and $\aa$ are related to
$\mathfrak{f}, \underline{\mathfrak{f}}, \mathfrak{b},
\underline{\mathfrak{b}}$ through transport equations.

Observe that as a consequence of the linearizations and the definition of $\Gag, \Gab$, \eqref{eq:M3:def-ff}--\eqref{eq:M3:remarkable-transport:bF} reduce to  
\begin{align*}
  \mathfrak{f}&=- \nab \hot \bF +\frac{Q}{r^2} \chih +\Gag\c\Gag, \\
  \mathfrak{b}&=\frac{2Q}{r^2}\b+\left(\frac{6M}{r^3}-\frac{6Q^2}{r^4}\right)\bF + r^{-1}\Gag\c\Gag,\\
  \underline{\mathfrak{b}}&=\frac{2Q}{r^2}\bb+\left(\frac{6M}{r^3}-\frac{6Q^2}{r^4}\right)\bbF + r^{-1}\Gab\c\Gag,\\
  \mathfrak{x}&=\nab_4\bF+\frac{3}{r} \bF.
\end{align*}

\begin{remark}\label{rem:decay-gauge-invariant}

 In the context of the nonlinear stability of charged black holes, we expect to obtain control of the seed data quantities through the analysis of the hyperbolic system of PDEs they satisfy, as in \cite{giorgiBoundednessDecayTeukolsky2023}. 
In the context of the nonlinear stability of Reissner–Nordström, \cite{wanjingboNonlinearStabilitySubextremal2025} already establishes decay estimates for the gauge-invariant quantities $\pf$ and $\qf^\F$, valid for all $0\leq k\le  k_*$ for some large $k_*$. More precisely, for $\dee > \dec$, 
    \begin{align*}
\| \qk^\F\|_{\infty, k}+\| \pf\|_{\infty, k+1} \les\frac{\ep_0}{ru^{\frac 1 2 +\de_{extra}}}, \qquad
\| \nab_3 \qk^\F\|_{\infty, k-1} \les\frac{\ep_0}{ru^{1+\dee}}.
\end{align*}
In addition, one expects further bounds for $\mathfrak{b}$, $\mathfrak{f}$, and their fluxes, obtained through transport estimates from the ones of $\pf, \qf^\F$. More precisely, 
\begin{gather*}
\|\mathfrak{b}\|_{\infty, k+1} \les\frac{\ep_0}{r^{\frac{11}{2}+\dee}}, \qquad 
\|\mathfrak{f}\|_{\infty, k} \les\frac{\ep_0}{r^{\frac{7}{2}+\dee}},\\
    \int_{\Si_*(u, u_*)} u^{2+2\dec}|\nab_3 \dk^k \qf^\F|^2 + 
     \int_{\Sigma_*(u, u_*) } r^6u^{2+2\dec}    |\dk^{ k}_*\underline{\mathfrak{b}}|^2    \les \ep^2_0,   
\end{gather*}
see \cite{giorgiBoundednessDecayTeukolsky2020} for the proof in the linear case.
By the definitions of the seed data norms, these estimates would imply
\begin{align}\label{eq:bounds-EE-0-gauge-inv}
\mathcal{F}^k[\aa, \underline{\mathfrak{b}}, \qf^\F]\les \ep_0^2, \qquad
    \mathcal{G}^k[\a, \mathfrak{b}, \mathfrak{f},\pf, \qf^\F]&\les\frac{\ep_0}{r^2u^{\frac 1 2 +\de_{extra}}}.
\end{align} 
We further expect that analogous decay estimates will follow from the nonlinear generalized Regge–Wheeler system in Kerr–Newman; see \cite{giorgiBoundednessDecayTeukolsky2023} for the corresponding linear case.
\end{remark}

We are now ready to state the precise version of the main theorem. 

\begin{theorem}\label{prop:decayonSigamstarofallquantities}
Let $\Sigma_* \subset \MM$ be an electrovacuum mass-centered GCM hypersurface as in Definition \ref{def-mass-cent-GCM-h}, satisfying the dominance condition \eqref{dominanceM3}, in a spacetime $\MM$ solving the Einstein–Maxwell equations \eqref{eq:EM}. Assume further that the bootstrap assumptions \eqref{decayGagGabM3} hold for $\Gag$ and $\Gab$, and that for all $k\le  N$ one has 
\begin{align}\label{eq:assumption-FFk}
  \mathcal{F}^k[\aa, \underline{\mathfrak{b}}, \qf^\F]  \;\les\; \ep^2.  
\end{align}
Then, for all $ k\le  N-11$, the perturbation quantities along $\Si_*$ satisfy the decay estimates
 \begin{align}
\begin{split}
    \|\Ga_b \|_{\infty, k} &\les \frac{\sqrt{\mathcal{F}^N[\aa, \underline{\mathfrak{b}}, \qf^\F]+\ep_0^2}}{ru^{1+\dec}},\\
  \|\Gag\|_{\infty, k}&\les \mathcal{G}^k[\a, \mathfrak{b}, \mathfrak{f},\pf, \qf^\F]
                        +\frac{\sqrt{\mathcal{F}^N[\aa, \underline{\mathfrak{b}}, \qf^\F]+\ep_0^2}}{r^2u^{1+\dec}},\\
\|\nab_\nu\Gag\|_{\infty, k-1}&\les \frac{\sqrt{\mathcal{F}^N[\aa, \underline{\mathfrak{b}}, \qf^\F]+\ep_0^2}}{r^2u^{1+\dec}}.
\end{split}
\end{align}

Moreover, for all $k\le  N-10$, we obtain the improved bounds
\begin{align}\label{eq:improved-decay-M3}
  \begin{split}
    \big|\dk_*^{\leq k }\trchbc\big|+ r  \big|\dk_*^{\leq k }\muc\big|
    &\les\frac{\ep_0}{r^2u^{1+\dec}},\\
    r\big|\dk_*^{\leq k-1}\nab_\nu\b\big|+
    \big|\dk_*^{\leq k-1}\nab_\nu\bF\big|
    &\les r^{-1}\mathcal{G}^k[\a, \mathfrak{b}, \mathfrak{f},\pf, \qf^\F]
      +\frac{\sqrt{\mathcal{F}^N[\aa, \underline{\mathfrak{b}}, \qf^\F]+\ep_0^2}}{r^3u^{1+\dec}}, \\
    r\big|\dk_*^{\leq k-2} \nab^2_\nu\b|+\big| \dk_*^{\leq k-2}\nab^2_\nu\bF| &\les \frac{\sqrt{\mathcal{F}^N[\aa, \underline{\mathfrak{b}}, \qf^\F]+\ep_0^2}}{r^3u^{1+\dec}}.
  \end{split}
\end{align}
\end{theorem}
We separate the proof of Theorem \ref{prop:decayonSigamstarofallquantities} in three parts: first we control the $\Gab$ quantities, then we use the estimates for the $\Gab$ quantities to deduce additional control on $\ell=1$ modes, and finally we obtain control for the $\Gag$ quantities.

\begin{remark}
If the bounds in \eqref{eq:bounds-EE-0-gauge-inv} were proven to hold independently, then Theorem \ref{prop:decayonSigamstarofallquantities} would yield the following unconditional bounds for $\Gab, \Gag$ on $\Si_*$:
 \begin{align*}
\begin{split}
    \|\Ga_b \|_{\infty, k} &\les \frac{\sqrt{\mathcal{F}^k[\aa, \underline{\mathfrak{b}}, \qf^\F]+\ep_0^2}}{ru^{1+\dec}} \les \frac{\ep_0}{ru^{1+\dec}},\\
    \|\Gag\|_{\infty, k }&\les \mathcal{G}^k[\a, \mathfrak{b}, \mathfrak{f},\pf, \qf^\F]+\frac{\sqrt{\mathcal{F}^k[\aa, \underline{\mathfrak{b}}, \qf^\F]+\ep_0^2}}{r^2u^{1+\dec}} \\
    &\les \frac{\ep_0}{r^2u^{\frac 1 2 +\de_{extra}}}+\frac{\ep_0}{r^2u^{1+\dec}} \les \frac{\ep_0}{r^2u^{\frac 1 2 +\dec}},\\
\|\dk_*^{\leq k-1}\nab_\nu\Gag\|_{\infty}&\les \frac{\sqrt{\mathcal{F}^k[\aa, \underline{\mathfrak{b}}, \qf^\F]+\ep_0^2}}{r^2u^{1+\dec}} \les \frac{\ep_0}{r^2u^{1 +\dec}},
\end{split}
\end{align*}
therefore improving the bootstrap assumptions \eqref{decayGagGabM3}.
\end{remark}

\begin{remark}
    Assumption \eqref{eq:assumption-FFk} on the flux norms $\mathcal{F}^k[\aa, \underline{\mathfrak{b}}, \qf^\F]$ is not implied by the bootstrap assumptions \eqref{decayGagGabM3} but it is needed in the proof of Proposition \ref{prop:control.ell=1modes-Si-2} to improve the local bootstrap assumptions there.
\end{remark}

\subsection{Control of \texorpdfstring{$\Ga_b$}{} quantities on \texorpdfstring{$\Si_*$}{}}\label{sec:controlofthefluxofsomequantitiesonSigmastar}

In this section we obtain control of the quantities in the set $\Ga_b$ as defined in \eqref{eq:Definition-Ga_b}.
The goal of this section is to establish the following.
\begin{proposition}
  \label{Prop.Flux-bb-vthb-eta-xib}
  The following estimate holds true for all  $k\le N- 7$ 
  \begin{equation}
    \label{Estimate:Flux-bb-vthb-eta-xib}
    \int_{\Si_*}u^{2+2\dec}    \big|\dk_*^ k\Ga_b\big|^2  \les \mathcal{F}^N[\aa, \underline{\mathfrak{b}}, \qf^\F]+\ep_0^2.
  \end{equation}
Using the trace theorem and Sobolev, we deduce that for $k\le N - 9$,
  \begin{equation}
    \label{Estimate:Flux-bb-vthb-eta-xib-weak}
    \|\Ga_b \|_{\infty, k} \les \frac{\sqrt{\mathcal{F}^N[\aa, \underline{\mathfrak{b}}, \qf^\F]+\ep_0^2}}{ru^{1+\dec}}.
  \end{equation}
\end{proposition}
Observe that the flux estimate \eqref{Estimate:Flux-bb-vthb-eta-xib} is trivially verified for $\aa \in \Gab$ in view of the definition of $\mathcal{F}^k[\qf^\F, \underline{\mathfrak{b}}, \aa]$ in \eqref{eq:def-FF-qf}.

\begin{proof}
  {\bf{Step 1: Estimate for $\bbF_{\ell\geq 2}$.}}
  From \eqref{eq:elliptic-identity-nab3-qfF}, we infer, using the bootstrap assumptions \eqref{decayGagGabM3} and the dominance condition \eqref{dominanceM3} that
  \begin{align*}
    r^4|\dk_*^{k}\dds_2\dds_1(-\div\bbF, \curl\bbF)|
    &\les|\dk_*^{k}\nab_3(r\qk^\F)|+\frac{1}{r^2}  + |r\dkb^{\le 3}\Gamma_g|
      + r^3|\dk^{\le 3}(\Gamma_b\cdot\Gamma_g)|\\
    &\les|\dk_*^{k}\nab_3(r\qk^\F)|+\frac{\ep_0}{u^{2+2\dec}} +\frac{\ep}{ru^{\frac 1 2 +\dec}}
      +\frac{\ep^2}{u^{\frac 3 2 +2\dec}}\\
    &\les r|\dk_*^{k}\nab_3\qk^\F|+|\dk_*^{\leq k}\qk^\F| +\frac{\ep_0}{u^{\frac 3 2 +2\dec}}.
  \end{align*}
  Dividing by $r$, squaring and integrating on $\Si_*$ we deduce
  \begin{align*}
    \int_{\Si_*}r^6u^{2+2\dec}|\dk_*^k\dds_2\dds_1\big(-\div\bbF, \curl\bbF\big)|^2\les{}
    & \int_{\Si_*}u^{2+2\dec}|\dk_*^{k}\nab_3\qf^\F|^2\\
    &+\int_{\Si_*}u^{2+2\dec}r^{-2}|\dk_*^{k}\qf^\F|^2+\int_{\Si_*}\frac{\ep_0^2}{r^2u^{1+\dec}}.
  \end{align*}
  Using the dominance condition \eqref{dominanceM3} in the second term on the right, we deduce for  $k\le  N$,
  \begin{equation*}
    \int_{\Si_*}r^6u^{2+2\dec}|\dk_*^k\dds_2\dds_1\big(-\div\bbF, \curl\bbF\big)|^2\lesssim \mathcal{F}^N[\aa, \underline{\mathfrak{b}}, \qf^\F]+\ep_0^2.
  \end{equation*}
  Taking into account the commutator Lemma \ref{Lemma:Commutation-Si_*},  this yields, for $k\leq N$,
  \begin{equation*}
    \int_{\Si_*}r^6u^{2+2\dec}|\dds_1(\div, -\curl)\dds_2\,\dk_*^k\bbF|^2 \lesssim \mathcal{F}^N[\aa, \underline{\mathfrak{b}}, \qf^\F]+\ep_0^2.
  \end{equation*}
  Using the Hodge estimates of Lemma \ref{prop:2D-Hodge1} for
  $\dds_1$ and $\ddd_1$, we infer, for $k\le N-2$, 
  \begin{equation}
    \label{Estim:Flux-dds_2bbF}
    \int_{\Si_*}r^2 u^{2+2\dec}|\dds_2\,\dk_*^k\bbF|^2 \les \mathcal{F}^N[\aa, \underline{\mathfrak{b}}, \qf^\F]+\ep_0^2.
  \end{equation}

  {\bf{Step 2: Estimate for $\bb_{\ell\geq 2}$.}}
  From the definition of $\underline{\mathfrak{b}}$ in
  \eqref{eq:def-bbb-m3}, we infer
  \begin{align*}
    \dds_2\underline{\mathfrak{b}}
    &= \frac{2Q}{r^2} \dds_2\bb+ \left(\frac{6M}{r^3}-\frac{6Q^2}{r^4} \right) \dds_2\bbF +r^{-2}\dk^{\leq 1}( \Ga_b \c \Ga_g).
  \end{align*}
  We deduce using the bootstrap assumptions \eqref{decayGagGabM3}:
  \begin{align*}
    r^2| \dds_2\dk_*^k\bb |
    &\les  r^4 |\dds_2\dk_*^k\underline{\mathfrak{b}}| + r|\dds_2\dk_*^k\bbF| +\frac{\ep_0}{ru^{\frac 3 2 +2\dec}}.
  \end{align*}
  Squaring and integrating on $\Si_*$ we deduce $k\le N-2$
  \begin{align*}
    \int_{\Si_*}r^4u^{2+2\dec}| \dds_2\dk_*^k\bb |^2
    &\les    \int_{\Si_*} r^8u^{2+2\dec} |\dds_2\dk_*^k\underline{\mathfrak{b}}|^2 +    \int_{\Si_*}r^2 u^{2+2\dec}|\dds_2\dk_*^k\bbF|^2 +   \int_{\Si_*} \frac{\ep_0^2}{r^2u^{3 +4\dec}}\\
    &\les \mathcal{F}^N[\aa, \underline{\mathfrak{b}}, \qf^\F]+ \ep_0^2,
  \end{align*}
  where we used \eqref{Estim:Flux-dds_2bbF}. Using the Hodge estimate
  \eqref{eq:extra-estimate-Hodge} of Lemma \ref{prop:2D-Hodge1}, we
  infer for $k\le  N-1$,
  \begin{align}\label{Estim:Flux-dds_2bb}
    \int_{\Si_*} r^2u^{2+2\dec}\big|\dk_*^k\bb\big|^2\les \int_{\Si_*}r^4u^{2+2\dec}\left|\left(\ddd_1\nab_\nu^{\leq k-3}\bb\right)_{\ell=1}\right|^2+\mathcal{F}^N[\aa, \underline{\mathfrak{b}}, \qf^\F]+\ep_0^2.
  \end{align}

  {\bf{Step 3: Estimate for $\bbF_{\ell=1}$ and $\bb_{\ell=1}$.}}
  We now need to add the control of the $\ell=1$ mode of $\bb$ and $\bbF$.
  Writing the Codazzi equation as
  \begin{align}\label{eq:codazzi-simplified}
    \ddd_2\hchb=\bb+r^{-1}\dkb^{\leq 1}\Gag+\Gab\c\Gag,
  \end{align}
  and differentiating it with respect to $\nu^k\ddd_1$, we infer
  \begin{align*}
    \nu^k\ddd_1\ddd_2\hchb=\nu^k\ddd_1\bb +r^{-2}\dk_*^{\leq k+2}\Gag+r^{-1}\dk_*^{\leq k+1}(\Gab\c\Gag).
  \end{align*}
  Taking into account Lemma \ref{Lemma:Commutation-Si_*} and projecting on the $\ell=1$ modes, this yields for
  \begin{align*}
    (\ddd_1\ddd_2\nab_\nu^k\hchb)_{\ell=1}=(\ddd_1\nab_\nu^k\bb)_{\ell=1} +r^{-2}\dk_*^{\leq k+2}\Gag+r^{-1}\dk_*^{\leq k+1}(\Gab\c\Gag).
  \end{align*}
  Next, we have from Lemma \ref{JpGag} that for 
  \begin{align*}
    (\ddd_1\ddd_2\nab_\nu^k\hchb)_{\ell=1,p}=\frac{1}{|S|}\int_S\ddd_1\ddd_2\nab_\nu^k\hchb\Jp =\frac{1}{|S|}\int_S\nab_\nu^k\hchb\c\dds_2\dds_1\Jp=r^{-1}\dk^k\Gab\c\dkb^{\leq 3}\Gag.
  \end{align*}
  Using the bootstrap assumptions \eqref{decayGagGabM3} and the dominance condition \eqref{dominanceM3}, we then infer for $k\le  N-2$
  \begin{align}\label{eq:estimate(bb*)(ell=1)}
    \begin{split}
      |(\ddd_1\nab_\nu^k\bb)_{\ell=1}|&\les|(\ddd_1\ddd_2\nab_\nu^k\hchb)_{\ell=1}|+\frac{\ep}{r^4u^{\frac 1 2 +\dec}}+\frac{\ep^2}{r^4u^{\frac 3 2 +2\dec}}\les \frac{\ep_0}{r^3u^{\frac 3 2 +2\dec}}.
    \end{split}
  \end{align}
  Combining \eqref{Estim:Flux-dds_2bb} and \eqref{eq:estimate(bb*)(ell=1)}, we obtain for $k\le N-2$,
  \begin{eqnarray}\label{eq:flux-estimate-bb}
    \int_{\Si_*} r^2u^{2+2\dec}\big|\dk_*^k\bb\big|^2\les \mathcal{F}^N[\aa, \underline{\mathfrak{b}}, \qf^\F]+\ep_0^2.
  \end{eqnarray}
  Using again
  \eqref{eq:def-bbb-m3}, i.e.
  \begin{align*}
    \ddd_1\underline{\mathfrak{b}}
    &= \frac{2Q}{r^2} \ddd_1\bb
      + 3\left( \frac{2M}{r^3}-\frac{2Q^2}{r^4}  \right) \ddd_1\bbF +r^{-2}\dk^{\leq 1}( \Ga_b \c \Ga_g),
  \end{align*}
  we deduce for $k\le  N-2$ 
  \begin{align*}
    | (\ddd_1\nab_\nu^k\bbF )_{\ell=1}|&\les r^3|(\ddd_1\nab_\nu^k\underline{\mathfrak{b}})_{\ell=1}|
                                         + r |(\ddd_1\nab_\nu^k\bb)_{\ell=1}|+\frac{\ep^2}{r^2u^{\frac 3 2 +2\dec}}\\
                                       &\les r^3|(\ddd_1\nab_\nu^k\underline{\mathfrak{b}})_{\ell=1}|
                                         +  \frac{\ep_0}{r^2u^{\frac 3 2 +2\dec}},
  \end{align*}
  where we used \eqref{eq:estimate(bb*)(ell=1)}.
  Combining \eqref{Estim:Flux-dds_2bbF} and the above we obtain for $k\le N-2$,
  \begin{equation*}
    \int_{\Si_*} u^{2+2\dec}\big|\dk_*^k\bbF\big|^2\les \mathcal{F}^N[\aa, \underline{\mathfrak{b}}, \qf^\F]+\ep_0^2.
  \end{equation*}

  {\bf{Step 4: Estimate for $\chibh$.}}
  Using once again the Codazzi equation \eqref{eq:codazzi-simplified}, in view of  Lemma \ref{Lemma:Commutation-Si_*}, we have that
  \begin{align*}
    \ddd_2\dk_*^k\hchb=(\dk_*^k\bb)+r^{-1}\dkb^{\le k+2}\Gag +\dkb^{\leq k+1}\left(\Gab\c\Gag\right).
  \end{align*}
  Using the bootstrap assumptions \eqref{decayGagGabM3}, the dominance
  condition \eqref{dominanceM3}, and \eqref{eq:flux-estimate-bb}, we
  infer for $k\le  N-2$,
  \begin{eqnarray}\label{eq:thefluxesimateforthetabaronSi*tobeprovedbelow}
    \int_{\Si_*}  u^{2+2\dec}   \big|\dk_*^{k} \chibh|^2  \les \mathcal{F}^k[\aa, \underline{\mathfrak{b}}, \qf^\F]+\ep_0^2.
  \end{eqnarray}

  {\bf{Step 5: Estimate for $\eta_{\ell\geq2}$ and
      $\xib_{\ell\geq2}$.}}  From Proposition \ref{Prop:nu*ofGCM:0},
  and using Lemma \ref{Lemma:Commutation-Si_*}, we deduce for  
  \begin{align*}
    \begin{split}
      2\dds_2\dds_1 \ddd_1\ddd_2\dds_2\dk_*^k\eta
      &=r^{-4}\dk_*^{\leq k+3}\bb+r^{-5}\dk_*^{\leq k+5}\Gag+r^{-4}\dk_*^{\leq k+4}(\Gab\c\Gab),\\
      2\dds_2\dds_1 \ddd_1\ddd_2\dds_2\dk_*^k\xib
      &=r^{-4}\dk_*^{\leq k+3}\bb+r^{-5}\dk_*^{\leq k+5}\Gag+r^{-4}\dk_*^{\leq k+6}(\Gab\c\Gab).
    \end{split}
  \end{align*}
  Using the bootstrap assumptions \eqref{decayGagGabM3}, the dominance
  condition \eqref{dominanceM3}, we infer for
  $k\le  N-6$,
  \begin{align*}
    |\dds_2\dds_1 \ddd_1\ddd_2\dds_2\dk_*^k\eta|+|\dds_2\dds_1 \ddd_1\ddd_2\dds_2\,\dk_*^k\xib|
    &\les r^{-4}|\dk_*^{\leq N-3}\bb|+\frac{\ep}{r^7u^{\frac 1 2 +\dec}}+\frac{\ep^2}{r^6u^{2+2\dec}}\\
    &\les r^{-4}|\dk_*^{\leq N-3}\bb|+\frac{\ep}{r^7u^{\frac 1 2 +\dec}}+\frac{\ep^2}{r^6u^{\frac 3 2 +2\dec}}.
  \end{align*}
  Using \eqref{eq:flux-estimate-bb} and the coercivity of
  $\dds_2 \ddd_2$ (see Lemma 5.44 in
  \cite{klainermanKerrStabilitySmall2023}), we deduce for
  $k\le  N-6$,
  \begin{equation}
    \label{fluxestimates-foreta-xib}
    \int_{\Si_*}r^2u^{2+2\dec}\Big(|\dds_2\,\dk_*^k\eta|^2+|\dds_2\,\dk_*^k\xib|^2\Big)
    \les \mathcal{F}^N[\aa, \underline{\mathfrak{b}}, \qf^\F]+\ep_0^2.
  \end{equation}

  {\bf{Step 6: Estimate for $\eta_{\ell=1}$ and $\xib_{\ell=1}$.}}
  Using the elliptic identity \eqref{eq:elliptic-identity-diveta}, we can estimate
  \begin{align*} 
    \abs*{\int_S\div\eta\Jp}
    \les{}&r\abs*{\int_S \widehat{\rho}\Jp}
            +\abs*{\int_S\div \bbF  \Jp }+ r\dk^{\le 1}\Gamma_g  + r^3\dk^{\le 1}(\Gamma_b\cdot\Gamma_g)\\
    \les{}&\frac{\ep_0}{u^{1+2\dec}}+\frac{\ep_0}{u^{1+\dec}}+ \frac{\ep}{ru^{\frac 1 2 +\dec}} + \frac{\ep^2}{u^{\frac 3 2 +2\dec}}\les \frac{\ep_0}{u^{1+\dec}},
  \end{align*}
  where we used Proposition \ref{prop:control.ell=1modes-Si} to bound $\widehat{\rho}_{\ell=1}$, Step 3 above to bound $(\div\bbF)_{\ell=1}$, the bootstrap assumptions  \eqref{decayGagGabM3} and the dominance condition \eqref{dominanceM3}. Similarly for the commuted identity with $\nu^k$.

  Also recall the GCM condition \eqref{eq:Si_*-GCM2-xib}
  $(\div\xib)_{\ell=1} =0$, from which in view of Lemma
  \ref{Corr:nuSof integrals} we infer 
  \begin{align*}
    (\nu^k\div\xib)_{\ell=1} = r^{-2}\dk_*^{\leq k-1}\Gab+\dk_*^{\leq k}(\Gab\c\Gab).
  \end{align*}
  Using the bootstrap assumptions \eqref{decayGagGabM3} and the
  dominance condition \eqref{dominanceM3}, we infer for
  $k\le  N$,
  \begin{align}\label{etaxibell=1est}
    |(\nu^k\div\xib)_{\ell=1}|\les\frac{\ep}{r^3u^{1+\dec}}+\frac{\ep^2}{r^2u^{2+2\dec}}\les\frac{\ep_0}{r^2u^{2+2\dec}}.
  \end{align}
  Using the null structure equations $\curl\eta=r^{-1}\Gag+\Gab\c\Gag$ and $\curl\xib=\Gab\c\Gab$ and Lemma \ref{Lemma:Commutation-Si_*}, we have for $k\le  N$,
  \begin{align}
    \begin{split}\label{ell=2etaxib}
      |(\ddd_1\nab_\nu^k\eta)_{\ell=1}|+|(\ddd_1\nab_\nu^k\xib)_{\ell=1}|
      &\les|(\nu^k\div\eta)_{\ell=1}|+|(\nu^k\div\xib)_{\ell=1}|+\frac{\ep}{r^3u^{\frac 1 2 +\dec}}+\frac{\ep^2}{r^2u^{2+2\dec}}\\
      &\les\frac{\ep_0}{r^2u^{1+\dec}}+\frac{\ep_0}{r^2u^{2+2\dec}}+\frac{\ep_0}{r^2u^{2+2\dec}}\les\frac{\ep_0}{r^2u^{1+\dec}},
    \end{split}
  \end{align}
  where we used once again the bootstrap assumptions \eqref{decayGagGabM3}, the dominance condition \eqref{dominanceM3} and \eqref{etaxibell=1est}. Using Lemma \ref{prop:2D-Hodge1} and \eqref{fluxestimates-foreta-xib} we deduce for $k\le  N-6$
  \begin{equation}
    \label{Estimates:Flux-eta-xib}
    \int_{\Si_*}u^{2+2\dec}\Big(|\dk_*^k\eta|^2+|\dk_*^k\xib|^2\Big)
    \lesssim \mathcal{F}^N[\aa, \underline{\mathfrak{b}}, \qf^\F]+\ep_0^2.
  \end{equation}

  {\bf{Step 7: Estimate for $\yc$, $\zc$ and $\widecheck{b_*}$.}}
  From \eqref{byz} and Lemma \ref{Lemma:Commutation-Si_*}, we infer
  \begin{align*}
    \nab\dk_*^k \widecheck{y} = -\dk_*^k\xib +\dk_*^k\eta+\dk_*^{\leq k}\Gag+r\dk_*^{\leq k}(\Gab\c\Gab),\qquad \nab\dk_*^k \widecheck{z} = -2\dk_*^k\eta+\dk_*^{\leq k}\Gag+r\dk_*^{\leq k}(\Gab\c\Gab),
  \end{align*}
  Using the bootstrap assumptions \eqref{decayGagGabM3} and the
  dominance condition \eqref{dominanceM3}, we obtain
  \begin{align*}
    |\nab\dk_*^k \widecheck{y}|+|\nab\dk_*^k \widecheck{z}| \les{}& |\dk_*^k\xib| +|\dk_*^k\eta|+\frac{\ep}{r^2u^{\frac{1}{2}+\dec}}+\frac{\ep^2}{ru^{2+2\dec}}\\
    \les{}& |\dk_*^k\xib| +|\dk_*^k\eta|+\frac{\ep_0}{ru^{\frac{3}{2}+2\dec}}.
  \end{align*} 
  Squaring and integrating on $\Si_*$, we infer that
  \begin{equation*}
    \int_{\Si_*}u^{2+2\dec}\Big(|\nab\dk_*^k \widecheck{y}|^2+|\nab\dk_*^k \widecheck{z}|^2\Big) \les \ep_0^2+\int_{\Si_*}u^{2+2\dec}\Big(|\dk_*^k\eta|^2+|\dk_*^k\xib|^2\Big).
  \end{equation*}
  Together with \eqref{Estimates:Flux-eta-xib}, this yields, for $k\le  N-6$,
  \begin{equation*}
    \int_{\Si_*}u^{2+2\dec}\Big(|\nab\dk_*^k \widecheck{y}|^2+|\nab\dk_*^k \widecheck{z}|^2\Big) \les  \mathcal{F}^N[\aa, \underline{\mathfrak{b}}, \qf^\F]+\ep_0^2.
  \end{equation*}
  Since $b_*=-y-z$, we deduce, for $k\le  N-6$,
  \begin{equation}\label{eq:intermediatefluxesitmateaveragefreexyb}
    \int_{\Si_*}u^{2+2\dec}\Big(|\nab\dk_*^k \widecheck{y}|^2+|\nab\dk_*^k \widecheck{z}|^2+|\nab\dk_*^k \widecheck{b_*}|^2\Big) \les \mathcal{F}^N[\aa, \underline{\mathfrak{b}}, \qf^\F]+\ep_0^2.
  \end{equation}

  In view of \eqref{eq:intermediatefluxesitmateaveragefreexyb}, it remains to estimate the averages $\ov{\nu^k(\widecheck{y})}$, $\ov{\nu^k(\widecheck{z})}$ and $\ov{\nu^k(\widecheck{b_*})}$, whose proof follows identically the one in Step 7 of Proposition 5.42 in \cite{klainermanKerrStabilitySmall2023}. 
  Together with \eqref{eq:intermediatefluxesitmateaveragefreexyb}, and using a Poincar\'e inequality, we infer, for $k\le  N-7$,
  \begin{equation}\label{eq:fluxesitmateforxyb}
    \int_{\Si_*}r^{-2}u^{2+2\dec}\left(\left|\dk_*^k(\widecheck{y})\right|^2+\left|\dk_*^k(\widecheck{z})\right|^2+\left|\dk_*^k(\widecheck{b_*})\right|^2\right)
    \les \mathcal{F}^N[\aa, \underline{\mathfrak{b}}, \qf^\F]+\ep_0^2.
  \end{equation}

  {\bf{Step 8: Estimate for $\ombc$.}}
  Finally, from the null structure equations for $\nab_4\trchc$ and $\nab_3\trchc$, we deduce
  \begin{align*}
    \ombc= -\frac{r}{2}\div \eta -\frac{1}{2r}\widecheck{y} +\Gag+r\Gab\c\Gab.
  \end{align*}
  Using the bootstrap assumptions \eqref{decayGagGabM3} and the dominance condition \eqref{dominanceM3}, we obtain that 
  \begin{align*}
    |\dk_*^k\ombc| &\les |\dk_*^{k+1}\eta|+r^{-1}|\dk_*^k\widecheck{y}|+\frac{\ep}{r^2u^{\frac{1}{2}+\dec}}+\frac{\ep^2}{ru^{2+2\dec}}\\
                   &\les |\dk_*^{k+1}\eta|+r^{-1}|\dk_*^k\widecheck{y}|+\frac{\ep_0}{ru^{\frac{3}{2}+2\dec}}.
  \end{align*}
  Squaring and integrating on $\Si_*$, we infer that
  \begin{equation*}
    \int_{\Si_*}u^{2+2\dec}|\dk_*^k\ombc|^2 \les \int_{\Si_*}u^{2+2\dec}|\dk_*^{k+1}\eta|^2+\int_{\Si_*}r^{-2}u^{2+2\dec}|\dk_*^k\widecheck{y}|^2+\ep_0^2.
  \end{equation*}
  Together with \eqref{Estimates:Flux-eta-xib} for $\eta$ and \eqref{eq:fluxesitmateforxyb} for $\widecheck{y}$, we infer, for $k\le  N-7$, 
  \begin{equation*}
    \int_{\Si_*}u^{2+2\dec}|\dk_*^k\ombc|^2
    \les \mathcal{F}^N[\aa, \underline{\mathfrak{b}}, \qf^\F]+\ep_0^2.
  \end{equation*}  
  This concludes the proof of Proposition \ref{Prop.Flux-bb-vthb-eta-xib}.
\end{proof}

The control of the fluxes for the $\Gab$ quantities allow us to have a
better control of the nonlinear terms
$\int \dk^{\leq k}( \Gag \c \Gab )$, that according to the bootstrap
assumptions \eqref{decayGagGabM3} would only decay as
$u^{-\frac 1 2 -2\dec}$. It is nevertheless necessary to improve this
decay in order to control $\bm{J}_{\ell=1}$ in the subsequent section.

\begin{corollary}\label{cor:improved-bounds-nonlinear-=terms}
  Any nonlinear term of the form $\Gag \c \Gab$ satisfies on $\Si_*$
  for $k\le  N-9$:
  \begin{equation}
    \label{eq:improved-bounds-nonlinear-=terms}
    \int_u^{u_*} r^3  | \dk^{\leq k}( \Gag \c \Gab )  |\les \frac{\ep}{ u^{ 1 +2\dec}}\sqrt{\mathcal{F}^N[\aa, \underline{\mathfrak{b}}, \qf^\F]+\ep_0^2}.
  \end{equation}
\end{corollary}
\begin{proof}
  Using the bootstrap assumptions \eqref{decayGagGabM3} for the $\Gag$
  we have for $k\le N$
  \begin{align*}
    | \dk^{\leq k}( \Gag \c \Gab )  |\les  \frac{\ep}{r^2 u^{\frac 1 2+\dec}}|\dk^{\leq k} \Gab|.
  \end{align*}
  By integration in $u$ and using Sobolev and H\"older inequalities,
  we obtain 
  \begin{align*}
    \int_u^{u_*} r^3  | \dk^{\leq k}( \Gag \c \Gab )  |
    &\les \int_u^{u_*}\frac{\ep}{ u'^{\frac 1 2+\dec}}\|\dk^{\leq k+2} \Gab\|_{L^2(S)} \\
    &\les \ep \left(\int_u^{u_*}\frac{du'}{u'^{3+4\dec}}\right)^{\frac{1}{2}}\left(\int_u^{u_*}u'^{2+2\dec}\|\dk^{\leq k+2} \Gab\|^2_{L^2(S)} du'\right)^{\frac{1}{2}} \\
    &\les \frac{\ep}{ u^{ 1 +2\dec}}\left(\int_{\Si_*}u^{2+2\dec}    \big|\dk_*^ {k+2}\Ga_b\big|^2 \right)^{\frac 1 2 }.
  \end{align*}
  Using Proposition \ref{Prop.Flux-bb-vthb-eta-xib}, we infer the
  stated inequality in \eqref{eq:improved-bounds-nonlinear-=terms}.
\end{proof}

\subsection{Improved control of some \texorpdfstring{$\ell=1$}{} modes}\label{sec:l=1-2}

We now complete the improved decay for the $\ell=1$ modes of
Proposition \ref{prop:control.ell=1modes-Si} with the remaining
quantities.

\begin{proposition}\label{prop:control.ell=1modes-Si-2}
  Suppose $\mathcal{F}^N[\aa, \underline{\mathfrak{b}}, \qf^\F]  \;\les\; \ep^2$. Then the following estimates hold true on $\Si_*$: 
  \begin{align*}
r\left|\bm{J}_{\ell=1,\pm}\right|
    +r\left|\bm{J}_{\ell=1,0}-\frac{2aM}{r^5}\right|+\left|\bm{Y}_{\ell=1,\pm}\right|
    +\left|\bm{Y}_{\ell=1,0}-\frac{2aM}{r^4}\right|
    &\les\frac{\sqrt{\mathcal{F}^N[\aa, \underline{\mathfrak{b}}, \qf^\F]+\ep_0^2}}{r^4u^{1+\dec}}, \\
    \left|\dual\widehat{\rho}_{\ell=1,\pm}\right|
    +\left|\dual \widehat{\rho}_{\ell=1,0}-\frac{2aM}{r^4}\right|
    &\les\frac{\ep_0}{r^3u^{2+2\dec}}.
  \end{align*}
\end{proposition}

\begin{proof}
  We make the following local bootstrap assumptions on $\Si_*$:
  \begin{align}
    \label{eq:BAlocal-Si_*-2}
    \begin{split}
|\bm{J}_{\ell=1,\pm}|+\left|\bm{J}_{\ell=1,0}-\frac{2aM}{r^5}\right|
      &\leq\frac{\ep}{r^5u^{1+\dec}}.
    \end{split}
  \end{align}
  Notice that these assumptions are stronger than the ones implied by the bootstrap assumptions \eqref{decayGagGabM3} and will be recovered in the process of the proof (in Step 3). In order to improve the local bootstrap assumptions we need to assume \eqref{eq:assumption-FFk}.

{\bf{Step 1: Preliminary estimate for $\bm{Y}_{\ell=1}$.}}  Differentiating with respect to $\curl$ and projecting on the $\ell=1$
  modes the Codazzi equation for $\hch$, we infer 
  \begin{align*}
    (\curl\ddd_2\hch)_{\ell=1}=\frac{1}{r}(\curl\ze)_{\ell=1}-(\curl\beta)_{\ell=1} +\frac{Q}{r^2}(\curl\bF)_{\ell=1} +r^{-1}\dkb^{\leq 1}(\Gag\c\Gag).
  \end{align*}
  Recalling the definitions \eqref{eqn:C-J-renormalization} of $\bm{J}$
  and \eqref{eq:zeta-renormalization} of $\bm{Y}$ we deduce
  \begin{align*}
    (\curl\ddd_2\hch)_{\ell=1}
    =\frac{1}{r}\bm{Y}_{\ell=1}
   -\bm{J}_{\ell=1}
    +r^{-1}\dkb^{\leq 1}(\Gag\c\Gag).
  \end{align*}
  Using as in Proposition \ref{prop:control.ell=1modes-Si} that $(\curl\ddd_2\hch)_{\ell=1}=r^{-1}\Gag\c\dkb^{\leq 3}\Gag$
  we deduce
  \begin{align}\label{eq:relation-Y-J} 
    \bm{Y}_{\ell=1}
    =r\bm{J}_{\ell=1}
    +\dkb^{\leq 3}(\Gag\c\Gag),
  \end{align}
  and thus using the bootstrap assumptions \eqref{decayGagGabM3} and
  \eqref{eq:BAlocal-Si_*}, we have
  \begin{align}
    \label{eq:control.ell=1modes-Si-Step1-2}
\left|\bm{Y}_{\ell=1,\pm}\right|+\left|\bm{Y}_{\ell=1,0}-\frac{2aM}{r^4}\right|\les\frac{\ep}{r^4u^{1+\dec}}. 
  \end{align}

     {\bf{Step 2: Estimate for $\dual \widehat{\rho}_{\ell=1}$.}} Using that $\bm{Y}=\widehat{\dual\rho}$, we infer from \eqref{eq:control.ell=1modes-Si-Step1-2} and the dominance condition \eqref{dominanceM3},
  \begin{equation}
    \label{eq:control.ell=1modes-Si-Step5}
    \begin{split}
\abs*{\left( \dual\widehat{\rho} \right)_{\ell=1,\pm}}
           + \abs*{\left( \dual\widehat{\rho} \right)_{\ell=1,0} - \frac{2aM}{r^4}}   \les\frac{\ep}{r^4u^{1+\dec}}
           \les\frac{\ep_0}{r^3u^{2+2\dec}}.
    \end{split}        
  \end{equation}

     {\bf{Step 3: Estimate for $\bm{J}_{\ell=1}$.}} 
  Recall from Corollary \ref{COROFLEMMA:TRANSPORT.ALONGSI_STAR1} that
  we have along $\Si_*$, for $p=0,+,-$,
  \begin{align*}
    \nu\left( \int_S\bm{J}\Jp \right)
    ={}& O(r^{-1})\int_S\bm{J}\Jp+ h_2,
  \end{align*}
   where the scalar function $h_2$ is given by 
  \begin{equation*}
    \begin{split}
      h_2={}& O(r^{-2})\int_S\dual\widehat{\rho}\Jp
         + r^{-3}\dk^{\le 1}\Gamma_g
         + \dk^{\le 1}(\Gamma_b\cdot\Gamma_g).
    \end{split}    
  \end{equation*}
   In view of \eqref{eq:control.ell=1modes-Si-Step5}, the bootstrap assumptions \eqref{decayGagGabM3} and the dominance condition \eqref{dominanceM3} we obtain
  \begin{align*}
    |h_2|\les\frac{\ep_0}{r^3u^{2+2\dec}}+ \frac{\epsilon}{r^5u^{\frac 1 2 +\dec}}
    +\dk^{\le 1}(\Gamma_b\cdot\Gamma_g)
    \les\frac{\ep_0}{r^3u^{2+2\dec}}+\dk^{\le 1}(\Gamma_b\cdot\Gamma_g).
  \end{align*}
  Thus, using Corollary \ref{cor:improved-bounds-nonlinear-=terms} we infer 
  \begin{align*}
    \int_u^{u_*}r^3|h_2|\les \frac{\ep_0}{u^{1+2\dec}}+    \int_u^{u_*}r^3\dk^{\le 1}(\Gamma_b\cdot\Gamma_g)
    \les\frac{1}{u^{1+2\dec}}\sqrt{\mathcal{F}^N[\aa, \underline{\mathfrak{b}}, \qf^\F]+\ep_0^2}.
  \end{align*}
  Applying Lemma \ref{evolutionlemmaSi*}, we
  deduce for $p=\pm$, since we have $\bm{J}_{\ell=1,\pm}=0$ on $S_*$,
  \begin{align}\label{eq:control.ell=1modes-Si-Step7:intermediate}
    \left|\bm{J}_{\ell=1,\pm}\right|\les\frac{\sqrt{\mathcal{F}^N[\aa, \underline{\mathfrak{b}}, \qf^\F]+\ep_0^2}}{r^5u^{1+\dec}}.
  \end{align}
  Next, we focus on the case $p=0$. Using that $\nu(r)=-2+r\Gab$, we rewrite the above transport
  equation as 
  \begin{align*}
    \nu\left(r^3\int_S\bm{J} J^{(0)}\right)
    &= r^3\nu\left(\int_S\bm{J} J^{(0)}\right) -6r^2\int_S\bm{J} J^{(0)}+r^5\Gab\bm{J}_{\ell=1,0}\\
    ={}&-\frac{2}{r}r^3(1+ O(r^{-1}))\int_S\bm{J} J^{(0)}
        +2r(1+ O(r^{-1}))\int_S\dual\widehat{\rho}J^{(0)} + \dk^{\le 1}\Gamma_g
         +r^3\dkb^{\leq 1}(\Gab\c\Gag)
  \end{align*}
  which we rewrite as
  \begin{align}\label{curlbetaevolution}
    \nu\left(r^3\int_S\bm{J} J^{(0)} -8\pi aM\right)
    =O(r^{-1})\left(r^3\int_S\bm{J} J^{(0)}-8\pi aM\right)+h_3,
  \end{align}
  where the scalar function $h_3$ is given by 
  \begin{align*}
    h_3 ={}& O(r^3)\left(\dual\widehat{\rho}_{\ell=1,0} -\frac{2aM}{r^4}\right)
             +O(r^3)(\bm{J} J^{(0)})_{\ell=1,0}+O(r^2)\dual\widehat{\rho}_{\ell=1,0}
             \\
             &+ \dk^{\le 1}\Gamma_g
         +r^3\dkb^{\leq 1}(\Gab\c\Gag).
  \end{align*}
  In view of \eqref{eq:control.ell=1modes-Si-Step5} and
  \eqref{eq:BAlocal-Si_*-2}, the bootstrap assumptions \ref{decayGagGabM3} and the dominance condition \eqref{dominanceM3},
  \begin{align*}
    |h_3|\les \frac{\ep_0}{u^{2+2\dec}}+\frac{1}{r^2}
    + \frac{\epsilon}{r^2u^{\frac 1 2 +\dec}}+r^3\dkb^{\leq 1}(\Gab\c\Gag)
    \les\frac{\ep_0}{u^{2+2\dec}}+r^3\dkb^{\leq 1}(\Gab\c\Gag).
  \end{align*}
Thus, we infer using Corollary \ref{cor:improved-bounds-nonlinear-=terms}
  \begin{align*}
    \int_u^{u_*}|h_3|\les \frac{\ep_0}{u^{1+2\dec}}+    \int_u^{u_*}r^3\dk^{\le 1}(\Gamma_b\cdot\Gamma_g)
    \les\frac{\sqrt{\mathcal{F}^N[\aa, \underline{\mathfrak{b}}, \qf^\F]+\ep_0^2}}{u^{1+2\dec}}.
  \end{align*}
  Applying Lemma \ref{evolutionlemmaSi*} to \eqref{curlbetaevolution}
  and recalling that there holds
  $\bm{J}_{\ell=1,0}=\frac{2aM}{r^5}$ on $S_*$, we deduce
  \begin{equation*}
    r^3\left|\int_S\bm{J} J^{(0)}-\frac{8\pi aM}{r^3}\right|
    \les\frac{\sqrt{\mathcal{F}^N[\aa, \underline{\mathfrak{b}}, \qf^\F]+\ep_0^2}}{u^{1+2\dec}},
  \end{equation*}
  and hence, together with
  \eqref{eq:control.ell=1modes-Si-Step7:intermediate}, we have
  obtained
  \begin{equation}
    \label{eq:control.ell=1modes-Si-Step7}
    \left|\bm{J}_{\ell=1,\pm}\right|
    +\left|\bm{J}_{\ell=1,0}-\frac{2aM}{r^5}\right|
    \les \frac{\sqrt{\mathcal{F}^N[\aa, \underline{\mathfrak{b}}, \qf^\F]+\ep_0^2}}{r^5u^{1+\dec}}, 
  \end{equation}
  which, because of assumption \eqref{eq:assumption-FFk},  improves the local bootstrap assumption \eqref{eq:BAlocal-Si_*} for $\bm{J}$.

  {\bf{Step 4: Estimate for $(\bm{Z}, \bm{Y})_{\ell=1}$.}}  Using
  \eqref{eq:relation-Y-J} and \eqref{eq:control.ell=1modes-Si-Step7}
  we conclude
  \begin{equation*}
\left|\bm{Y}_{\ell=1,\pm}\right|+\left|\bm{Y}_{\ell=1,0}-\frac{2aM}{r^4}\right|\les\frac{\sqrt{\mathcal{F}^N[\aa, \underline{\mathfrak{b}}, \qf^\F]+\ep_0^2}}{r^4u^{1+\dec}},  
  \end{equation*}
  which improves the preliminary estimate \eqref{eq:control.ell=1modes-Si-Step1}.
  This concludes the proof of Proposition \ref{prop:control.ell=1modes-Si-2}.
\end{proof}

\subsection{Control of \texorpdfstring{$\Gamma_g$}{} quantities}\label{sec::improvementofdecaybootassonSigmastar}

In this section we obtain control of the quantities in the set
$\Ga_g$, defined in \eqref{eq:defintionofGagforChapter5glsdfiuhgs}.
Notice that $\trchc=0$ in view of our GCM conditions. Moreover, the
estimate is trivially verified\footnote{In the context of linear
  stability of Reissner-Nordstr\"om, the estimate for $\a$ can in fact
  be deduced from the ones for $\mathfrak{f}$ and $\mathfrak{b}$.} for
$\a \in r^{-1}\Gag$ in view of the definition of
$\mathcal{G}^k[\a, \mathfrak{b}, \mathfrak{f},\pf, \qf^\F]$ in
\eqref{eq;def-ee1}, while the improved estimate for $\nab_\nu \a$ is
obtained by the Bianchi identity for $\a$.

    {\bf{Step 1: Estimate for $\trchbc$ and $\muc$.}} 
Using the GCM conditions \eqref{eq:Si_*-GCM1}, we infer from Lemma \ref{JpGag} and Corollary \ref{cor:Cb0CbpM0MareGagandrm1Gag} that
\begin{align*}
\dds_2\dds_1\trchbc&=\sum_p\underline{C}_p\dds_2\dds_1\Jp=r^{-1}\dkb^{\leq 3}(\Gag\c\Gag),\\
\dds_2\dds_1\muc&=\sum_pM_p\dds_2\dds_1\Jp=r^{-2}\dkb^{\leq 3}(\Gag\c\Gag).
\end{align*}
which yields for $k\le  N-3$
\begin{align*}
\|\dds_2\dds_1\trchbc\|_{\hk_k(S)} &\les\|\dkb^{\leq k+3}(\Gag\c\Gag)\|_{L^\infty(S)}\les\frac{\ep^2}{r^4u^{1+2\dec}},\\ 
\|\dds_2\dds_1\muc\|_{\hk_k(S)} &\les r^{-1}\|\dkb^{\leq k+3}(\Gag\c\Gag)\|_{L^\infty(S)}\les\frac{\ep^2}{r^5u^{1+2\dec}}.
\end{align*}
Together with Corollary \ref{prop:2D-Hodge4}, we infer from Propositions \ref{prop:control.ell=1modes-Si} and \ref{prop:controlofell=0modesonSigmastar} that for $k\le  N-3$
\begin{align*}
\|\trchbc\|_{\hk_{k+2}(S)}&\les\frac{\ep^2}{r^2u^{1+2\dec}}+r|(\trchbc)_{\ell=1}|+r|\ov{\trchbc}|\les\frac{\ep_0}{ru^{2+3\dec}}+\frac{\ep_0}{ru^{1+2\dec}}\les\frac{\ep_0}{ru^{1+2\dec}},\\ 
\|\muc\|_{\hk_{k+2}(S)} &\les\frac{\ep^2}{r^3u^{1+2\dec}}+r|\mu_{\ell=1}|+r|\ov{\muc}|\les\frac{\ep_0}{r^2u^{2+3\dec}}+\frac{\ep_0}{r^2u^{1+2\dec}}\les\frac{\ep_0}{r^2u^{1+2\dec}},
\end{align*}
where we used the dominance condition \eqref{dominanceM3}. Together with Sobolev, this implies, for $k\le  N-3$, 
\begin{align}\label{muccontrol}
|\dkb^k\trchbc|\les \frac{\ep_0}{r^2u^{1+2\dec}}, \qquad |\dkb^k\muc|\les \frac{\ep_0}{r^3u^{1+2\dec}},
\end{align}
which gives the first line of \eqref{eq:improved-decay-M3}.

{\bf{Step 2: Estimate for $(\rhoFc, \dual \rhoF)$.}} From the
definition of $\qf^\F$ in \eqref{eq:elliptic-identity-qfF}, we infer
  \begin{align*}
\|\dds_2\dds_1(\rhoFc, -\dual\rhoF)\|_{\hk_k(S)} &\les r^{-2}\|\dkb^{\leq k}\qk^\F\|_{L^\infty(S)}+r^{-4}+r^{-2}\|\dkb^{\leq k+2}\Gab\|_{L^\infty(S)}+\|\dkb^{\leq k+2}(\Gab\c\Gag)\|_{L^\infty(S)}.
\end{align*}
Using the bootstrap assumptions \eqref{decayGagGabM3} the dominance
condition \eqref{dominanceM3}, and the control of $\Gab$ obtained in
\eqref{Estimate:Flux-bb-vthb-eta-xib-weak}, we obtain for
$k\le  N-11$,
\begin{align*}
  \|\dds_2\dds_1(\rhoFc, -\dual\rhoF)\|_{\hk_k(S)}
  &\les r^{-2}\|\dkb^{\leq k}\qk^\F\|_{L^\infty(S)} +\frac{\ep_0}{r^3u^{1 +\dec}}+\frac{\sqrt{\mathcal{F}^N[\aa, \underline{\mathfrak{b}}, \qf^\F]+\ep_0^2}}{r^3u^{1+\dec}}\\
  &\les r^{-2}\|\dkb^{\leq k}\qk^\F\|_{L^\infty(S)}+\frac{\sqrt{\mathcal{F}^N[\aa, \underline{\mathfrak{b}}, \qf^\F]+\ep_0^2}}{r^3u^{1+\dec}}.
\end{align*}
 In view of Corollary \ref{prop:2D-Hodge4} and \eqref{eq:rhofc-ov-0}, we deduce, for $k\le  N-11$, 
 \begin{equation}
   \label{eq:intermediate-rhoFc-1}
   \begin{split}
     \|(\rhoFc, -\dual\rhoF)\|_{\hk_{k+2}(S)}
     \les{}& \|\dkb^{\leq k}\qk^\F\|_{L^\infty(S)}+\frac{\sqrt{\mathcal{F}^N[\aa, \underline{\mathfrak{b}}, \qf^\F]+\ep_0^2}}{ru^{1+\dec}}\\
     &+r|(\rhoFc, -\dual\rhoF)_{\ell=1}|.   
   \end{split}
\end{equation}

We now control the $\ell=1$ mode of $\rhoFc$ and $\dual\rhoF$. From the definitions \eqref{eq:rho-renormalization} and \eqref{eq:rhod-renormalization}, we have
\begin{align*}
   \rhoc&= \widehat{\rho}+r^{-2}\Gag+\Gag \c \Gab,    \\
   \dual \rho&=   \dual\widehat{\rho}+r^{-2}\Gag+\Gag \c \Gab.
\end{align*}  
From Proposition \ref{prop:control.ell=1modes-Si} and the bootstrap assumptions, we deduce
\begin{align*}
        \left|\left(\rhoc \right)_{\ell=1}\right|+\left|\left(\dual\rho\right)_{\ell=1,\pm}\right|
    +\left|\left(\dual \rho\right)_{\ell=1,0}-\frac{2aM}{r^4}\right|&\les\frac{\ep_0}{r^3u^{1+2\dec}}+\frac{\ep}{r^4u^{\frac 12 +\dec}}+\frac{\ep^2}{r^3 u^{1+2\dec}}\les \frac{\ep_0}{r^3u^{1+2\dec}},
\end{align*}
where we used the dominance condition. Finally, using \eqref{eq:elliptic-identity-pF} we have
\begin{align*}
  \ddd_1 \dds_1(\rhoF,\dual \rhoF)
  &=   r^{-2}\ddd_1\Re(\pf)+r\ddd_1\dds_1(\rho, \dual\rho) + O(r^{-5})
         +r^{-3} \dk^{\leq 1}\Ga_b       
         + r^{-1}\dk^{\le 1}(\Gamma_b\cdot\Gamma_g).
\end{align*}
Projecting the above identity to $\ell=1$ and using the bootstrap assumptions and dominance condition, Lemma \ref{prop:2D-Hodge1} and Sobolev,  and the control of $\Gab$ obtained in \eqref{Estimate:Flux-bb-vthb-eta-xib-weak} we deduce
\begin{align*}
  \abs*{(\rhoFc)}_{\ell=1}+|( \dual\rhoF)_{\ell=1} \big|
  &\les  r^{-1}\|\dkb^{\le3}\pf\|_{L^\infty(S)}+\frac{\ep_0}{r^2u^{1+2\dec}} + \frac{1}{r^3}
    +\frac{\sqrt{\mathcal{F}^N[\aa, \underline{\mathfrak{b}}, \qf^\F]+\ep_0^2}}{r^2u^{1+\dec}} \\
  &\les  r^{-1}\|\dkb^{\le3}\pf\|_{L^\infty(S)}
    +\frac{\sqrt{\mathcal{F}^N[\aa, \underline{\mathfrak{b}}, \qf^\F]+\ep_0^2}}{r^2u^{1+\dec}} .
\end{align*}
In view of Proposition \ref{prop:controlofell=0modesonSigmastar} and
Step 2, we infer from \eqref{eq:intermediate-rhoFc-1} for
$k\le  N-11$,
\begin{align}
  \label{rhoFcrhoFdcontrol}
  \|(\rhoFc, -\dual\rhoF)\|_{\hk_{k+2}(S)}
  &\les \|\dkb^{\leq k}\qk^\F\|_{L^\infty(S)}
    +\|\dkb^{\le 3}\pf\|_{L^\infty(S)}+\frac{\sqrt{\mathcal{F}^N[\aa, \underline{\mathfrak{b}}, \qf^\F]+\ep_0^2}}{ru^{1+\dec}}.
\end{align}
Together with Sobolev, this implies for $k\le  N-11$,
\begin{align}
  \label{eq:control-rhofc-fin}
  |\dkb^k\rhoFc|+|\dkb^k\dual \rhoF|
  \les \mathcal{G}^k[\a, \mathfrak{b}, \mathfrak{f},\pf, \qf^\F]
  +\frac{\sqrt{\mathcal{F}^N[\aa, \underline{\mathfrak{b}}, \qf^\F]+\ep_0^2}}{r^2u^{1+\dec}}.
\end{align}

  {\bf{Step 3: Estimate for $(\rhoc, \rhod)$.}}
From \eqref{eq:elliptic-identity-pF}, we infer that 
\begin{align*}
\|\dds_1(-\rhoc, \rhod)\|_{\hk_k(S)} \les{}& r^{-2}\|\dkb^{\leq k}\pf\|_{L^\infty(S)}+r^{-1}|\dkb^{\leq k}(\rhoF,\dual \rhoF)|\\
&+r^{-4}+r^{-2}\|\dkb^{\leq k+2}\Gab\|_{L^\infty(S)}+\|\dkb^{\leq k+2}(\Gab\c\Gag)\|_{L^\infty(S)}.
\end{align*}
Using the bootstrap assumptions and dominance condition,  and the control of $\Gab$ obtained in \eqref{Estimate:Flux-bb-vthb-eta-xib-weak} and \eqref{eq:control-rhofc-fin} we deduce
 for $k\le  N-11$,
\begin{align*}
  \|\dds_1(-\rhoc, \rhod)\|_{\hk_k(S)}&\les r^{-1}\mathcal{G}^k[\a, \mathfrak{b}, \mathfrak{f},\pf, \qf^\F]
                                        +\frac{\sqrt{\mathcal{F}^N[\aa, \underline{\mathfrak{b}}, \qf^\F]+\ep_0^2}}{r^3u^{1+\dec}}.
\end{align*}
In view  of Lemma \ref{prop:2D-Hodge1}, we deduce, for $k\le N-11$, 
\begin{align*}
  \|(-\rhoc, \rhod)\|_{\hk_{k+1}(S)}
  \les \mathcal{G}^k[\a, \mathfrak{b}, \mathfrak{f},\pf, \qf^\F]
  +\frac{\sqrt{\mathcal{F}^N[\aa, \underline{\mathfrak{b}}, \qf^\F]+\ep_0^2}}{r^2u^{1+\dec}}+r|(\ov{\rhoc},\ov{\rhod})|.
\end{align*}
In view of Proposition \ref{prop:controlofell=0modesonSigmastar}, we infer for $k\le N-11$, 
\begin{align}
  \label{rhocrhodcontrol}
  \|(\rhoc,\rhod)\|_{\hk_{k+2}(S)}
  \les \mathcal{G}^k[\a, \mathfrak{b}, \mathfrak{f},\pf, \qf^\F]
  +\frac{\sqrt{\mathcal{F}^N[\aa, \underline{\mathfrak{b}}, \qf^\F]
  +\ep_0^2}}{r^2u^{1+\dec}}.
\end{align}
Together with Sobolev, this implies for $k\le N-11$, 
\begin{align}
  |\dkb^k\rhoc|+ |\dkb^k\rhod|
  \les r^{-1}\mathcal{G}^k[\a, \mathfrak{b}, \mathfrak{f},\pf, \qf^\F]
  +\frac{\sqrt{\mathcal{F}^N[\aa, \underline{\mathfrak{b}}, \qf^\F]+\ep_0^2}}{r^3u^{1+\dec}}.
\end{align}

{\bf{Step 4: Estimate for $\ze$.}} Using that
$\ddd_1\ze=\left(-\muc-\rhoc,\rhod\right)+\Gab\c\Gag$, in view of
Lemma \ref{prop:2D-Hodge1}, we infer that
\begin{align*}
 \|\ze\|_{\hk_{k+1}   (S)} &\les  r\|\muc\|_{\hk_k(S)}+r\|\rhoc\|_{\hk_k(S)}+r\|\rhod\|_{\hk_k(S)}+r^2\|\dkb^{\leq k}(\Gab\c\Gag)\|_{L^\infty(S)}.
\end{align*}
Together with \eqref{muccontrol} and \eqref{rhocrhodcontrol} and the
bootstrap assumptions of Lemma \ref{decayGagGabM3}, we obtain, for
$k\le N-11$,
\begin{align}\label{zetacontrol}
  \|\ze\|_{\hk_{k+1}   (S)}
  &\les \frac{\ep_0}{r^2u^{1+2\dec}}
    +r\mathcal{G}^k[\a, \mathfrak{b}, \mathfrak{f},\pf, \qf^\F]
    +\frac{\sqrt{\mathcal{F}^N[\aa, \underline{\mathfrak{b}}, \qf^\F]+\ep_0^2}}{ru^{1+\dec}}+\frac{\ep^2}{ru^{\frac 3 2 +2\dec}}\\
  &\les r\mathcal{G}^k[\a, \mathfrak{b}, \mathfrak{f},\pf, \qf^\F]
    +\frac{\sqrt{\mathcal{F}^N[\aa, \underline{\mathfrak{b}}, \qf^\F]+\ep_0^2}}{ru^{1+\dec}}.
\end{align}
Together with Sobolev, this implies, for $k\le  N-11$,
\begin{align}
  |\dkb^k\ze|
  \les \mathcal{G}^k[\a, \mathfrak{b}, \mathfrak{f},\pf, \qf^\F]
  +\frac{\sqrt{\mathcal{F}^N[\aa, \underline{\mathfrak{b}}, \qf^\F]+\ep_0^2}}{r^2u^{1+\dec}}.
\end{align}

 {\bf{Step 5: Estimate for $\bF$.}} 
We first control the $\ell=1$ mode of $\ddd_1\bF$.
From the definition \eqref{eqn:C-J-renormalization} of $(\bm{C}, \bm{J})$ we have
\begin{align*}
      \ddd_1 \b  
  ={}& (\bm{C}, \bm{J})+r^{-3}\dk^{\leq 1} \Ga_g  .
\end{align*}
Then using Proposition \ref{prop:control.ell=1modes-Si-2}, we have
 \begin{align*}
       | (\ddd_1 \b)_{\ell=1}  | &\les |\bm{C}_{\ell=1}| +\left|\bm{J}_{\ell=1,\pm}\right|
    +\left|\bm{J}_{\ell=1,0}-\frac{2aM}{r^5}\right|+\frac{1}{r^5} + r^{-3}\dk^{\leq 1} \Ga_g \\
       &\les \frac{\sqrt{\mathcal{F}^N[\aa, \underline{\mathfrak{b}}, \qf^\F]+\ep_0^2}}{r^5u^{1+\dec}}+ \frac{1}{r^5} + \frac{\ep_0}{r^6u^{\frac 12 +\dec}} \les \frac{\ep_0}{r^4u^{1+\dec}}.
\end{align*}
Using \eqref{eq:def-bb-M3}, we have
\begin{align*}
6M\ddd_1\bF &= r^3 \ddd_1 \mathfrak{b}-2Qr\ddd_1\b+r^{-5}\dk^{\leq1}\Ga_g+ r^{-2}\dk^{\leq 1}(\Gag\c\Gag ).
\end{align*}
Using the above bound, the bootstrap assumptions \eqref{decayGagGabM3}
and dominance condition \eqref{dominanceM3}, we deduce that 
\begin{align*}
|(\ddd_1\bF)_{\ell=1}| &\les  r^3| \ddd_1 \mathfrak{b}|+r|(\ddd_1\b)_{\ell=1}|+r^{-5}|\dk^{\leq 1}\Ga_g|+ r^{-2}|\dk^{\leq 1}(\Gag\c\Gag )|\\
&\les  r^3| \ddd_1 \mathfrak{b}|+\frac{\ep_0}{r^3u^{1+\dec}} +\frac{\ep}{r^7u^{\frac 12 +\dec}}+ \frac{\ep_0}{r^6u^{1+2\dec}}\\
&\les r^3| \ddd_1 \mathfrak{b}|+\frac{\ep_0}{r^3u^{1+\dec}} .
\end{align*}
Using \eqref{eq:M3:def-ff}, we have
 \begin{align*}
     \dds_2 \bF=-\mathfrak{f} +O(r^{-2})\Ga_g +\Gag\c\Gag.
 \end{align*}
In view of Lemma \ref{prop:2D-Hodge1}, we infer that
\begin{align*}
\|\bF\|_{\hk_{k+1}  (S)}&\les  r\|\dkb^{\leq k}\mathfrak{f}\|_{L^2(S)}+\|\dkb^{\leq k}\Gag\|_{L^\infty(S)}  +r^2\|\dkb^{\leq k}(\Gag\c\Gag)\|_{L^\infty(S)}+r^2 \big| (\ddd_1\bF)_{\ell=1}\big|\\
&\les  r^2\|\dkb^{\leq k}\mathfrak{f}\|_{L^\infty(S)}+\|\dkb^{\leq k}\Gag\|_{L^\infty(S)}  +r^2\|\dkb^{\leq k}(\Gag\c\Gag)\|_{L^\infty(S)}+r^2 \big| (\ddd_1\bF)_{\ell=1}\big|.
\end{align*}
Using the bootstrap assumptions \eqref{decayGagGabM3}, the dominance condition \eqref{dominanceM3}, and the above
we obtain
\begin{align*}
\|\bF\|_{\hk_{k+1}  (S)} &\les  r^2\|\dkb^{\leq k}\mathfrak{f}\|_{L^\infty(S)}+ r^4\|\dkb^{\leq k+1}\mathfrak{b}\|_{L^\infty(S)}+\frac{\ep}{r^2u^{\frac 12 +\dec}}  +\frac{\ep}{r^2u^{1+2\dec}} + \frac{\ep_0}{ru^{1+\dec}}\\
&\les r^2\|\dkb^{\leq k}\mathfrak{f}\|_{L^\infty(S)}+ r^4\|\dkb^{\leq k+1}\mathfrak{b}\|_{L^\infty(S)}+ \frac{\ep_0}{ru^{1 +\dec}}.
\end{align*}
Together with Sobolev, this implies for $k\le  N$,
\begin{align*}
  |\dkb^k\bF|\les \mathcal{G}^k[\a, \mathfrak{b}, \mathfrak{f},\pf, \qf^\F]+ \frac{\ep_0}{r^2u^{1 +\dec}}.
\end{align*}

{\bf{Step 6: Estimate for $\b$.}}  Using \eqref{eq:def-bb-M3}, the
bootstrap assumptions \eqref{decayGagGabM3} and the dominance
condition \eqref{dominanceM3}, we immediately deduce that for $k\le N$,
 \begin{align}\label{bcontrol}
 \begin{split}
     |\dkb^k\b |&\les r^2  |\dkb^k\mathfrak{b}|+r^{-1}|\dkb^k\bF|+r^{-5}|\dkb^k\Ga_g|+ r^{-2}|\dkb^k(\Gag\c\Gag )|\\
&\les  \|\dkb^{\leq k}\mathfrak{f}\|_{L^\infty(S)}+ r^2\|\dkb^{\leq k+1}\mathfrak{b}\|_{L^\infty(S)}+ \frac{\ep_0}{r^3u^{1 +\dec}}+\frac{\ep}{r^7u^{\frac 12 +\dec}}+ \frac{\ep_0}{r^6u^{1+2\dec}}\\
&\les r^{-1} \mathcal{G}^k[\a, \mathfrak{b}, \mathfrak{f},\pf, \qf^\F]+ \frac{\ep_0}{r^3u^{1 +\dec}}.
 \end{split}
\end{align}

 {\bf{Step 7: Estimate for $\chih$.}} 
Using the Codazzi equation for $\hch$, from Lemma \ref{prop:2D-Hodge1}, we infer that
\begin{align*}
 \|\hch\|_{\hk_{k+1}   (S)} &\les  \|\ze\|_{\hk_k(S)}+r\|\b\|_{\hk_k(S)}+\|\dkb^{\leq k}\Gag\|_{L^\infty(S)}+r^2\|\dkb^{\leq k}(\Gag\c\Gag)\|_{L^\infty(S)}.
\end{align*}
Together with \eqref{zetacontrol} and \eqref{bcontrol} and the
bootstrap assumptions \eqref{decayGagGabM3}, and Sobolev, we obtain
for $3\le k\le  N-11$,
\begin{align}
  |\dkb^k\hch|\les \mathcal{G}^k[\a, \mathfrak{b}, \mathfrak{f},\pf, \qf^\F]
  +\frac{\sqrt{\mathcal{F}^N[\aa, \underline{\mathfrak{b}}, \qf^\F]
  +\ep_0^2}}{r^2u^{1+\dec}}.
\end{align}

 {\bf{Step 8: Estimate for $\nab_\nu\Gag$.}} 
From the Maxwell equations, null structure equations and Bianchi identities, one observes that all quantities in $\Gag$ verify schematically
\begin{align*}
  \nab_\nu\Gag &= r^{-1}\dkb^{\leq 1}\Gab+ r^{-1} \Gag+\Gab\c \Gab.
\end{align*}
Combining with the control of the $\Gab$ quantities obtained in
\eqref{Estimate:Flux-bb-vthb-eta-xib-weak}, the bootstrap assumptions
\eqref{decayGagGabM3} and dominance condition \eqref{dominanceM3}, we
infer for $k\le  N-9$,
\begin{align*}
  \big|\dk_*^{\leq k-1}\nab_\nu\Gag\big|&\les r^{-1}\mathcal{G}^k[\a, \mathfrak{b}, \mathfrak{f},\pf, \qf^\F]
  +\frac{\sqrt{\mathcal{F}^N[\aa, \underline{\mathfrak{b}}, \qf^\F]+\ep_0^2}}{r^2u^{1+\dec}}.
\end{align*}
Using the dominance condition and \eqref{eq:ba-on-Gk}, we deduce
\begin{align*}
  \big|\dk_*^{\leq k-1}\nab_\nu\Gag\big|&\les \frac{\ep_0}{u^{1+\dec}}\frac{\ep}{r^2u^{\frac 1 2+\dec}}
  +\frac{\sqrt{\mathcal{F}^N[\aa, \underline{\mathfrak{b}}, \qf^\F]+\ep_0^2}}{r^2u^{1+\dec}}\les \frac{\sqrt{\mathcal{F}^N[\aa, \underline{\mathfrak{b}}, \qf^\F]+\ep_0^2}}{r^2u^{1+\dec}}.
\end{align*}

{\bf{Step 9: Estimate for $\nab_\nu\b$ and $\nab_\nu \bF$.}}
In view
of the Bianchi identity for $\nab_3\b$ and the Maxwell equation for
$\nab_3\bF$, and using the fact that $\nu=e_3+b_*e_4$, we have
\begin{align*}
\nab_\nu\b = r^{-2}\dkb^{\leq 1}\Gag+r^{-3}\Ga_b+r^{-1} \dk^{\leq 1} (\Gab \c \Gag), \\
\nab_\nu\bF = r^{-1}\dkb^{\leq 1}\Gag+r^{-2}\Ga_b+\Gab \c \Gag,
\end{align*}
from which we infer for $3\le k\le  N-11$,
\begin{align*}
  r\big|\dk_*^{\leq k-1}\nab_\nu\b|+\big|\dk_*^{\leq k-1}\nab_\nu\bF|
  \les r^{-1}\mathcal{G}^k[\a, \mathfrak{b}, \mathfrak{f},\pf, \qf^\F]
  +\frac{\sqrt{\mathcal{F}^N[\aa, \underline{\mathfrak{b}}, \qf^\F]+\ep_0^2}}{r^3u^{1+\dec}}. 
\end{align*}
Similarly, applying another $\nab_\nu$ derivative one obtains
 \begin{align*}
 \nab^2_\nu\b = r^{-2}\dkb^{\leq 1}\nabla_\nu \Gag+r^{-3}\dkb^{\leq 2}\Gag+r^{-3}\dk^{\leq 1}\Ga_b+r^{-1} \dk^{\leq 2} (\Gab \c \Gag), \\
\nab^2_\nu\bF = r^{-1}\dkb^{\leq 1}\nabla_\nu\Gag+r^{-2}\dkb^{\leq 2}\Gag+r^{-2}\dk^{\leq 1}\Ga_b+\dk^{\leq 1}(\Gab \c \Gag).
\end{align*}
Using the improved decay for $\nabla_\nu \Gag$ of Step 8, we obtain
for $3\le k\le N-11$, 
\begin{align*}
  r\big|\dk_*^{\leq k-2} \nab^2_\nu\b|+\big| \dk_*^{\leq k-2}\nab^2_\nu\bF|
  \les \frac{\sqrt{\mathcal{F}^N[\aa, \underline{\mathfrak{b}}, \qf^\F]+\ep_0^2}}{r^3u^{1+\dec}}, 
\end{align*}
as stated.

This concludes the proof of Theorem \ref{prop:decayonSigamstarofallquantities}.

\subsection{Control of the canonical one-forms}
\label{sec:controlofJpandJkonSigmastar}

Here we collect a consequence of Theorem \ref{prop:decayonSigamstarofallquantities} regarding the control of the canonical one-forms $\Jp$.

\begin{corollary}\label{cor:phi-Jpo}
  We have the following estimate for $\phi$ on $\Si_*$ for $k\le  N-11$:
  \begin{align*}
    \|\dkb^{k}\phi\|_{L^\infty(S_*)}
    &\les r\mathcal{G}^k[\a, \mathfrak{b}, \mathfrak{f},\pf, \qf^\F]
      +\frac{\sqrt{\mathcal{F}^N[\aa, \underline{\mathfrak{b}}, \qf^\F]+\ep_0^2}}{r u^{1+\dec}}.
  \end{align*}
  The functions $\Jp$ verify the following properties on $\Si_*$:
  \begin{enumerate}
  \item We have
    \begin{align*}
      \begin{split}
        \int_{S}J^{(p)}&=r^3\mathcal{G}^0[\a, \mathfrak{b}, \mathfrak{f},\pf, \qf^\F]
                         +\frac{r\sqrt{\mathcal{F}^N[\aa, \underline{\mathfrak{b}}, \qf^\F]+\ep_0^2}}{ u^{1+\dec}}, \\
        \int_{S}J^{(p)}J^{(q)}&=\frac{4\pi}{3}r^2\de_{pq}+r^3\mathcal{G}^0[\a, \mathfrak{b}, \mathfrak{f},\pf, \qf^\F]
                                +\frac{r\sqrt{\mathcal{F}^N[\aa, \underline{\mathfrak{b}}, \qf^\F]+\ep_0^2}}{ u^{1+\dec}}.
      \end{split}
    \end{align*}
  \item We have
    \begin{align*}
      \abs*{\left(\triangle+\frac{2}{r^2}\right)\Jp}
      \lesssim r^{-1}\mathcal{G}^0[\a, \mathfrak{b}, \mathfrak{f},\pf, \qf^\F]
      +\frac{\sqrt{\mathcal{F}^N[\aa, \underline{\mathfrak{b}}, \qf^\F]+\ep_0^2}}{r^3 u^{1+\dec}},\\
      \big|\dds_2\dds_1\Jp \big|
      \les r^{-1}\mathcal{G}^1[\a, \mathfrak{b}, \mathfrak{f},\pf, \qf^\F]
      +\frac{\sqrt{\mathcal{F}^N[\aa, \underline{\mathfrak{b}}, \qf^\F]+\ep_0^2}}{r^3 u^{1+\dec}}.
    \end{align*}
    \item We have for all $k\le  N-12$
      \begin{align*}
        \left|\dk_*^{k}\Big[\widecheck{\nab\Jp}\Big]\right|
        &\les \mathcal{G}^k[\a, \mathfrak{b}, \mathfrak{f},\pf, \qf^\F]
          +\frac{\sqrt{\mathcal{F}^N[\aa, \underline{\mathfrak{b}}, \qf^\F]+\ep_0^2}}{r^2 u^{\frac1 2+\dec}}.
  \end{align*}
  \end{enumerate}
\end{corollary}
\begin{proof}
  The proof of the first five estimates follows from the one of Lemma \ref{JpGag} upon replacing the bootstrap assumptions for $\Kc$ with the estimated control for $\Kc \in r^{-1}\Gag$ given by Theorem \ref{prop:decayonSigamstarofallquantities}.

  For the last estimate, using that on $S_*$ (see Lemma 5.61 in \cite{klainermanKerrStabilitySmall2023}) we have
  \begin{align*}
    \widecheck{\nab J^{(0)}}=-\frac{1}{r} (e^{-\phi}-1) \dual f_0, \qquad \widecheck{\nab J^{(+)}}=\frac{1}{r} (e^{-\phi}-1) f_+, \qquad \widecheck{\nab J^{(-)}}=\frac{1}{r} (e^{-\phi}-1) f_-.
  \end{align*}
  and using the control of $\phi$ above we deduce on $S_*$
  \begin{align*}
    |\dkb^k \widecheck{\nab\Jp} | \les \left|\dk_*^{k}\Big[\widecheck{\nab\Jp}\Big]\right|&\les \mathcal{G}^k[\a, \mathfrak{b}, \mathfrak{f},\pf, \qf^\F]
                                                                                            +\frac{\sqrt{\mathcal{F}^N[\aa, \underline{\mathfrak{b}}, \qf^\F]+\ep_0^2}}{r^2 u^{1+\dec}}.
  \end{align*}
  On $\Si_*$ we have (see also Lemma 5.66 in \cite{klainermanKerrStabilitySmall2023})
  \begin{align}\label{eq:nab-nu-nab-Jp}
    \nabla_\nu \big[ r \widecheck{\nab\Jp}\big] = \Gab \c \dkb^{\leq 1} \Jp,
  \end{align}
  and integrating it along $\Si_*$ using Lemma
  \ref{evolutionlemmaSi*}, we deduce for $k\le N-11$,
  \begin{align*}
    r \big| \dkb^k \big[r \widecheck{\nab\Jp}\big] \big|
    &\les   r \big| \dkb^k \big[r \widecheck{\nab\Jp}\big] \big|_{L^\infty (S_*)} + \int_u^{u_*} r |\dkb^{\leq k} \Gab |\\
    &\les   r \big| \dkb^k \big[r \widecheck{\nab\Jp}\big] \big|_{L^\infty (S_*)} +  \left(\int_u^{u_*}\frac{du'}{u'^{2+2\dec}}\right)^{\frac{1}{2}}\left(\int_u^{u_*}u'^{2+2\dec}\|\dk^{\leq k+2} \Gab\|^2_{L^2(S)} du'\right)^{\frac{1}{2}}\\
    &\les r^2\mathcal{G}^k[\a, \mathfrak{b}, \mathfrak{f},\pf, \qf^\F]
      +\frac{\sqrt{\mathcal{F}^N[\aa, \underline{\mathfrak{b}}, \qf^\F]+\ep_0^2}}{ u^{\frac 1 2+\dec}},
  \end{align*}
  where we used Sobolev, H\"older inequalities and Proposition
  \ref{Prop.Flux-bb-vthb-eta-xib} for the control of the fluxes of
  $\Gab$.

  We finally infer that on $\Si_*$ for $k\le N-11$,
  \begin{align*}
    |\dkb^k \widecheck{\nab\Jp} |
    &\les \mathcal{G}^k[\a, \mathfrak{b}, \mathfrak{f},\pf, \qf^\F]
      +\frac{\sqrt{\mathcal{F}^N[\aa, \underline{\mathfrak{b}}, \qf^\F]+\ep_0^2}}{r^2 u^{\frac 1 2 +\dec}}.
  \end{align*}
  from which, using again \eqref{eq:nab-nu-nab-Jp}, we deduce the stated.
\end{proof}

\begin{corollary}\label{corollary:C-J-f_0}
    We have on $\Si_*$:
    \begin{align*}
\Big|\big((\bm{C}, \bm{J})  -\frac{3aM}{r^4}\ddd_1f_0\big)_{\ell=1}\Big|&\les \frac{\sqrt{\mathcal{F}^N[\aa, \underline{\mathfrak{b}}, \qf^\F]+\ep_0^2}}{r^5u^{1+\dec}}.
\end{align*}
\end{corollary}
\begin{proof}
   Recall that
\begin{align*}
    \ddd_1 f_0=(\div f_0, \curl f_0)=(\div f_0, \frac{2}{r}\cos\th+\widecheck{\curl(f_0)})=\left(\div f_0, \frac{2}{r}J^{(0)}+\widecheck{\curl(f_0)}\right).
\end{align*}
 In particular we can write 
\begin{align*}
\left((\bm{C}, \bm{J})  -\frac{3aM}{r^4}\ddd_1f_0\right)_{\ell=1}=\left(\bm{C}, \bm{J}-\frac{6aM}{r^5} J^{(0)}\right)_{\ell=1}+r^{-4}\mathcal{G}^0[\a, \mathfrak{b}, \mathfrak{f},\pf, \qf^\F]
      +\frac{\sqrt{\mathcal{F}^N[\aa, \underline{\mathfrak{b}}, \qf^\F]+\ep_0^2}}{r^6 u^{\frac 1 2 +\dec}},
\end{align*}
and therefore from Proposition \ref{prop:control.ell=1modes-Si} and Proposition \ref{prop:control.ell=1modes-Si-2},
\begin{align*}
\Big|\big((\bm{C}, \bm{J})  -\frac{3aM}{r^4}\ddd_1f_0)\big)_{\ell=1}\Big|\les{}&  \left|\bm{J}_{\ell=1,\pm}\right|
    +\left|\bm{J}_{\ell=1,0}-\frac{2aM}{r^5}\right|\\
    &+r^{-4}\mathcal{G}^0[\a, \mathfrak{b}, \mathfrak{f},\pf, \qf^\F]
      +\frac{\sqrt{\mathcal{F}^N[\aa, \underline{\mathfrak{b}}, \qf^\F]+\ep_0^2}}{r^6 u^{\frac 1 2 +\dec}}\\
    \les{}&\frac{\sqrt{\mathcal{F}^N[\aa, \underline{\mathfrak{b}}, \qf^\F]+\ep_0^2}}{r^5u^{1+\dec}}
\end{align*}
where we used Corollary \ref{cor:phi-Jpo}, and the dominance condition.
\end{proof}

\appendix

\section{Proof of Proposition \ref{lemma:improved-transport-along-Sigma-star}}
\label{appendix:lemma:improved-transport-along-Sigma-star} 


Here we derive the transport equation of the renormalized quantities.

\paragraph{Transport equation of \texorpdfstring{$\widehat{\rho}$}{}.}

From Proposition \ref{Prop.NullStr+Bianchi-lastslice}, we deduce
\begin{align*}
  \nabla_3\left(\frac{Q}{r^2}\rhoFc\right)
  ={}& \frac{8Q\Upsilon}{r^3}\rhoFc
       - \frac{2Q}{r^2}\div\bbF
       - \frac{2Q^2}{r^4}\kabc
       + \frac{4Q^2}{r^5}\yc
       + r^{-1}\Gamma_g\cdot\Gamma_b\\
       &= r^{-3} \Gab + r^{-1}\Gamma_g\cdot\Gamma_b, \\
  \nab_3(\chih\c\chibh)
  ={}& \chih\c\left( \frac{2\Up}{r}\chibh -\aa -\left( \frac{2M}{r^2} - \frac{2Q^2}{r^3} \right)\chibh +\nab\hot \xib +\Ga_b\c \Ga_b\right)\\
     &+\left(\frac{\Up}{r} \chih+\nab\hot \eta -\frac{1}{r} \chibh+\left( \frac{2M}{r^2}-\frac{2Q^2}{r^3} \right)\chih+\Ga_b\c \Ga_b\right)\c\chibh\\
  ={}&-\chih \c\aa+  r^{-1}\dkb^{\le  1}( \Ga_b \c \Ga_b),
\end{align*}
Combining the above with the equation for $\nab_3\rhoc$ in Proposition \ref{Prop.NullStr+Bianchi-lastslice},  we obtain for $\widehat{\rho}=\rhoc -\frac{2Q}{r^2}\rhoFc - \frac{1}{2}\chih\c\chibh$,
\begin{align*}
  \nab_3\widehat{\rho}
  ={}&\frac{3\Up}{r} \rhoc   -\div\bb-
          \frac{1}{2}\chih\c\aa+
          \frac{1}{2}\chih\c\aa
         +r^{-3} \Gab 
    +r^{-1}\dkb^{\le  1}( \Ga_b \c \Ga_b)\\
        ={}&\frac 3 r \widehat{\rho} - \underline{\bm{C}}+ r^{-3} \dkb^{\leq 1}\Gab
       +r^{-1}\dkb^{\le  1}( \Ga_b \c \Ga_b). 
\end{align*}
Using that $\nabla_4\left(\frac{Q}{r^2}\rhoFc\right)=r^{-3}\dk^{\leq 1} \Gag$ and $ \nab_4(\chih\c\chibh)=r^{-1}\dk^{\le  1}( \Ga_b \c \Ga_g)$ and
combining with the equation for $\nab_4\rhoc$ in Proposition \ref{Prop.NullStr+Bianchi-lastslice}  we obtain
\begin{equation*}
  \begin{split}
    \nab_4\widehat{\rho}
    ={}& -\frac{3}{r}\rhoc+ \div \b
          + r^{-3}\dk^{\leq 1} \Gag
         +r^{-1}\dk^{\le  1}( \Ga_b \c \Ga_g) 
    \\
    ={}&-\frac{3}{r}\widehat{\rho}+ \bm{C}
          + r^{-3}\dk^{\leq 1} \Gag
         +r^{-1}\dk^{\le  1}( \Ga_b \c \Ga_g).
  \end{split}  
\end{equation*}
Since $\nu=e_3+b_*e_4$ and $b_*=-1-\frac{2M}{r}+\frac{Q^2}{r^2}+r\Ga_b$, we infer
\begin{align*}
  \nab_\nu\widehat{\rho}
       ={}&\frac{6}{r}\widehat{\rho} -\underline{\bm{C}}
         -(1+O(r^{-1}))\bm{C}  +r^{-3}\dk^{\leq 1}\Gab 
         +r^{-1}\dkb^{\le  1}( \Ga_b \c \Ga_b),
\end{align*}
as stated.

\paragraph{Transport equation of \texorpdfstring{$\bm{C}$}{} and \texorpdfstring{$\bm{J}$}{}.} 

 Taking the divergence and the curl of the equation for $\nabla_3\b$ of Proposition \ref{Prop.NullStr+Bianchi-lastslice}, we infer
 \begin{align*}
   \nab_3\div\b -\frac{2\Up}{r}\div\b
   ={}&  [\nab_3,\div]\b
        +\lap\rho
        +\left(\frac{2M}{r^2} -\frac{2Q^2}{r^3}\right)\div\b
        -\left(\frac{6M}{r^3} -\frac{6Q^2}{r^4}\right)\div\eta\\
      & + \frac{Q}{r^2}\lap\rhoF
        + \frac{\Upsilon Q}{r^3}\div\bF
        - \frac{2Q}{r^3}\div\bbF
        +r^{-2}\dkb^{\leq 1}(\Ga_b\c\Ga_g),\\
   \nab_3\curl\b -\frac{2\Up}{r}\curl\b
   ={}&  [\nab_3,\curl]\b
        -\lap\rhod 
        +\left(\frac{2M}{r^2} -\frac{2Q^2}{r^3}\right)\curl\b
        -\left(\frac{6M}{r^3} -\frac{6Q^2}{r^4}\right)\curl\eta\\
      & + \frac{Q}{r^2}\lap\rhodF
       + \frac{\Upsilon Q}{r^3}\curl\bF
        - \frac{2Q}{r^3}\curl\bbF
        +r^{-2}\dkb^{\leq 1}(\Ga_b\c\Ga_g).
 \end{align*}
According to Lemma \ref{Lemma:Commutation-Si_*}, using  the equation for $\nab_3\b$, 
\begin{align*}
  \,[\nab_3,\div]\b &= \frac{\Up}{r} \div \b +\Ga_b\c  \nab_3\b +r^{-2}\dkb^{\le 1} ( \Ga_g \c \Ga_b)  =    \frac{\Up}{r} \div \b      +r^{-2}\dkb^{\le 1} ( \Ga_g \c \Ga_b)  \\
  \,  [\nab_3,\curl ] \b &= \frac{\Up}{r} \curl \b +\Ga_b\c \nab_3\b +r^{-2}\dkb^{\le 1}( \Ga_g \c \Ga_b)
  = \frac{\Up}{r} \curl \b    +r^{-2}\dkb^{\le 1} ( \Ga_g \c \Ga_b).
\end{align*} 
Hence,
\begin{align*}
  \nab_3\div\b
  ={}&  \left(\frac{3}{r}+O(r^{-2})\right)\div \b
       +\lap\rhoc
       -\left(\frac{6M}{r^3} -\frac{6Q^2}{r^4}\right)\div\eta\\
     & + \frac{Q}{r^2}\lap\rhoFc
       + \frac{ Q}{r^3}\div\bF
       - \frac{2Q}{r^3}\div\bbF     +r^{-5}\dkb^{\leq 1} \Gag
       +r^{-2}\dkb^{\leq 1}(\Ga_b\c\Ga_g),\\
  \nab_3\curl\b
  ={}&  \left(\frac{3}{r}+O(r^{-2}) \right)\curl \b
       -\lap\rhod
       -\frac{6M}{r^3} \rhod\\
     &+ \frac{Q}{r^2}\lap\rhodF
       + \frac{ Q}{r^3}\curl\bF
       - \frac{2Q}{r^3}\curl\bbF+r^{-5} \dkb^{\leq 1}\Gag
       +r^{-2}\dkb^{\leq 1}(\Ga_b\c\Ga_g),
\end{align*} 
where we used that  $\curl\eta = \rhod+\Ga_b\c \Ga_g$.

Taking the divergence and the curl of the equation for $\nabla_3\bF$ of Proposition \ref{Prop.NullStr+Bianchi-lastslice}, we infer 
\begin{equation*}
  \begin{split}
    \nabla_3\div \bF - \bigtriangleup \rhoF 
    ={}&  [\nabla_3,\div]\bF
         + \left(\frac{\Upsilon}{r} + \frac{2M}{r^2}-\frac{2Q^2}{r^3}\right)\div\bF
         + \frac{2Q}{r^2}\div\eta
         + r^{-1}\dkb^{\le 1}\left(\Gamma_b\cdot\Gamma_g\right),\\
    \nabla_3\curl\bF - \bigtriangleup\dual\rhoF
    ={}& [\nabla_3,\curl]\bF
         + \left(\frac{\Upsilon}{r} + \frac{2M}{r^2}-\frac{2Q^2}{r^3}\right)\curl\bF
         + \frac{2Q}{r^2}\curl\eta
         + r^{-1}\dkb^{\le 1}\left( \Gamma_b\cdot\Gamma_g \right),
  \end{split}
\end{equation*}
Then recall from Lemma \ref{Lemma:Commutation-Si_*} that
\begin{equation*}
  \begin{split}
    [\nabla_3, \div]\bF
    ={}& \frac{\Upsilon}{r}\div\bF
         + \Gamma_g\cdot\nabla_3\bF
         + r^{-1}\Gamma_b\cdot\dk^{\le 1}\bF= \frac{\Upsilon}{r}\div\bF
         + r^{-1}\dkb^{\le 1}(\Gab\c\Gag),\\
    [\nabla_3,\curl]\bF
    ={}& \frac{\Upsilon}{r}\curl\bF
         + r^{-1}\dkb^{\le 1}(\Gab\c\Gag).
  \end{split}  
\end{equation*}
Thus, we have that
\begin{equation*}
  \begin{split}
    \nabla_3\div \bF - \bigtriangleup \rhoFc 
    ={}& \frac{2}{r} \div\bF
         + \frac{2Q}{r^2}\div\eta+r^{-3} \dkb^{\leq1}\Gag
         + r^{-1}\dkb^{\le1}\left( \Gamma_b\cdot\Gamma_g \right),\\
    \nabla_3\curl\bF - \bigtriangleup\dual\rhoF
    ={}& \frac{2}{r} \curl\bF
        +r^{-3} \dkb^{\leq1}\Gag
         + r^{-1}\dkb^{\le 1}\left( \Gamma_b\cdot\Gamma_g \right).
  \end{split}
\end{equation*}
where we used that $\curl \eta = \dual\rho + \Gamma_b\cdot\Gamma_g=r^{-1}\Gag +  \Gamma_b\cdot\Gamma_g$.
In particular, we deduce
\begin{align*}
  \nab_3\left(\frac{Q}{r^2}\div \bF\right)
  &= \frac{4Q}{r^3}\div\bF+ \frac{Q}{r^2}\left(\bigtriangleup \rhoFc 
    + \frac{2Q}{r^2}\div\eta \right)
    +r^{-5} \dkb^{\leq1}\Gag+ r^{-3}\dkb^{\le1}\left( \Gamma_b\cdot\Gamma_g \right), \\
  \nab_3\left(\frac{Q}{r^2}\curl \bF\right)
  &= \frac{4Q}{r^3}\curl\bF+ \frac{Q}{r^2}\bigtriangleup \dual\rhoF +r^{-5} \dkb^{\leq1}\Gag+ r^{-3}\dkb^{\le1}\left( \Gamma_b\cdot\Gamma_g \right).
\end{align*}
On the other hand, from the equations for $\nab_3\rhoFc$ and $\nab_3\dual\rhoF$ of Proposition \ref{Prop.NullStr+Bianchi-lastslice}, we infer
\begin{align*}
    \nabla_{3}\left(\frac{2Q}{r^3}\rhoFc\right)
  ={}& -\frac{6Q}{r^4}\rhoFc(-\Up +\yc)
       + \frac{2Q}{r^3}\left(-\div \bbF+ \frac{2\Upsilon}{r}\rhoFc
         - \frac{Q}{r^2}\kabc           
         + \frac{2Q}{r^3}\yc
         + \Gamma_b\cdot\Gamma_b\right)\\
                    ={}& \frac{10Q}{r^4}\rhoFc
       - \frac{2Q}{r^3}\div \bbF + r^{-5}\Gag  + r^{-3}\Gamma_b\cdot\Gamma_b, \\
    \nabla_{3}\left(\frac{2Q}{r^3}\dual\rhoF\right)
  ={}& -\frac{6Q}{r^4}\dual\rhoF(-\Up +\yc)
       + \frac{2Q}{r^3}\left(\curl \bbF+ \frac{2\Upsilon}{r}\dual\rhoF
         + \Gamma_b\cdot\Gamma_b\right)\\
                     ={}& \frac{10Q}{r^4}\dual\rhoF
       + \frac{2Q}{r^3}\curl \bbF  + r^{-5}\Gag   + r^{-3}\Gamma_b\cdot\Gamma_b.
\end{align*}
Thus, for $\bm{C}=\div\b -\frac{Q}{r^2}\div \bF+\frac{2Q}{r^3}\rhoFc$ and $\bm{J}=\curl\b -\frac{Q}{r^2}\curl \bF+\frac{2Q}{r^3}\dual\rhoF$, we have
\begin{align*}
    \nab_3\bm{C}&=\big(\frac{3}{r}+O(r^{-2})\big)\div \b
       +\lap\rhoc
       -\left(\frac{6M}{r^3} -\frac{6Q^2}{r^4}\right)\div\eta + \frac{Q}{r^2}\lap\rhoFc
       + \frac{ Q}{r^3}\div\bF
       - \frac{2Q}{r^3}\div\bbF    \\
       & -\frac{4Q}{r^3}\div\bF- \frac{Q}{r^2}\big(\bigtriangleup \rhoFc 
         + \frac{2Q}{r^2}\div\eta \big)+\frac{10Q}{r^4}\rhoFc
       - \frac{2Q}{r^3}\div \bbF\\
       & +r^{-5}\dkb^{\leq 1} \Gag
       +r^{-2}\dkb^{\leq 1}(\Ga_b\c\Ga_g)\\
       &=\big(\frac{3}{r}+O(r^{-2})\big)\div \b- \frac{ 3Q}{r^3}\div\bF - \frac{4Q}{r^3}\div \bbF 
       +\bigtriangleup\rhoc  +\frac{10Q}{r^4}\rhoFc   -\left(\frac{6M}{r^3} -\frac{4Q^2}{r^4}\right)\div\eta\\
       &+r^{-5}\dkb^{\leq 1} \Gag
       +r^{-2}\dkb^{\leq 1}(\Ga_b\c\Ga_g) .
\end{align*}
and
\begin{align*}
    \nab_3\bm{J}&=\big(\frac{3}{r}+O(r^{-2}) \big)\curl \b
       -\lap\rhod
       -\frac{6M}{r^3} \rhod+ \frac{Q}{r^2}\lap\rhodF
       + \frac{ Q}{r^3}\curl\bF
       - \frac{2Q}{r^3}\curl\bbF\\
       &-\frac{4Q}{r^3}\curl\bF- \frac{Q}{r^2}\bigtriangleup \dual \rhoF+\frac{10Q}{r^4}\dual\rhoF
       + \frac{2Q}{r^3}\curl \bbF +r^{-5} \dkb^{\leq 1}\Gag
       +r^{-2}\dkb^{\leq 1}(\Ga_b\c\Ga_g)\\
       &=\big(\frac{3}{r}+O(r^{-2})\big)\curl \b- \frac{ 3Q}{r^3}\curl\bF 
       -\bigtriangleup\dual\rho  +\frac{10Q}{r^4}\dual\rhoF  -\frac{6M}{r^3}\dual\rho+ r^{-5} \dk^{\leq 1}\Gag+r^{-2}\dkb^{\leq 1}(\Ga_b\c\Ga_g).
\end{align*}
The above can be written as 
\begin{align}
       \nab_3\bm{C}       ={}&\big(\frac{3}{r}+O(r^{-2})\big)\bm{C}- \frac{4Q}{r^3}\div \bbF 
       +\bigtriangleup\widehat{\rho}+\frac{2Q}{r^2}\big(\bigtriangleup +\frac{2}{r^2}\big)\rhoFc  -\left(\frac{6M}{r^3} -\frac{4Q^2}{r^4}\right)\div\eta \nonumber\\
       &+r^{-5}\dkb^{\leq 1} \Gag
       +r^{-2}\dkb^{\leq 1}(\Ga_b\c\Ga_g)  ,  \label{eq:nab-3-C}\\
           \nab_3\bm{J}             ={}&\big(\frac{3}{r}+O(r^{-2})\big)\bm{J} 
       -\bigtriangleup\dual\widehat{\rho}-\frac{6M}{r^3} \dual\widehat{\rho}+\frac{2Q}{r^2}\big(\bigtriangleup +\frac{2}{r^2}\big)\dual\rhoF\nonumber \\
       & + r^{-5} \dk^{\leq 1}\Gag+r^{-2}\dkb^{\leq 1}(\Ga_b\c\Ga_g),\label{eq:nab-3-J}
\end{align}
where we also wrote $\widehat{\rho}=\rhoc -\frac{2Q}{r^2}\rhoFc - \frac{1}{2}\chih\c\chibh$ and $      \dual\widehat{\rho}
=\dual \rho + \frac{2Q}{r^2}\dual\rhoF -\frac{1}{2}\hch\wedge\hchb$. 

We now compute the $\nabla_4$ equations.
 Taking the divergence and the curl of the equation for $\nabla_4\b$ of Proposition \ref{Prop.NullStr+Bianchi-lastslice}, we infer
\begin{align*}
  \nab_4\div\beta +\frac{4}{r}\div\beta
  ={}& [\nab_4,\div]\beta  -\div\div\a
       + \frac{Q}{r^2}(\nab_4\div\bF-[\nab_4, \div]\bF)\\
       &
      +r^{-1}\dkb^{\leq 1}(\Gag\c (\a,\b, r^{-1}\bF)),\\
  \nab_4\curl\beta +\frac{4}{r}\curl\beta
  ={}& [\nab_4,\curl]\beta  -\curl\div\a
       + \frac{Q}{r^2}(\nab_4\curl\bF-[\nab_4, \curl]\bF)
      \\
      &+r^{-1}\dkb^{\leq 1}(\Gag\c (\a,\b, r^{-1}\bF)).
\end{align*}
According to Lemma \ref{Lemma:Commutation-Si_*},
\begin{align*}
  [\nab_4,\div]\b  &=    -\frac{1}{r} \div \b      +r^{-2}\dkb^{\le 1} ( \Ga_g \c \Ga_g),  \\
  [\nab_4,\curl ] \b &= -\frac{1}{r} \curl \b     +r^{-2}\dkb^{\le 1} ( \Ga_g \c \Ga_g),\\
    [\nabla_4,\div]\bF &= -\frac{1}{r}\div \bF + r^{-1}\dkb^{\le 1}\left(\Gamma_g\cdot\Gamma_g\right),\\
    [\nabla_4,\curl]\bF &= -\frac{1}{r}\curl \bF + r^{-1}\dkb^{\le 1}\left(\Gamma_g\cdot\Gamma_g\right).
\end{align*}
Hence,
\begin{align*}
  \nab_4\div\beta
  ={}& -\frac{5}{r}\div\beta  -\div\div\a
       + \frac{Q}{r^2}(\nab_4\div\bF+\frac{1}{r}\div \bF )
       +r^{-1}\dkb^{\leq 1}(\Gag\c (\a,\b, r^{-1}\bF)),\\
  \nab_4\curl\beta
  ={}& -\frac{5}{r}\curl\beta  -\curl\div\a
+\frac{Q}{r^2}(\nab_4\curl\bF+\frac{1}{r}\curl \bF )
      +r^{-1}\dkb^{\leq 1}(\Gag\c (\a,\b, r^{-1}\bF)).
\end{align*}
On the other hand, from the equations for $\nab_4\rhoFc$ and $\nab_4\dual\rhoF$ of Proposition \ref{Prop.NullStr+Bianchi-lastslice}, we infer
\begin{align*}
    \nabla_{4}\left(\frac{2Q}{r^3}\rhoFc\right)
       ={}& -\frac{10Q}{r^4}\rhoFc
       + \frac{2Q}{r^3}\div \bF, \\
    \nabla_{4}\left(\frac{2Q}{r^3}\dual\rhoF\right)
       ={}& -\frac{10Q}{r^4}\dual\rhoF
       + \frac{2Q}{r^3}\curl \bF.
\end{align*}
Thus we have
\begin{align*}
   \nab_4 \bm{C}&=-\frac{5}{r}\div\beta  -\div\div\a
       + \frac{Q}{r^2}(\nab_4\div\bF+\frac{1}{r}\div \bF )
       +r^{-1}\dkb^{\leq 1}(\Gag\c (\a,\b, r^{-1}\bF))\\
       &+ \frac{2Q}{r^3}\div \bF-\frac{Q}{r^2}\nab_4\div \bF-\frac{10Q}{r^4}\rhoFc
       + \frac{2Q}{r^3}\div \bF\\
                      &= -\frac{5}{r}\div\beta 
       + \frac{5Q}{r^3}\div\bF    -\frac{10Q}{r^4}\rhoFc 
       -\div\div\a +r^{-1}\dkb^{\leq 1}(\Gag\c (\a,\b, r^{-1}\bF)),
\end{align*}
and 
\begin{align*}
    \nab_4 \bm{J}&=  -\frac{5}{r}\curl\beta  -\curl\div\a
+\frac{Q}{r^2}(\nab_4\curl\bF+\frac{1}{r}\curl \bF )
       +r^{-1}\dkb^{\leq 1}(\Gag\c (\a,\b, r^{-1}\bF))
       \\
       &+ \frac{2Q}{r^3}\curl \bF-\frac{Q}{r^2}\nab_4\curl \bF-\frac{10Q}{r^4}\dual\rhoF
       + \frac{2Q}{r^3}\curl \bF\\
               &= -\frac{5}{r}\curl\beta 
       + \frac{5Q}{r^3}\curl\bF    -\frac{10Q}{r^4}\dual\rhoF
       -\curl\div\a +r^{-1}\dkb^{\leq 1}(\Gag\c (\a,\b, r^{-1}\bF)).
\end{align*}

The above can be written as 
\begin{align}
       \nab_4 \bm{C} &= -\frac{5}{r}\bm{C}    
       -\div\div\a +r^{-1}\dkb^{\leq 1}(\Gag\c (\a,\b, r^{-1}\bF)), \label{eq:nab-4-C}\\ 
         \nab_4 \bm{J}        &= -\frac{5}{r}\bm{J}    -\curl\div\a +r^{-1}\dkb^{\leq 1}(\Gag\c (\a,\b, r^{-1}\bF)). \label{eq:nab-4-J}
\end{align}
Since $\nu=e_3+b_*e_4$ and $b_*=-1-\frac{2M}{r}+\frac{Q^2}{r^2}+r\Ga_b$, combining \eqref{eq:nab-3-C}-\eqref{eq:nab-3-J} with \eqref{eq:nab-4-C}-\eqref{eq:nab-4-J} we have
\begin{align*}
    \nabla_{\nu}\bm{C}
  ={}& \big(\frac{8}{r}+O(r^{-2})\big)\bm{C}- \frac{4Q}{r^3}\div \bbF 
       +\bigtriangleup\widehat{\rho}+\frac{2Q}{r^2}\big(\bigtriangleup +\frac{2}{r^2}\big)\rhoFc+O(1+O(r^{-1}))\div\div\a \\
       & -\left(\frac{6M}{r^3} -\frac{4Q^2}{r^4}\right)\div\eta \nonumber+r^{-5}\dkb^{\leq 1} \Gag
       +r^{-2}\dkb^{\leq 1}(\Ga_b\c\Ga_g) ,
\end{align*}
and
\begin{align*}
    \nabla_\nu\bm{J}&= \big(\frac{8}{r}+O(r^{-2})\big)\bm{J}  -\bigtriangleup\dual\widehat{\rho}-\frac{6M}{r^3} \dual\widehat{\rho} +\frac{2Q}{r^2}\big(\bigtriangleup +\frac{2}{r^2}\big)\dual\rhoF +O(1+O(r^{-1}))\curl\div\a \\
    & +r^{-5}\dk^{\le 1}\Gamma_g +r^{-2}\dkb^{\leq 1}(\Ga_b\c\Ga_g),
\end{align*}
as stated.

\printbibliography

\end{document}
